\documentclass[3p]{elsarticle}

\journal{arXiv}

\usepackage{amsfonts}
\usepackage{amsmath, amscd, amssymb, bm, amsbsy, epsf, amsthm}
\usepackage{enumerate}
\usepackage{paralist}
\usepackage{graphicx,color}
\usepackage{subfigure}
\usepackage{mathrsfs}
\usepackage{verbatim}
\usepackage{inputenc}
\usepackage{diagbox}
\usepackage{algorithm}
\usepackage{algorithmicx}
\usepackage{algpseudocode}
\usepackage{hyperref}
\usepackage{pifont}
\usepackage{inputenc}
\usepackage{bbm, dsfont}
\usepackage{multicol}
\usepackage{balance}
\usepackage{verbatim, booktabs}

\usepackage{multirow, url}
\usepackage[table]{xcolor}
\usepackage{diagbox}
\usepackage{tikz}
\usepackage{capt-of}
\usepackage{natbib}

\usepackage{cmbright}

\DeclareMathAlphabet{\mathpzc}{OT1}{pzc}{m}{it}
\DeclareMathAlphabet{\mathpzc}{OT1}{pzc}{m}{it}

\newtheorem{theorem}{Theorem}
\newtheorem{lemma}[theorem]{Lemma}
\newtheorem{corollary}[theorem]{Corollary}
\newtheorem{proposition}[theorem]{Proposition}
\newdefinition{definition}{Definition}
\newdefinition{hypothesis}{Hypothesis}
\newdefinition{problem}{Problem}
\newdefinition{remark}{Remark}
\newdefinition{example}{Example}

\def\N{\boldsymbol{\mathbbm{N}}}
\def\R{\boldsymbol{\mathbbm{R}}}
\def\defining{\overset{\mathbf{def}}=}

\def\A{\mathcal{A}}
\def\S{\mathcal{S}}

\def\Exp{\boldsymbol{\mathbbm{E}}}
\def\prob{\boldsymbol{\mathbbm{P} } }
\def\Var{\boldsymbol{\mathbbm{V}\!\mathrm{ar}}}
\def\ind{\boldsymbol{\mathbbm{1} } }
\def\x{\mathbf{x} }

\def\T{\mathcal{T}}
\def\p{\mathbf{p}}
\def\g{\mathbf{g}}

\def\ind{\boldsymbol{\mathbbm{1}}}
\DeclareMathAlphabet{\mathpzc}{OT1}{pzc}{m}{it}

\newcommand\X{\mathop{}\!\mathbf{X} } 
\newcommand\Y{\mathop{}\!\mathbf{Y} } 
\newcommand\Z{\mathop{}\!\mathbf{Z} } 
\newcommand\ups{\mathop{}\!\mathbf{S} } 
\newcommand\upt{\mathop{}\!\mathbf{T} } 
\newcommand\upk{\mathop{}\!\mathbf{K} } 

\newcommand\fg{\mathop{}\!\mathrm{fG} } 
\newcommand\gr{\mathop{}\!\mathrm{G} } 
\newcommand\eg{\mathop{}\!\mathrm{eG} } 
\newcommand\ff{\mathop{}\!\mathrm{fF} } 
\newcommand\ef{\mathop{}\!\mathrm{eF} } 
\newcommand\lp{\mathop{}\!\mathrm{LP} } 

\newcommand\upg{\mathop{}\!\mathbf{G} } 
\newcommand{\upp}{\mathop{}\!\mathbf{P}}
\newcommand\upw{\mathop{}\!\mathbf{W}} 

\newcommand\mar{\mathop{}\!\textbf{m} } 

\newcommand\tree{\mathcal{T}}
\newcommand\lt{\mathop{}\!\mathrm{left} } 
\newcommand\rt{\mathop{}\!\mathrm{right} } 
\newcommand\lb{\mathop{}\!\mathrm{lb} } 
\newcommand\odd{\mathop{}\!\mathrm{odd} } 
\newcommand\even{\mathop{}\!\mathrm{even} } 
\newcommand\alg{\mathop{}\!\mathrm{alg} } 
\newcommand\side{\mathop{}\!\mathrm{side} } 

\def\cleft{\mathbf{C}_{\lt}}
\def\cright{\mathbf{C}_{\rt}}
\def\sleft{\boldsymbol{\ups}_{\lt}}
\def\sright{\boldsymbol{\ups}_{\rt}}
\def\kleft{\boldsymbol{\upk}_{\lt}}
\def\kright{\boldsymbol{\upk}_{\rt}}

\def\expcl{\boldsymbol{\mathbbm{E}}(\cleft)}
\def\expcr{\boldsymbol{\mathbbm{E}}(\cright)}

\def\expk{\boldsymbol{\mathbbm{E}}(\upk)}
\def\exps{\boldsymbol{\mathbbm{E}}(\ups)}

\def\expkl{\boldsymbol{\mathbbm{E}}(\kleft)}
\def\expsl{\boldsymbol{\mathbbm{E}}(\sleft)}
\def\expkr{\boldsymbol{\mathbbm{E}}(\kright)}
\def\expsr{\boldsymbol{\mathbbm{E}}(\sright)}
\def\expcl{\boldsymbol{\mathbbm{E}}(\cleft)}
\def\expcr{\boldsymbol{\mathbbm{E}}(\cright)}

\begin{document}

\begin{frontmatter}

\title{Expected Performance and Worst Case Scenario Analysis \\
of the Divide-and-Conquer Method \\
	for the 0-1 Knapsack Problem}
\tnotetext[mytitlenote]{This material is based upon work supported by project HERMES 45713 from Universidad Nacional de Colombia,
Sede Medell\'in.}

\author[mymainaddress]{Fernando A Morales} 
\cortext[mycorrespondingauthor]{Corresponding Author}
\ead{famoralesj@unal.edu.co}

\author[mysecondaryaddress]{Jairo A Mart\'inez}



\address[mymainaddress]{Escuela de Matem\'aticas
Universidad Nacional de Colombia, Sede Medell\'in \\
Carrera 65 \# 59A--110, Bloque 43, of 106,
Medell\'in - Colombia\corref{mycorrespondingauthor}}

\address[mysecondaryaddress]{Departamento de Ciencias Matem\'aticas, Universidad EAFIT. \\
Carrera 49 \# 7 Sur-50, Bloque 38, of 501, Medell\'in - Colombia}


\begin{abstract}
In this paper we furnish quality certificates for the Divide-and-Conquer method solving the 0-1 Knapsack Problem: the worst case scenario and an estimate for the expected performance. The probabilistic setting is given and the main random variables are defined for the analysis of the expected performance. The performance is accurately approximated for one iteration of the method then, these values are used to derive analytic estimates for the performance of a general Divide-and-Conquer tree. Most of the theoretical results are verified vs numerical experiments for a  wider illustration of the method.
\end{abstract}


\end{frontmatter}



%
%
%
%
%
%
%
%
%
\section{Introduction}   
%
%
%
%
The 0-1 knapsack problem (0-1KP) is one of the most widely discussed problems in the combinatorial optimization literature and it is certainly the simplest prototype of a maximization problem \cite{kellerer2005knapsack}. It is defined as follows: given a set of $n$ items, each item $j$ having a weight $w(j)$ and a profit $p(j)$, the problem is to choose a subset of items such that the sum of profits is maximized, while the sum of weights does not exceed the knapsack capacity $\delta$. The simplicity of its formulation (see \textsc{Problem} \ref{Pblm Original 0-1 KP}) contrasts with its surprising theoretical and practical relevance:  its decision version is one of Karp's $21$ NP-complete problems \cite{karp1972reducibility}, 0-1KP itself, or some of its well-known variants, is used in the modeling of important practical problems such as portfolio management and container optimization \cite{vaezi2019portfolio,jacko2016resource}. In addition, it appears as a subproblem when applying some decomposition technique to large problems, for example, solving material cutting models using a column generation method \cite{blado2020column,muter2018algorithms}. It has also played an interesting role in the development of cryptographic systems \cite{smart2016cryptography}. 

The Divide-and-Conquer method for solving the 0-1KP was recently introduced by Morales and Mart\'inez in \cite{MoralesMartinez}. The method seeks to reduce the computational complexity of a large instance of the problem, by executing a recursive subdivision into smaller instances, so that the process can be visualized as the construction of a binary tree whose nodes are knapsack subproblems. As it was emphasized in the original work, the method does not \emph{compete} with the existing algorithms, it \emph{complements} them (observe that in Example \ref{Exm 0-1KP and DC tree Asymetric}, it is not specified how to solve the defined subproblems). The experimental results presented in \cite{MoralesMartinez}, show that the method is a good \emph{middle grounds} alternative, halfway between computational complexity and quality of the solution. So far, the quality performance of Divide-and-Conquer has been measured only empirically. The aim of this paper is to analyze theoretically its quality performance from two points of view: the worst-case scenario and its expected/average performance.
%
%
\subsection{Literature Review}
%
%
In the last three decades of the 20th century, the algorithms implemented for the resolution of 0-1KP reached a great maturity, standing out the primal and dual variants of  branch and bound, \cite{kolesar1967branch,horowitz1974computing,pisinger1995expanding}, dynamic programming \cite{pisinger1999linear,pisinger2003dynamic}, the \textit{core}-type algorithms \cite{balas1980algorithm,martello1988new}  and hybrid procedures like the \textit{Combo} algorithm \cite{martello1999dynamic}. Although in terms of worst case time complexity, the best bounds achieved are pseudpolynomial, the combined application of different techniques made it possible to effectively solve a large number of benchmark instances, which led to the designation of the knapsack problem as one of the ``easy to solve'' NP-hard problems. Consequently, the research line directed at understanding the characteristics of the most computationally challenging instances \cite{pisinger2005hard}, was developed.

The discrepancy observed between the good performance of simple heuristics and exact methods when applied on pure random instances, and the high complexity pointed out by the worst case analysis, started to be explained theoretically through  probabilistic analysis. In this respect, Kellerer et al. \cite{kellerer2005knapsack}  classified the contributions depending on whether the results are: $ \bullet $ \textit{structural}, if they give a probabilistic statement e.g. on the optimal solution value. $ \bullet $ Expected performance of algorithms, which produce an optimal solution with a certain probability. $ \bullet $ Expected running time of algorithms, which always produce solutions of a certain quality.

A probabilistic model for the knapsack problem widely used in the literature is the one proposed by Lueker  \cite{lueker1982average},  in which it is assumed that weights and profits are  uniformly selected from the interval $[0,1]$, so that the choice of the parameters of the $n$ items  can be understood as the random location of $n$ points in the unit square. The knapsack capacity should be specified as $\delta=\beta n$, where $\beta$ is some constant in the interval $(0,1]$. Several random models in literature differ from Lueker's proposal only in the $\beta$ parameter, see for example \cite{calvin2003average,diubin2008greedy,frieze1984approximation}.

The structural result presented by Lueker \cite{lueker1982average} consisted in estimating the expected value of the linear relaxation gap, to formally explain the empirically observed good performance of the branch and bound algorithms (B\&B). The fundamental result was that integrality gap has order $ \mathcal{O}(\log^2 (n)/n) $ which means it decreases with problem size increase. Regarding the analysis of the exact solution, we highlight two works: Frieze and Clarke \cite{frieze1984approximation}  conducted a probabilistic analysis of 0-1KP and obtained an interesting bound for the behavior of the  objective value, showing that it is asymptotically equal to $\sqrt{2n/3}$ with probability going to $1$ as $n$ tends to infinity. Mamer et al.  \cite{mamer1990growth} carried out a similar analysis for a very large class of joint distributions and deducted the same upper bound as Frieze and Clarke.

Since the 80's, of the last century the probabilistic method was applied to the study of different versions of the greedy algorithm. Szkatula \& Libura \cite{szkatula1983probabilistic}  obtained moments and distribution functions for some parameters of the greedy algorithm without ordering, obtaining recursive equations for the distribution function of the accumulated weight in any iteration. Under slightly different hypotheses than those in the standard model, Calvin and Leung \cite{calvin2003average}  proved, using convergence in distribution,  that the sorted greedy algorithm produces results that differ from the optimum value by order $1/\sqrt{n}$.  Diubin et al. \cite{diubin2008greedy} address the analysis of the minimization version of the 0-1KP and  proved that the primal and dual greedy methods  for the minimization knapsack problem are also asymptotically good. They showed that despite the complementarity between the minimization and maximization problems, the result concerning the former cannot be obtained from the result addressing the latter problem. It is worth noting that most of the mathematical analyses involved in these investigations exploit the geometric interpretation of the extended greedy algorithm, in particular the critical or splitting ray.

There are also very relevant works related to the expected running time of exact and approximation algorithms. Beier and V{\"o}cking \cite{beier2004probabilistic} presented the first average-case analysis proving a polynomial upper bound on the expected running time of a sparse dynamic programming algorithm for the 0-1KP; originally proposed  by Nemhauser and Ullman \cite{nemhauser1969discrete}. The algorithm iteratively extend non-dominated or Pareto-efficient subset, contained in the set of the first $i$ items. The main conclusion is that the number of Pareto-efficient knapsack fillings is polynomially bounded in the number of available items. The random input model used in this study is more general, the weights of the items are chosen by an adversary and their profits are chosen according to arbitrary continuous or discrete probability distributions with finite mean, allowing to address the effects of correlation between parameters. It is interesting to point out that when using discrete distributions, they were able to prove a trade-off, ranging from polynomial to pseudo-polynomial running time, depending on the randomness of the specified instances.

In a later work Beier and V{\"o}cking  \cite{beier2004random} studied the average-case performance of \textit{core algorithms} for the 0-1KP. They proved an upper bound of $\mathcal{O}(n \text{ polylog}(n))$ on the expected running time of a core algorithm on instances with $n$ items, whose profits and weights are  drawn random and independently from a uniform distribution.  Unlike previous works such as Goldberg and \& Marchetti-Spaccamela \cite{goldberg1984finding}, the degree of the polynomial involved is relatively low, but the probabilistic analysis is complicated due to the dependence between random variables.

More recent research has attempted to theoretically understand, the efficiency  of simple and successful heuristics such as rollout algorithms. These iterative methods use a \emph{base policy}, whose performance is evaluated to obtain an \emph{improved policy}, by one-step look ahead. Rollout algorithms are  easy to implement and guarantee a \emph{not worse}, and usually much better results than corresponding base policies. Bertazzi \cite{bertazzi2012minimum}  proved   minimum and worst case performance ratio when the greedy, full greedy and the extended greedy algorithms are chosen as base policies, respectively. In all cases the analysis was applied to only the first iteration, showing furthermore, that for the algorithms considered there exists an instance in which the worst-case performance ratio is obtained at the first iteration, so that the expected value deducted cannot be subsequently improved. The worst case performance ratios was improved from $0$ to $1/2$ for the greedy algorithm, and from $1/2$ to $2/3$ for the extended greedy algorithm. Motivated by Bertazzi's results, Mastin \& Jaillet \cite{mastin2015average}  provided a complementary study of rollout algorithms for knapsack-type problems from an average-case perspective.  The authors started from Lueker's random model with profits and weights taken at random and independently generated. They analyzed the exhaustive rollout and consecutive rollout techniques, both using as base policy the unsorted greedy algorithms. The authors derived bounds for both techniques, showing that the expected performance of the rollout algorithms is strictly better than the performance obtained by only using the base policy. These results hold after only a single iteration and provide bounds for additional iterations. The authors state that it was not possible to apply the same analysis to a sorted greedy algorithm, due to the dependencies between random variables originated in the ordering step.
%
\subsection{Contributions}
%
First, a worst case performance ratio of $1/2$ is derived for the Divide-and-Conquer heuristics (see \textsc{Theorem} \ref{Thm Worst Case Scenario}), then a probabilistic analysis is presented for the same method. The defined random model (see \textsc{Section} \ref{Sec Probabilistic Model}), differs in several aspects from Lueker's basic model \cite{lueker1982average}, which is the literature's mainstream: $ \bullet $ Discrete uniform probability distributions are assumed for the parameters. $ \bullet $ A very simple relation is defined between the number of items $n$ and the knapsack capacity $\delta$. $ \bullet $ The profits are defined by means of the weights and the efficiencies, which in turn are given in terms of random variables called the increments. We point out that according to the literature review, discrete distributions were considered only in Beier and V{\"o}cking's work  \cite{beier2004probabilistic}. The adopted model allowed to obtain structural results for the (sorted) greedy and the eligible first item algorithm, which are difficult to approach from the usual model (see for example Bertazzi \cite{bertazzi2012minimum}). Similarly to the Mastin \& Jaillet  proof strategy \cite{mastin2015average}, the theoretical analysis of the Divide-and-Conquer method concentrates on its first iteration. Asymptotic relationships are presented, these permit to define and evaluate numerically, the performance ratios for the entire solution process (see \textsc{Theorem} \ref{Thm Split Item Asymptotic}, \textsc{Lemmas} \ref{Thm Distribution Split Item Left Right} and \ref{Thm Knapsack Slack Expectation Left Right} and \textsc{Corollaries} \ref{Thm Approximation Greedy Left Right Expected Profit}, \ref{Thm Approximation Expectation Eligible First Left and Right}).
\section{Preliminaries}\label{Sec Preliminaries}
%
%
%
%
%
%
In this section the general setting and preliminaries of the problem are presented. We start introducing the mathematical notation. For any natural number $ \mu\in \N $, the symbol $ [\mu] \defining \{ 1, 2, \ldots , \mu \} $ indicates the sorted set of the first $ \mu $ natural numbers. In the same fashion $ [0, 6, 1,  3] $ stands for the set containing the mentioned elements in the order $ 0, 6, 1, 3 $. Greek lowercase letters ($ \delta, \lambda, \mu, \nu, \ldots $), are used for important fixed constants. For any set $ E $ we denote by $ \# E $ its cardinal and by $ \wp(E) $ its power set. Given an event $ E \subseteq \Omega $, we denote its indicator function by $ \ind_{E}: \Omega \rightarrow \{0,1\} $, with $\ind_{E}(\omega) = 1 $ if $ \omega \in E $ and zero otherwise. Random variables will be represented with bold capital letters, e.g. $ \X, \Y, \Z, ... $ and its respective expectations with $ \Exp(\X),  \Exp(\Y), \Exp(\Z), ... $. Vectors are indicated with bold letters, namely $ \p, \g, ... $ etc. Particularly important collections of objects will be written with calligraphic characters, e.g. $ \mathcal{A}, \mathcal{D}, \mathcal{E} $ to add emphasis. A particularly important set is $ \mathcal{S}_{N} $, where $ \mathcal{S}_{N} $ denotes the collection of all permutations in $ [N] $. For any real number $ x \in \R $ the floor and ceiling function are given (and denoted) by $ \lfloor x \rfloor \defining \max\{k: \ell\leq x, \, k \text{ integer}\} $, $ \lceil x \rceil \defining \max\{k: k\geq x, \, k \text{ integer}\} $, respectively.
%
%
\subsection{The Problem}\label{Sec The Problem}
%
%
In the current section we introduce the 0-1 Knapsack Problem and review a list of greedy algorithms, to be used in the analysis of the Divide-and-Conquer method for both ends: attain a quality certificate in the worst case scenario and compute the expected performance of the method. 
\begin{problem}[0-1KP]\label{Pblm Original 0-1 KP}
	Consider the problem
	\begin{subequations}\label{Eqn Integer Problem}
		\begin{equation}\label{Eqn Integer Problem Objective Function}
		z^{*} \defining \max \sum\limits_{i\, = \, 1 }^{ \mu } p(i) \, x(i) ,
		\end{equation}
		subject to
		\begin{equation}\label{Eqn Integer Problem Capacity Constraint}
		\sum\limits_{i \, = \, 1 }^{\mu} w(i) \, x(i) \leq \delta,
		\end{equation}
		\begin{align}\label{Eqn Integer Problem Choice Constraint}
		& x(i) \in \{0,1\}, &
		& \text{for all } i \,\in \,[\mu] .
		\end{align}	
	\end{subequations}
	Here, $ \delta $ is the \textbf{knapsack capacity} and $  \big( x(i) \big)_{i = 1}^{\mu} $ is the list of binary valued \textbf{decision variables}. In addition, the \textbf{weight coefficients} $ \big(w(i)\big)_{i = 1}^{\mu} $, as well as the knapsack capacity $ \delta $
	are all positive integers. In the sequel,
	$ z^{*} $ denotes the  objective function \textbf{optimal solution value}. We refer to the parameters $ \big(p(i) \big)_{i = 1}^{\mu} \subseteq (0, \infty]^{\mu} $ as the \textbf{profits} and introduce the \textbf{efficiency} rate $ g(i) \defining \frac{p(i)}{w(i)} $. Finally, in the sequel the problem is indicated by the acronym \textbf{0-1KP} and we denote by $ \Pi = \big\langle \delta, (p(i))_{i \in [\mu]}, (w(i))_{i \in [\mu]} \big\rangle $ one of its instances.	
\end{problem}
Before we continue our analysis, the next hypothesis is adopted.
\begin{hypothesis}\label{Hyp Non Triviality and Sorting}
	In the sequel we assume that the instances $ \Pi $ of the 0-1KP satisfy the following
	\begin{enumerate}[(i)]
		\item The items of \textsc{Problem} \ref{Pblm Original 0-1 KP} are sorted according to their efficiencies in decreasing order i.e., 
		\begin{equation}\label{Ineq Sorting}
		g(1) \geq 
		g(2) \geq 
		\ldots \geq 
		g(\mu) .
		\end{equation}

		\item The weights of the items satisfy 
		\begin{align}\label{Ineq Weights non-triviality conditions}
		& w(i) \leq \delta, \text{ for all } i \in [\mu], &
		& \sum\limits_{i \, = \, 1}^{\mu} w(i) > \delta .
		\end{align}
	\end{enumerate}
\end{hypothesis}
\begin{remark}[0-1KP Setting]\label{Rem Setting of 0-1 KP}
	We make the following observations about the setting of the problem \ref{Pblm Original 0-1 KP}.
	\begin{enumerate}[(i)]
		\item The condition \eqref{Ineq Sorting} in \textsc{Hypothesis} \ref{Hyp Non Triviality and Sorting} is assumed to ease the algorithm analysis later on.
		
		\item The condition \eqref{Ineq Weights non-triviality conditions} in \textsc{Hypothesis} \ref{Hyp Non Triviality and Sorting} guarantees two things. First, every item is eligible to be chosen. Second, the complete set of items is not eligible. Both conditions are introduced to prevent trivial instances of \textsc{Problem} \ref{Pblm Original 0-1 KP}.
		
		\item  Due to the condition \eqref{Ineq Weights non-triviality conditions}, the split item and the greedy algorithm solutions of \textsc{Definition} \ref{Def The split item} are well-defined.
		
	\end{enumerate}
\end{remark}
Next, we recall a catalog of greedy algorithms for the solution of \textsc{Problem} \ref{Pblm Original 0-1 KP}, to be used in the probabilistic analysis of the Divide-and-Conquer method.
\begin{definition}[Greedy Solutions]\label{Def The split item}
	Let $ \Pi = \big\langle \delta, \big(p(i) \big)_{ i = 1 }^{ \mu }, \big(w(i) \big)_{ i = 1 }^{ \mu } \big\rangle $ be an instance of \textsc{Problem} \ref{Pblm Original 0-1 KP}. Let $ \ind_{ \{ J \} }  $ be the indicator function of the singleton $ \{J\} $, with $ J \in \N $. Define the following
	\begin{enumerate}[(i)]
		\item The \textbf{split item} is the index $ s \in [\mu] $ satisfying 
		\begin{align}\label{Eqn The split item}
		& \sum\limits_{ i \, = \, 1 }^{ s - 1 } w(i) \leq \delta, &
		& \sum\limits_{i \, = \, 1}^{s} w(i) > \delta .
		\end{align}
		%
		
		
		\item The \textbf{greedy algorithm solution} to the problem \ref{Pblm Original 0-1 KP} and its corresponding objective function values are given by 
		%
		\begin{align}\label{Eq Greedy Decision Variables}
		& x^{\gr}(i)  \defining \begin{cases}
		1 , & i = 1, \ldots, s - 1 , \\
		0 , & i = s, \ldots, \mu ,
		\end{cases} &
		& z^{\gr} \defining \sum\limits_{i \, = \, 1}^{ s - 1 } p(i) .
		\end{align}
		%
		%
		\item The \textbf{extended-greedy algorithm solution} yields the following objective function value and corresponding solution to the problem \ref{Pblm Original 0-1 KP} 
		%
		\begin{align}\label{Eq Extended Greedy Decision Variables}
		& z^{\eg} \defining \max\big\{ z^{\gr}, \max\limits_{ i \, \in \, [\mu] } \{p(i): i \in [\mu]\} \big\}  , &
		& x^{\eg}(i)  \defining \begin{cases}
		x^{\gr}(i) , &  z^{\eg} = z^{\gr}, \\
		\ind_{ \{J\} }(i) , & z^{\eg} > z^{\gr} .
		\end{cases} 
		\end{align}
		Here, $ J \defining \min\big\{ j\in [\mu]: p(j) = \max\limits_{\ell \, \in \, [\mu]} p(\ell)\big\} $. 
		
		%
		%
		\item The \textbf{eligible-First greedy algorithm solution} defines the following set
		\begin{subequations}\label{Eq Eligible First Algorithm}
			\begin{equation}\label{Eq Eligible First Set}
			E \defining \Big\{ i > s: w(i) \leq \delta - \sum\limits_{ i \, = \, 1 }^{ s - 1 } w(i) \Big\} ,
			\end{equation}
			to yield the following objective function value and corresponding solution to the problem \ref{Pblm Original 0-1 KP} 
			\begin{align}\label{Eq Eligible-First Greedy Decision Variables}
			& z^{\ef} \defining 
			\begin{cases}
			z^{\gr}, & E = \emptyset, \\
			z^{\gr} + z_{J}, & J = \min E ,
			\end{cases} &
			& x^{\ef}(i)  \defining \begin{cases}
			x^{\gr}(i) , &  E = \emptyset, \\
			x^{\gr}(i) + \ind_{ \{J\} }(i) , & J = \min E.
			\end{cases} 
			\end{align}
		\end{subequations}
		%
		%
		\item Finally we describe the \textbf{full-greedy algorithm solution} for solving problem \ref{Pblm Original 0-1 KP} with the following pseudocode 

		\begin{algorithm}[H]
			\caption{Greedy Algorithm, returns feasible solution $ \big(x(i) \big)_{i  = 1}^{\mu} $ and the associated value $ z^{\fg} = \sum_{i = 1}^{\mu}p(i) x(i) $ of the objective function for \textsc{Problem} \ref{Pblm Original 0-1 KP}.}\label{Alg Full Greedy Algorithm}
			\begin{algorithmic}[1]
				\Procedure{Greedy-Algorithm pseudo-code}{Input:
					Capacity: $ \delta $, Profits: $ (p(i))_{i = 1}^{ \mu } $, Weights: $ (w(i))_{i = 1}^{ \mu } $. 
					The items' efficiencies satisfy $ g(1) \geq g(2) \geq \ldots \geq g(\mu) $.}
				\State $ w^{\fg} \defining 0 $ \Comment{$ w^{\fg} $ is the total weight of the currently packed items}
				\State $ z^{\fg} \defining 0 $ \Comment{ $ z^{\fg} $ is the profit of the current solution}
				\For{$ j = 1, \ldots, \mu $}{
					\If{$ w^{\fg} + w(j) \leq  \delta $ }
					\State $ x(j) = 1 $ \Comment{put item $ j $ into the knapsack}
					\State $ w^{\fg} = w^{\fg} + w(j) $
					\State $ z^{\fg} = z^{\fg} + p(j) $
					\Else
					\State $ x(j) = 0 $
					\EndIf
				}
				\EndFor
				\EndProcedure
			\end{algorithmic}
		\end{algorithm}
	\end{enumerate}
\end{definition}
\begin{remark}[Greedy Algorithms]\label{Rem Greedy Algorithms}
	It is direct to see that $ z^{\gr}  \leq \min \{ z^{\ef}, z^{\eg} \} \leq z^{\fg} $ for any instance of 0-1KP and that all the algorithms are of the same order in terms of computational cost. Therefore, only the full-greedy algorithm should be implemented in practice however, it is very hard to analyze from the probabilistic point of view. The extended-greedy algorithm furnishes a quality certificate for the worst case scenario, as it can be seen in \textsc{Theorem} \ref{Thm Greed and LP Solutions} (ii), however its probabilistic performance analysis is as hard as in the previous case. On the other hand, the probabilistic analysis of the greedy algorithm is tractable (see \textsc{Theorem} \ref{Thm Greedy Expected Profit}) and it characterizes the linear programming relaxation of 0-1KP (see \textsc{Theorem} \ref{Thm Greed and LP Solutions} (i)), which contributes to the probabilistic analysis of the latter problem (see \textsc{Theorem} \ref{Thm Linear Relaxation Expected Profit}). Finally, the eligible-first greedy algorithm is introduced because its probabilistic analysis is tractable at the time of furnishing better approximation estimates to the optimal solution, than the greedy algorithm, see \textsc{Section} \ref{Sec Full-First and Eligible-First algorithms' expected performance}.
\end{remark}
\begin{definition}\label{Pblm Natural LOP Problem}
	The natural \textbf{linear programming relaxation} of \textsc{Problem} \ref{Pblm Original 0-1 KP}, is given by 
	\begin{problem}[0-1LPK]\label{Pblm Integer Problem LOP Relaxation}
		\begin{subequations}\label{Eqn Integer Problem LOP Relaxation}
			\begin{equation}\label{Eqn Integer Problem Objective Function LOP}
			\max \sum\limits_{i\, = \, 1}^{ \mu } p(i) \, x(i) ,
			\end{equation}
			subject to
			\begin{equation}\label{Eqn Integer Problem Capacity Constraint LOP}
			\sum\limits_{i\, = \, 1}^{ \mu } w(i) \, x(i) \leq \delta,
			\end{equation}
			\begin{align}\label{Eqn Integer Problem Choice Constraint LOP}
			& 0 \leq x(i) \leq 1, &
			& \text{for all } i \,\in \,[\mu] ,
			\end{align}	
		\end{subequations}
		i.e., the decision variables $ \big( x(i) \big)_{i = 1}^{\mu} $ are are now real-valued. 
	\end{problem}
	In the sequel the acronym \textbf{0-1LPK} will stand for the associated linear relaxation problem.
\end{definition}
We close this section recalling a couple of classical results for the sake of completeness
\begin{theorem}\label{Thm Greed and LP Solutions}
	Let $ \Pi = \big\langle \delta, \big( p(i) \big)_{ i = 1 }^{ \mu }, \big( w(i) \big)_{ i = 1 }^{ \mu } \big\rangle $ be an instance of \textsc{Problem} \ref{Pblm Original 0-1 KP}, then
	\begin{enumerate}[(i)]
		\item The optimal solution of the problem \ref{Pblm Natural LOP Problem} (0-1 LPK) is given by
		\begin{subequations}\label{Eq LR Solution}
			\begin{equation}\label{Eq LP Decision Variables}
			x^{\lp}(i)  = \begin{cases}
			1 , & i = 1, \ldots, s - 1 , \\
			\frac{1}{w(s)} \big( \delta -  \sum\limits_{i = 1}^{s - 1} w(i) \big) , & i = s,  \\
			0 , & i = s + 1, \ldots, \mu ,
			\end{cases} 
			\end{equation}
			with the corresponding objective function value
			\begin{equation}\label{Eq LR Objective}
			z^{\lp} =  \sum\limits_{i \, = \, 1}^{ s - 1 } p(i) + 
			\Big(\delta -  \sum\limits_{j = 1}^{s - 1} w(j) \Big)\frac{p(s)}{w(s)} .
			\end{equation}
		\end{subequations}
		
		\item Let $ z^{*} $, $ z^{\eg} $ be respectively, the optimal and the extended greedy algorithm objective values for \textsc{Problem} \ref{Pblm Original 0-1 KP}. Then, 
		\begin{equation}\label{Ineq Greedy Efficiency}
		\frac{z^{*}}{2} \leq  z^{\eg},
		\end{equation}
		i.e., the extended greedy algorithm has a relative performance quality certificate of $ 50\% $.
		
	\end{enumerate}
\end{theorem}
\begin{proof}
	\begin{enumerate}[(i)]
		\item See 
		\textsc{Theorem} 2.2.1 in \cite{kellerer2005knapsack}. 
		
		\item See \textsc{Theorem} 2.5.4 in \cite{kellerer2005knapsack}. 
	\end{enumerate}
\end{proof}
%
%
%
%
%
%
\subsection{The Divide-and-Conquer Approach}
%
%
%
%
The Divide-and-Conquer method for solving the 0-1KP was introduced in \cite{MoralesMartinez}. Here was presented an extensive discussion (theoretical and empirical) on the possible strategies to implement it and conclude that the best strategy is the one described by the following algorithm
\begin{definition}[Divide-and-Conquer pairs and trees]\label{Def Divide and Conquer Setting}
	Let $ \Pi = \big\langle \delta, \big( p(i) \big)_{ i = 1 }^{ \mu }, \big( w(i) \big)_{ i = 1 }^{ \mu }  \big\rangle $ be an instance of \textsc{Problem} \ref{Pblm Original 0-1 KP}
	\begin{enumerate}[(i)]
		\item Let $ V $ be a subset of $ [\mu] $ and $ \delta_{V} \leq \delta $ with $ \delta_{V} \in \N $. A \textbf{subproblem} of \textsc{Problem} \ref{Pblm Original 0-1 KP} is an integer problem with the following structure
		\begin{equation*} 
		\max \sum\limits_{i\, \in \, V } p(i) \, x(i) ,
		\end{equation*}
		subject to
		\begin{align*} 
		& \sum\limits_{i\, \in \, V } w(i) \, x(i) \leq \delta_{V},
		\end{align*}
		\begin{align*} 
		& x(i) \in \{0,1\}, &
		& \text{for all } i \,\in \,V .
		\end{align*}	
		In the sequel, the subproblem will be denoted by $ \Pi_{V} \defining \big\langle \delta_{V}, \big( x(i) \big)_{i \in V} ,  \big( w(i) \big)_{i \in V} \big\rangle $
		
		\item Let $ (V_{0}, V_{1}) $ be a set partition of $ [\mu] $ and let $ (\delta_{0}, \delta_{1}) $ be an integer partition of $ \delta $ (i.e., $ \delta = \delta_{0} + \delta_{1} $). We say a Divide-and-Conquer pair of \textsc{Problem} \ref{Pblm Original 0-1 KP} is the couple of subproblems $ \big(\Pi_{b}: \, b \in \{0,1\} \big) $, each with input data $ \Pi_{b} = \big\langle \delta_{b}, \big( p(i)\big)_{ i \, \in \, V_{b} }, \big( w(i)\big)_{ i \, \in \, V_{b} } \big\rangle $.
		%
		In the sequel, we refer to $ \big(\Pi_{b}, b = 0, 1 \big) $ as a \textbf{D\&C pair} and denote by $ z_{b}^{*} $ the optimal solution value of the problem $ \Pi_{b} $. 		
		
		\item A \textbf{D\&C tree} (see \textsc{Example} \ref{Exm 0-1KP and DC tree Asymetric} and \textsc{Figure} \ref{Fig Tree Generated by Balanced Algorithm} below) for \textsc{Problem} \ref{Pblm Original 0-1 KP} is defined recursively by \textsc{Algorithm} \ref{Alg Branch Function}. Its input is an instance $ \Pi_{0} = \big\langle \delta, (p(i))_{i \in [\mu]}, (w(i))_{i \in [\mu]} \big\rangle $ of \textsc{Problem} \ref{Pblm Original 0-1 KP} and a minimum size of subproblems $ \zeta $. It satisfies the following properties
		\begin{enumerate}[a.]
			\item Every \textbf{vertex} of the tree is in \textbf{bijective correspondence} with a \textbf{subproblem} $ \Pi $ of $ \Pi_{0} $.
			
			\item The \textbf{root} of the tree is associated with \textsc{Problem} \ref{Pblm Original 0-1 KP} itself.
			
			\item Every \textbf{internal vertex} $ \Pi $ (which is not a leave) has a \textbf{left} and \textbf{right child}, $ \Pi_{\lt}, \Pi_{\rt} $ respectively. Its children make a D\&C pair for the subproblem $ \Pi $, whose generation is given by \textsc{Algorithm} \ref{Alg Branch Function}.
			
		\end{enumerate} 
		\item Let $ \Pi = \big\langle \delta, (p(i))_{i \in [\mu]}, (w(i))_{i \in [\mu]} \big\rangle $ be an instance of a 0-1KP and let $ \T $ be a D\&C tree. The method uses the search space and objective values
		\begin{align}\label{Eq DAC tree solution}
		& \x_{\T} \defining \bigcup\limits_{ L \text{ is a leave of } \T } \x_{L}, &
		& z_{\T} \defining \sum\limits_{ L \text{ is a leave of } \T } z_{L}  .
		\end{align}
		Here, we introduce some abuse of notation, denoting by $ \x_{L} $ a feasible solution (a vector) of $ \Pi_{L} $ and using the same symbol as a set of chosen items (instead of a vector) in the union operator. In particular, the maximal possible value occurs when all the summands are at its maximum i.e., the method approximates the optimal solution by $ \x_{\T}^{*} \defining \bigcup\{  \x_{L}^{*}: L \text{ is a leave of } \T \}  $ with objective value $ z_{\T}^{*} \defining \sum\{  z_{L}^{*}: L \text{ is a leave of } \T \} $.
		\begin{algorithm}[h!]
			\caption{Divide-and-Conquer tree generation branch function, returns a D\&C tree $ \tree $ of \textsc{Problem} \ref{Pblm Original 0-1 KP}.}
			\label{Alg Branch Function}
			\begin{algorithmic}[1]
				\Function{Branch}{ Subproblem: $ \Pi = \langle \delta, (p(i))_{i \in V}, (w(i))_{i \in V} \rangle  $, D\&C Tree: $ \tree $, Minimum problem size: $ \zeta $ }
				\State \textbf{compute} $ s $ (split item), $ z^{\gr} = \sum_{i = 1}^{s - 1} p(i) $ (objective function value), 
				\State \textbf{compute} $ k = \delta - \sum_{ i = 1 }^{s - 1} w(i) $ (slack) for problem $ \Pi $ 
				\Comment{Greedy Algorithm, Definition \ref{Def The split item} (ii)}
				\State \textbf{compute} $ z^{\eg}  $ for problem $ \Pi $ 
				\Comment{Extended Greedy Algorithm, Definition \ref{Def The split item} (iii)}
				\If {$ z^{\gr} \geq z^{\eg} $ and $ \vert V \vert \geq 2\zeta $}
				\Comment{Branching condition}
				\State $ V_{\lt} \defining \big[ i: i \in V, i \text{ is in odd relative position}   \big] $ \Comment{Computing the left child indexes}
				\State $ V_{\rt} \defining \big[ i: i \in V, i \text{ is in even relative position}   \big] 
				$\Comment{Computing the right child indexes}
				\State $ \delta_{\lt} \defining \lceil \frac{ 1 }{ 2 } \times k \rceil 
				+ \sum\limits_{ i = 1, \, i \, \text{odd} }^{s - 1} w(i) $ \Comment{Computing the left capacity}
				\State $ \delta_{\rt} \defining \lfloor \frac{ 1 }{ 2 } \times k \rfloor 
				+ \sum\limits_{i = 1, \, i \, \text{even} }^{s - 1}w(i) $
				\Comment{Computing the right capacity}
				\State $ \Pi_{\lt} \defining \big\langle \delta_{\lt}, (p(i))_{i \in V_{\lt}}, (w(i))_{i \in V_{\lt}}\big\rangle$
				\Comment{Defining the left child problem $ \Pi_{\lt} $}
				\State $ \Pi_{\rt} \defining \big\langle \delta_{\rt}, (p(i))_{i \in V_{\rt}}, (w(i))_{i \in V_{\rt}}\big\rangle$
				\Comment{Defining the right child problem $ \Pi_{\rt} $}
				\State $ \Pi_{\lt} \hookrightarrow V(\tree)  $, $ ( \Pi, \Pi_{\lt} ) \hookrightarrow E(\tree) $, 
				$ \Pi_{\rt} \hookrightarrow V(\tree)  $, $ ( \Pi, \Pi_{\rt} ) \hookrightarrow E(\tree) $ \\
				\Comment{Pushing problems $ \Pi_{\lt}, \Pi_{\rt} $ as nodes and $ (\Pi, \Pi_{\lt}), ( \Pi, \Pi_{\rt}) $ as arcs of the D\&C tree $ \tree $}
				\State \textsc{Branch}($ \Pi_{\lt}, \tree, \zeta $) 
				\Comment{Recursing for the left subtree}
				\State \textsc{Branch}($ \Pi_{\rt}, \tree, \zeta $) \Comment{Recursing for the right subtree}
				\State \Return $ \tree $ \Comment{output D\&C tree}
				\Else 
				\State \Return $ \tree $ \Comment{output D\&C tree}
				\EndIf
				\EndFunction
			\end{algorithmic}
		\end{algorithm}
	\end{enumerate}
\end{definition}
\begin{remark}[Divide-and-Conquer pairs and trees]\label{Rem Divide-and-Conquer tree}
	Observe the following about the algorithm \ref{Alg Branch Function} defined below 
	\begin{enumerate}[(i)]
		\item The instance of \textsc{Problem} \ref{Pblm Original 0-1 KP}, $ \Pi_{0} = \big\langle \delta, (p(i))_{i \in [\mu]}, (w(i))_{i \in [\mu]} \big\rangle $, to be solved with the Divide-and-Conquer method is assumed to satisfy \textsc{Hypothesis} \ref{Hyp Non Triviality and Sorting}.
		
		\item Before calling the \textsc{Branch} function for the first time, the D\&C tree $ \tree $ must be initialized as $ V(\tree) \defining \{ \Pi_{0} \} , E(\tree) \defining \emptyset $.
		
		\item When defining the ordered sets $ V_{\lt} $ the sentence ``is in odd relative position" is used, signifying:, those indexes which occupy odd positions in the sorted set $ V $ (the analogous holds for $ V_{\rt} $). For instance, observe the subproblem $ \Pi_{1} $ in \textsc{Example} 
		\ref{Exm 0-1KP and DC tree Asymetric}, \textsc{Figure} \ref{Fig Tree Generated by Balanced Algorithm}. Here the indexes $ 1, 5 $ are in odd relative positions (1 and 3 respectively), while $ 3, 7 $ are in even relative positions (2 and 4 respectively) inside the sorted set $ [1, 3, 5, 7] $. Hence, $ V_{\lt} = [ 1, 5 ] $ and $ V_{\rt} = [3, 7] $ (subsets for $ \Pi_{2} $ and $ \Pi_{3} $ subproblems of problem $ \Pi_{1} $).
		
		\item The definition of $ V_{\lt}, V_{\rt} $ subdividing the list of eligible items $ V $ for each node of the tree $ \tree $, is adopted because it has been observed empirically in \cite{MoralesMartinez} (balanced left-right subtrees, Section 4.2) that the Divide-and-Conquer method is expected to produce better results with this branching process.
		%
		\item The condition for branching: ($ z^{\gr} \geq z^{\eg} $ and $ \vert V \vert \geq 2 \zeta $) states that a subproblem will not be further subdivided if $ z^{\gr} < z^{\eg} $ or if the number of items $ \vert V \vert < 2 \zeta $. The first condition is discussed in \textsc{Theorem} \ref{Thm Worst Case Scenario} and \textsc{Remark} \ref{Rem Greedy Algorithm Possible Growth} below, while the second aims to ensure that no problem will be smaller that $ \zeta $. The latte condition is adopted, because it has been observed empirically in \cite{MoralesMartinez} that the Divide-and-Conquer method no longer produces good results beyond a problem size threshold, namely $ \zeta $.
	\end{enumerate}
\end{remark}
\begin{example}[Divide-and-Conquer tree]\label{Exm 0-1KP and DC tree Asymetric}
	Consider the 0-1KP instance described by the table \ref{Tbl Example Problem}, with knapsack capacity $ \delta = 7 $ and number of items $ \mu = 8 $. 
	\begin{table}[h!]
		\begin{centering}
						\rowcolors{2}{gray!25}{white}
			\begin{tabular}{c | c c c c c c c c}
								\rowcolor{gray!80}
				\hline
				$ i $ & 1 & 2 & 3 & 4 & 5 & 6 & 7 & 8
				\\
				$ w(i) $ & 3 & 2 & 3 & 3 & 4 & 7 & 1 & 5\\
				$ p(i) $ & 11.7 & 7.0 & 9.3 & 8.4 & 8.4 & 9.1 & 0.7 & 1.0 \\
				$ g(i) $ & 3.9 & 3.5 & 3.1 & 2.8 & 2.1 & 1.3 & 0.7 & 0.2\\
				\hline		
			\end{tabular}
			\caption{0-1KP problem of \textsc{Example} \ref{Exm 0-1KP and DC tree Asymetric}, knapsack capacity $ \delta = 7 $, number of items $ \mu = 8 $.}
			\label{Tbl Example Problem}
		\end{centering}
	\end{table}

	\noindent In this particular case 
	\begin{align*}
	& s = 3, &
	\x^{ \eg} = \x^{ \gr } & = [1, 1, 0, 0, 0, 0, 0 , 0] ,
	&
	z^{\gr} & = 18.7 = z^{\eg}, \\
	& & & &  k & = 7 - \sum\limits_{ i \, = \, 1 }^{ 8 } w(i) x^{\gr}(i) = 2,
	\\
	& &
	\x^{ *} & = [1, 0, 1, 0, 0, 0, 1, 0] , 
	&
	z^{*} & = 21.7 .
	\end{align*}
	Here, $ k $ denotes the slack in the knapsack. Hence, due to Algorithm \ref{Alg Branch Function} it follows that 
	\begin{align*}
	& \Pi_{ \lt }: &
	V_{ \lt} & = [1, 3, 5, 7 ], &
	\delta_{\lt} = 3 + 1 & \\
	& & & & (3 \text{ from item 1 and 1 from the slack } \lceil \tfrac{ k }{ 2 }\rceil ) , & 
	\\
	& &
	\x_{\lt}^{*} & = [1, 0, 0, 1], &
	z_{\lt}^{*} = 12.4, &
	\\
	& &
	\x_{\lt}^{\gr} = \x_{\lt}^{\eg} & = [1, 0, 0, 0], &
	z_{\lt}^{\gr} = z^{\eg}_{\lt} = 11.7. & \\
	& \Pi_{ \rt }: &
	V_{ \rt} & = [2, 4, 6, 8 ], &
	\delta_{\rt} = 2 + 1 & \\
	& & & & (2 \text{ from item 2 and 1 from the slack } \lfloor \tfrac{ k }{ 2 }\rfloor ) , &
	\\
	& &
	\x_{\rt}^{\eg} = \x_{\rt}^{*} & = [0, 1, 0, 0], &
	z_{\rt}^{\eg} =  z_{\rt}^{*} = 8.4. &
	\\
	& &
	\x_{\rt}^{\gr} & = [1, 0, 0, 0], &
	z_{\rt}^{\gr} = 7.8. &
	\end{align*}
	In this case $ z^{*} > z^{*}_{\lt} + z^{*}_{\rt} $. Next, given that $ z^{\gr}_{\lt} = z^{\eg}_{\lt} $ we repeat the same procedure for $ \Pi_{\lt} $, however we do not branch on $ \Pi_{\rt} $ since $ z^{\gr}_{\rt} < z^{\eg}_{\rt} $; this is observed in \textsc{Table} \ref{Tbl Balanced Tree for Particular Example} and \textsc{Figure} \ref{Fig Tree Generated by Balanced Algorithm}.
	\begin{table}[h!]
		\scriptsize{
			\begin{minipage}{0.3\textwidth}
				\begin{centering}
										\rowcolors{2}{gray!25}{white}
					\begin{tabular}{c | ccccc}
												\rowcolor{gray!80}
						\hline
						\diagbox{Item}{Vertex} & $ V_{0} $ & $ V_{1} $ & $ V_{2} $ & $ V_{3} $ & $ V_{4} $ 
						\\
						\hline
						1 &	1 & 1 & 1 & 0 & 0 \\			
						2 &	1 & 0 & 0 & 0 & 1 \\
						3 &	1 & 1 & 0 & 1 & 0 \\		
						4 &	1 & 0 & 0 & 0 &	1 \\
						5 &	1 & 1 & 1 & 0 & 0 \\		
						6 &	1 & 0 & 0 & 0 & 1 \\
						7 &	1 & 1 & 0 & 1 & 0 \\			
						8 &	1 & 0 & 0 & 0 &	1  \\
						\hline
												\rowcolor{gray!80}
						Capacity $ \delta $ & 7 & 3 &	 3 & 0 & 4  \\
						\hline		
					\end{tabular}
					\caption{\textsc{Algorithm} \ref{Alg Branch Function}, D\&C tree generated for the 0-1KP instance described in \textsc{Table} \ref{Tbl Example Problem}.}
					\label{Tbl Balanced Tree for Particular Example}
				\end{centering}
			\end{minipage}
			%
			\hspace{0.3in}
			\begin{minipage}{0.65\textwidth}
				\centering
				\begin{tikzpicture}
				[scale=.7,auto=left,every node/.style={}]
				\node (n0) at (10,6) {
					$ 
					\begin{pmatrix}
					V_{0} = [1, 2, 3, 4, 5, 6, 7, 8] \\[3pt]
					\delta_{0} = 7
					\end{pmatrix} \equiv \Pi_{0} $
				};
				\node (n1) at (6,3)  {$ 
					\begin{pmatrix}
					V_{1} = [1, 3, 5, 7] \\[3pt]
					\delta_{1} = 3
					\end{pmatrix} \equiv \Pi_{1} $};
				\node (n2) at (3.5,0)  {$ 
					\begin{pmatrix}
					V_{2} = [1, 5]\\[3pt]
					\delta_{2} = 3
					\end{pmatrix} \equiv \Pi_{2} $};
				\node (n3) at (8.5, 0)  {$ 
					\begin{pmatrix}
					V_{3} =  [3, 7]\\[3pt]
					\delta_{3} = 0
					\end{pmatrix} \equiv \Pi_{3} $};
				\node (n4) at (14,3)  {$ 
					\begin{pmatrix}
					V_{4} = [2, 4, 6, 8] \\[3pt]
					\delta_{4} = 4
					\end{pmatrix} \equiv \Pi_{4} $};
				
				\foreach \from/\to in {n0/n1, n1/n2,n1/n3, n0/n4
				}
				\draw[thick, ->] (\from) -- (\to);
				\end{tikzpicture}
				
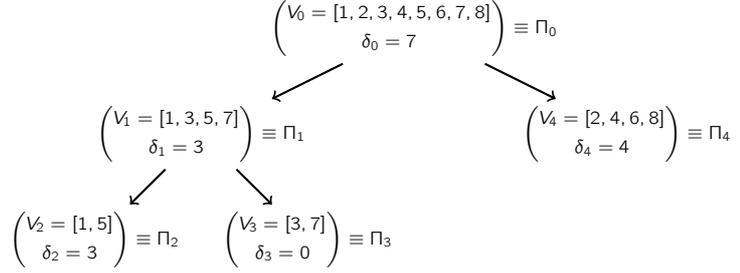
\captionof{figure}{\textsc{Algorithm} \ref{Alg Branch Function} D\&C tree generated for \textsc{Table} \ref{Tbl Balanced Tree for Particular Example}. Every vertex $ \Pi_{\ell} $ is a subproblem of the 0-1 KP instance $ \Pi_{0} = \langle \delta_{0},  (p(i))_{i \in V_{0}} , (w(i))_{i \in V_{0}} \rangle $.}
				\label{Fig Tree Generated by Balanced Algorithm}
			\end{minipage}
		}
	\end{table}
\end{example}
\begin{theorem}\label{Thm quality certificate DCM}
	Let $ \Pi = \big\langle \delta, \big( p(i) \big)_{ i = 1 }^{ \mu }, \big( w(i) \big)_{ i = 1 }^{ \mu } \big\rangle  $ be an instance of the 0-1KP introduced in \ref{Pblm Original 0-1 KP}. Let $ (V_{n})_{ n = 1 }^{ N } $ be a partition of $ [\mu] $ and let $ \tilde{\x} \defining \big( \tilde{x}(i) \big)_{i = 1}^{\mu} $ be a fixed feasible solution to the 0-1 KP problem. Hence, if  
	\begin{align}\label{Stmt Capacity Split}
	& \delta = \sum\limits_{n \, = \, 1}^{ N } \delta_{n} \, , &
	& \sum\limits_{i\, \in \, V_{n}  } w(i) \, \tilde{x}(i) \leq \delta_{n} , \text{ for all } n \in [ N ] , 
	\end{align}
	then 
	\begin{equation}\label{Ineq quality certificate DCM}
	\sum\limits_{i \, = \, 1}^{\mu} p(i) \, \tilde{x}(i) \leq 
	\sum\limits_{ n \, = \, 1 }^{ N } z^{*}_{n} .
	\end{equation}
	Here $ z^{*}_{n} $ is the optimal solution of the subproblem $ \Pi_{n} = \big\langle \delta_{n}, \big( p(i)\big)_{ i \, \in \, V_{n} }, \big( w(i)\big)_{ i \, \in \, V_{n} } \big\rangle $ for all $ n  = 1, \ldots , N $. In the following we refer to $ \tilde{\x} $ as the \textbf{control solution}. 
\end{theorem}
\begin{proof}
	It is direct to see that that $ \big( \tilde{x}(i) \big)_{i \in V_{n}} $ is a feasible solution of $ \Pi_{n} $ for all $ n = 1, \ldots, N $, due to the capacities condition \eqref{Stmt Capacity Split}. Hence, $ \sum\limits_{i \, \in \, V_{n}} p(i) \, \tilde{x}(i) \leq z^{*}_{n} $ for each $ n = 1, \ldots , N $, then
	\begin{equation*}
	\sum\limits_{ n \, = \, 1 }^{ N }\sum\limits_{i \, \in \, V_{n}} p(i) \, \tilde{x}(i) 
	\leq \sum\limits_{ n \, = \, 1 }^{ N }\ z^{*}_{n} . 
	\end{equation*}
	Given that $ (V_{n})_{ n = 1 }^{ N } $ is a partition of $ [\mu] $, the inequality \eqref{Ineq quality certificate DCM} follows. 
\end{proof}
\begin{remark}\label{Rem Capacities Split}
	Notice that if an optimal solution $ \big(x^{*}(i) \big)_{ i = 1 }^{ \mu } $ of \textsc{Problem} \ref{Pblm Original 0-1 KP} satisfies the set of capacities constraint \eqref{Stmt Capacity Split} then $ z^{*} \leq  \sum_{ n \, = \, 1 }^{ N }\ z^{*}_{ n }$ i.e., the D\&C collection of subproblems $ (\Pi_{ n } )_{ n = 1 }^{ N } $ reduces the computational complexity of \textsc{Problem} \ref{Pblm Original 0-1 KP} at no expense of precision, which is the ideal scenario.
\end{remark}
\begin{theorem}\label{Thm Worst Case Scenario}
	Let $ \Pi $ be a 0-1KP instance,
	\begin{enumerate}[(i)]
		\item Let $ \Pi_{\lt}, \Pi_{\rt} $ be a D\&C pair for the 0-1KP instance $ \Pi $. Let $ \x^{\gr}, \x^{\gr}_{\lt}, \x^{\gr}_{\rt} $ and $ z^{\gr}, z^{\gr}_{\lt} + z^{\gr}_{\rt},$  be their corresponding solutions and objective function values, furnished by the greedy algorithm. Then $ \x^{\gr} $ and $ V_{\lt}, V_{\rt} $ satisfy the hypothesis of \textsc{Theorem} \ref{Thm quality certificate DCM}. Moreover,
		\begin{equation}\label{Ineq DC pair greedy algorithm}
		z^{\gr} \leq z^{\gr}_{\lt} + z^{\gr}_{\rt}, 
		\end{equation}
		where $ z^{\gr}_{\lt}, z^{\gr}_{\rt} $ are the greedy algorithm solutions for $ \Pi_{\lt} $ and $ \Pi_{\rt} $ respectively.
		
		\item 
		Let $ z^{*}_{\T} $ be the optimal approximation value furnished by a D\&C tree $ \T $ of $ \Pi $, generated by \textsc{Algorithm} \ref{Alg Branch Function}. Then
		\begin{equation}\label{Ineq Worst Case Scenario}
		\frac{ 1 }{ 2 } \leq \frac{ z^{*}_{\T} }{ z^{*} }, 
		\end{equation}
		where $ z^{*} $ is the optimal value for the problem $ \Pi $.
	\end{enumerate}
\end{theorem}
\begin{proof}
	\begin{enumerate}[(i)]
		\item It is direct to see that $ \x^{\gr} $ and $ V_{\lt}, V_{\rt} $ satisfy the hypothesis of \textsc{Theorem} \ref{Thm quality certificate DCM}, because of how $ \delta_{\lt} $ and $ \delta_{\rt} $ are defined in \textsc{Algorithm} \ref{Alg Branch Function}. Moreover, such definition ensures that the inequality \eqref{Ineq DC pair greedy algorithm} holds.

		\item Let $ \x^{\eg} \defining \big( x^{\eg}(i) \big)_{i = 1}^{\mu} $ be the extended-algorithm solution for the problem $ \Pi $; observe that if $ \x^{\gr} \neq \x^{\eg} $ then $ \T = \{ \Pi_{0} \}$, due to the method's definition (see \textsc{Algorithm} \ref{Alg Branch Function}) and the result is obvious. Hence, from now on we assume that $ \x^{\gr} = \x^{\eg} $. 
		
		Consider $ \big\{ \Pi_{L}  = \big\langle \delta_{L}, \big( p(i)\big)_{ i \, \in \, V_{L} }, \big( w(i)\big)_{ i \, \in \, V_{L} } \big\rangle : L \text{ is a leave of } \T \big\} $, due to the theorem 4 in \cite{MoralesMartinez}, the collection $ \big\{ V_{L}: L \text{ is a leave of } \T \big\} $ is a partition of $ [\mu] $. Then, in order to prove the result, it suffices to show that $ \x^{\gr} $ and $ \big\{ V_{L}: L \text{ is a leave of } \T \big\} $ satisfy the hypothesis of \textsc{Theorem} \ref{Thm quality certificate DCM}. We prove this by induction on the number of Divide-and-Conquer iterations used to generate the tree. Let $ \{ \Pi \} \defining \T_{0} , \T_{1}, \ldots, \T_{n} = \T $ be the colection of trees attained by subsequent iterations of the Divide-and-Conquer method, with $ \T_{0} $ the original problem and $ \T_{n} $ the tree of interest. For $ \T_{0} $ the result is obvious and for $ \T_{1} $ this was proved in the previous part. Denote by $ (\Pi^{j})_{ j = 1}^{ J } $ the leaves of $ \T_{n - 1} $, due to the induction hypothesis, the solution $ \x^{\gr} $ and $ (V^{j} )_{j = 1}^{J} $ satisfy the hypothesis of \textsc{Theorem} \ref{Thm quality certificate DCM}. But then, due to the first part, for each problem $ \Pi^{j} $, it holds that 
		\begin{align*}
		& \delta^{j} = \delta^{j}_{\lt} + \delta^{j}_{\rt} \, , &
		& \sum\limits_{i\, \in \, V^{j}_{\lt}  } w(i) \, x^{\gr}(i) \leq \delta^{j}_{\lt} , &
		& \sum\limits_{i\, \in \, V^{j}_{\rt}  } w(i) \, x^{\gr}(i) \leq \delta^{j}_{\rt} .
		\end{align*}
		Hence,
		\begin{equation*}
		\delta = \sum\limits_{j \, = \, 1}^{ J } \delta^{j} 
		= \sum\limits_{j \, = \, 1}^{ J }  \delta^{j}_{\lt} + \delta^{j}_{\rt}
		= \sum\limits_{L \text{ leave of } \T } \delta_{L}
		\end{equation*}
		and recalling that $ \{ L : L \text{ is a leave of } \T \} $ is in bijective correspondence with $ \big\{ \Pi^{j}_{\side}: j = 1, \ldots, J, \side \in \{ \lt, \rt\} \big\} $, we conclude that $ \x^{\gr} $ and $ \big\{V_{L}: L \text{ is a leave of } \T \big\} $ satisfy the hypothesis of \textsc{Theorem} \ref{Thm quality certificate DCM}. Hence,
		\begin{equation*}
		z^{\eg} = z^{\gr} = \sum\limits_{i \, = \, 1}^{\mu} p(i) \, x(i) \leq 
		\sum\limits_{ n \, = \, 1 }^{ \nu } z^{*}_{n} = z^{*}_{\T}.
		\end{equation*}
		But then, $ \dfrac{ z^{*}_{\T} }{z^{*}}
		\geq
		\dfrac{ z^{\gr} }{z^{*}} 
		\geq \dfrac{ 1 }{ 2 } $, 
		where the last bound holds due to the inequality \eqref{Ineq Greedy Efficiency} from \textsc{Theorem} \ref{Thm Greed and LP Solutions} part (ii).
	\end{enumerate}
\end{proof}
\begin{remark}\label{Rem Greedy Algorithm Possible Growth}
	We observe some facts in \textsc{Theorem} \ref{Thm Worst Case Scenario} above
	\begin{enumerate}[(i)]
		
		\item It is possible to have a strict inequality in the expression \eqref{Ineq DC pair greedy algorithm}. To see this, let $ s $ be the split items for $ \Pi $ then, $ w(s) > k = \delta - \sum_{i \, = \, 1}^{s - 1} w(i) $ which stops the algorithm. However, it is possible that $ w(s + 1) \leq \lceil \tfrac{ k }{ 2} \rceil $ for $ s $ even, or $ w( s + 1) \leq \lfloor \tfrac{ k }{ 2} \rfloor $ for $ s $ odd. In these cases we would necessarily have $ z^{\gr} < z^{\gr}_{\lt} + z^{\gr}_{\rt} $, because one more item could be packed by the greedy algorithm in the problem $ \Pi_{\side} $ ($ \side \in \{ \lt, \rt \} $), for which the item $ s $ is not assigned.
		
		\item When $ \x^{\gr} =\x^{\eg} $, this is a control solution for any D\&C tree built by \textsc{Algorithm} \ref{Alg Branch Function}. In order to have this global control solution, there is no need to require that $ z^{\gr}_{\Pi} = z^{\eg}_{\Pi} $ for every node $ \Pi $ of $ \T $ as the algorithm requires for branching. However, it has been observed empirically, that removing this requirement, heavily deteriorates the quality of the solution in a Divided-and-Conquer iteration.
		
		\item If $ z^{\gr} < z^{\eg} $ a rule for assigning capacities $ \delta_{\lt}, \delta_{\rt} $ different from the one used by \textsc{Algorithm} \ref{Alg Branch Function} could be defined. However, given that the extended-greedy algorithm is intractable from the probabilistic point of view (as mentioned in \textsc{Remark} \ref{Rem Greedy Algorithms}), this would also make intractable the probabilistic analysis of the Divide-and-Conquer method. 
		
		\item In the proof of \textsc{Theorem} \ref{Thm Worst Case Scenario} we introduced a slight inconsistency with the notation adopted so far, by switching from subindex to superscript to denote a particular family of problems $ \Pi^{j} $ and its associated elements $ \delta^{j} , V^{j} $. This was done out of necessity this one time throughout the paper.
		
	\end{enumerate}
\end{remark}
%
%
%
\subsection{Results from Combinatorics and Probabiilty}
%
%
We devote this subsection to recall some previous background necessary to analyze the 0-1KP from the probabilistic point of view. We begin with a concept from combinatorics 
\begin{definition}[Compositions]\label{Def Compositions}
	Let  $ (a_{ 1 }, \ldots, a_{ m } ) $ be a sequence of integers satisfying $ \sum_{ i \, = \, 1 }^{ m } a_{ i } = n $.
	%
	If $ a_{ i } \geq 1 $ for all $ i = 1, \ldots, m $, the sequence is said to be a \textbf{composition} of $ n $ in $ m $ parts. (Naturally $ m $ should be less or equal that $ n $.)
	%
\end{definition}
\begin{theorem}\label{Thm Theorems from Bona}
	Let $ n,  m $ be two natural numbers with $ m \leq n $ then
	%
	\begin{enumerate}[(i)]
		\item 
		\begin{equation}\label{Eq Theorem from Bona 2}
		{ n \choose m } = 
		{ m - 1 \choose m - 1 }
		+ { m \choose m - 1 }
		+ \ldots 
		+
		{ n - 1 \choose m - 1 } .
		\end{equation}
		
		
		\item The number of compositions of $ n $ into $ m $ parts is $ \displaystyle { n - 1 \choose m - 1 } $.
		
		\item The following identity holds
		\begin{equation}\label{Eq Binomial Coefficients Indexes Shift}
		{n \choose m} = \frac{ n }{ m } { n -1 \choose m - 1 } .
		\end{equation}
	\end{enumerate}
\end{theorem}
\begin{proof}
	\begin{enumerate}[(i)]
		\item	See Theorem 4.5 in \cite{BonaWalk}.
		
		
		\item See Corollary 5.3 in \cite{BonaWalk}.
		
		\item By direct calculation.  See also Theorem 2.4 in \cite{AllenbySlomson} for a combinatorial proof of this fact.
	\end{enumerate}
\end{proof}
\begin{proposition}\label{Thm Balance Even-Odd sums in compositions}
	Let $ \A $ be the set of compositions of $ n $ in $ m $ parts. Denote by $ \alpha = (a_{1}, \ldots, a_{m}) $,  $ \beta = (b_{1}, \ldots, b_{m}) $, the elements of $ \A $ and define the quantities
	\begin{align}\label{Eqn Balance Even-Odd sums in compositions}
	\Sigma_{\odd}\defining & \sum\limits_{ \alpha \, \in \, \A } \sum\limits_{i \odd } a_{i}, & 
	\Sigma_{\even}\defining & \sum\limits_{ \alpha \, \in \, \A } \sum\limits_{i \even } a_{i} .
	\end{align}
	\begin{enumerate}[(i)]
		\item If $ m $ even, then $ \Sigma_{\odd} \equiv \Sigma_{\even} $.
		
		\item If $ m $ odd, then 
		$ \displaystyle \Sigma_{\odd} \equiv 
		\Sigma_{\even} + \frac{ 1 }{ 2 } \frac{n + 1}{\ell + 1} { n \choose 2 \ell + 1}  
		= \Sigma_{\even} + \frac{ 1 }{ 2 } \frac{n + 1}{\ell + 1} \# \A $, where $ m = 2 \ell + 1 $.
	\end{enumerate}
\end{proposition}
\begin{proof}
	\begin{enumerate}[(i)]
		\item Since $ m = 2\ell $, consider the permutation $ \sigma \in \S([m]) $ defined by
		\begin{align*}
		\sigma & : [ m ] \rightarrow [ m ], &
		\sigma(i) \defining 
		\begin{cases}
		i + 1 , & i \text{ is odd},\\
		i - 1 , & i \text{ is even}.
		\end{cases}
		\end{align*}
		Define the map 
		\begin{align*}
		B : \A & \rightarrow \A \\
		\alpha = (a_{1}, a_{2}, \ldots, a_{2\ell - 1}, a_{2\ell}) & \mapsto
		(a_{2}, a_{1}, \ldots, a_{2\ell}, a_{2\ell - 1}) 
		=
		(a_{\sigma(1)}, a_{\sigma(2)}, \ldots, a_{\sigma(2\ell - 1)}, a_{\sigma(2\ell)}) 
		. 
		\end{align*}
		It is direct to see that $ B $ is a bijection, then $  \sum\limits_{ \alpha \, \in \, \A } \sum\limits_{i \, = \, 1 }^{m} a_{i} 
		=  \sum\limits_{ \substack{ \beta \, = \, B(\alpha) \\ \alpha \, \in \, \A }} \sum\limits_{i \, = \, 1 }^{m} b_{i} $. Moreover
		\begin{align*}
		%
		\Sigma_{\odd} =  \sum\limits_{ \alpha \, \in \, \A } \sum\limits_{i \text{ odd }} a_{i} 
		=  \sum\limits_{ \substack{ \beta \, = \, B(\alpha) \\ \alpha \, \in \, \A }} \sum\limits_{i \text{ odd }} b_{i}
		= \sum\limits_{ \alpha \, \in \, \A } \sum\limits_{i \text{ odd } } a_{ \sigma(i) }  
		& = \sum\limits_{ \alpha \, \in \, \A } \sum\limits_{i \text{ even }} a_{i} 
		= \Sigma_{\even} ,
		\end{align*}
		which concludes the first part.
		
		\item Since $ m = 2\ell + 1$, consider the permutation $ \sigma \in \S([m]) $ defined by
		\begin{align*}
		\sigma & : [ m ] \rightarrow [ m ], &
		\sigma(i) \defining 
		\begin{cases}
		2 \ell + 1 , & i = 2 \ell + 1 ,\\
		i + 1 , & i \text{ is odd, } i \neq 2 \ell + 1 ,\\
		i - 1 , &  i \text{ is even}.
		\end{cases}
		\end{align*}
		As in the previous part, define the map 
		\begin{align*}
		B : \A \rightarrow & \A \\
		\alpha = (a_{1}, a_{2}, \ldots, a_{2\ell - 1}, a_{2\ell}, a_{2\ell + 1}) \mapsto &
		(a_{2}, a_{1}, \ldots, a_{2\ell}, a_{2\ell - 1}, a_{2\ell + 1}) \\
		& =
		(a_{\sigma(1)} , a_{\sigma(2)}, \ldots, a_{\sigma(2\ell - 1)}, a_{\sigma(2\ell)}, a_{\sigma(2\ell + 1)}) 
		. 
		\end{align*}
		As before, this is a bijection, however if we are to use it for computing the difference between $ \Sigma_{\lt} $ and $ \Sigma_{\rt} $ further specifications need to be done. Observe that the range of $ a_{2 \ell + 1} $ is $ \{ 1, \ldots, n - 2\ell \} $ and define $ \A_{i} = \{ \alpha \in \A: a_{2\ell + 1 } = i \} $. Observe that $ B: \A_{i} \rightarrow \A_{i} $ is also a bijection and that there is a bijection between $ \A_{i} $ and the set of compositions of $ n - i $ in $ 2 \ell $ parts. In particular (due to Theorem \ref{Thm Theorems from Bona} (ii)), it has $ { n - i \choose 2 \ell } $ elements and due to the previous part, we have
		\begin{equation*}
		\Sigma_{ \odd } (\A_{i}) = \Sigma_{ \even}(\A_{i})
		+ i { n - i \choose 2 \ell } .
		\end{equation*}
		Here, $ \Sigma_{ \odd } (\A_{i}) $ and $ \Sigma_{ \even}(\A_{i}) $ are defined by equation \eqref{Eqn Balance Even-Odd sums in compositions}. Therefore
		\begin{equation*}
		\begin{split}
		\Sigma_{\odd} = \sum\limits_{i \, = \,1}^{n - 2 \ell} \Sigma_{\odd} (\A_{i}) & =
		\sum\limits_{i \, = \,1}^{n - 2 \ell} \Sigma_{\even} (\A_{i}) +
		\sum\limits_{i \, = \,1}^{n - 2 \ell} i {n - i \choose 2\ell} 
		=
		\Sigma_{\even}  + \sum\limits_{i \, = \,1}^{n - 2 \ell} i {n - i \choose 2\ell} .
		\end{split}
		\end{equation*}
		We focus on the last sum
		\begin{align*}
		\sum\limits_{i \, = \,1}^{n - 2 \ell} i {n - i \choose 2\ell} 
		= 
		&
		\sum\limits_{j \, = \,2 \ell}^{n - 1} ( n - j ) {j \choose 2\ell} 
		\\
		= 
		&
		(n + 1)\sum\limits_{j \, = \,2 \ell}^{n - 1}  {j \choose 2\ell}
		-
		\sum\limits_{j \, = \,2 \ell}^{n - 1} ( j + 1 ) {j \choose 2\ell}
		\\
		= 
		&
		(n + 1)\sum\limits_{j \, = \,2 \ell}^{n - 1}  {j \choose 2\ell}
		-
		(2 \ell + 1)\sum\limits_{m \, = \,2 \ell + 1}^{n}  {m \choose 2\ell + 1} 
		\\
		= &
		(n + 1) {n \choose 2 \ell + 1} - 
		(2 \ell + 1) { n + 1 \choose 2 \ell + 2 } 
		.
		\end{align*}
		In the expression above, the second equality is a convenient association of summands, the third equality uses the identity \eqref{Eq Binomial Coefficients Indexes Shift} to adjust the binomial coefficient, while the fourth equality applies the expression \eqref{Eq Theorem from Bona 2}. Simplifying the latter and combining with the previous we have
		\begin{equation*}
		\Sigma_{\odd} = \Sigma_{\even} + \frac{ 1 }{ 2 } \frac{n + 1}{\ell + 1} { n \choose 2 \ell + 1} ,
		\end{equation*}
		which is the desired result.
	\end{enumerate}
\end{proof}
Next we recall some results from basic discrete probability
\begin{theorem}\label{Thm Discrete Conditional Expectations}
	Let $ ( \Omega, \prob ) $ be a discrete probability space and let $ ( \Omega_{n} )_{ n \, = \, 1 }^{ N } $ be a partition of $ \Omega $ then 

	\begin{enumerate}[(i)]
		\item Let $ A, B \subseteq \Omega $ be two events then 
		\begin{subequations}
			\begin{align}
			\prob( A , B ) & = \prob(A \cap B) 
			= \prob( A \big\vert B \big) \prob( B ) ,
			\label{Eq Conditional Probability}
			\\
			\prob( A ) & = \sum\limits_{ n \, = \, 1 }^{ N } 
			\prob\big( A \big\vert \Omega_{ n } \big) \prob\big(\Omega_{n}\big)
			\label{Eq Probability Using Conditional Probability and Partition}
			.
			\end{align}
		\end{subequations}
		\item Let $ \X : \Omega \rightarrow \R $ be a discrete random variable, let $ A \subseteq \Omega $ be an event then 
		\begin{subequations}
			\begin{align}
			\Exp\big( \X \big\vert A \big) & = 
			\sum
			\limits_{ x \, \in \, \X(\Omega) } 
			x \prob\big( \X = x \big\vert A \big) 
			,
			\label{Eq Conditional Expectation wrt Event} \\
			%
			\Exp\big( \X \,\ind_{ A } \big) & = 
			\Exp\big( \X \big\vert A \big) \prob( A),
			\label{Eq Expectation Inside an Event} \\
			%
			\Exp( \X ) & = \sum\limits_{ n \, = \, 1 }^{ N } \Exp\big( \X \big\vert \Omega_{ n } \big) 
			\prob(\Omega_{n})
			\label{Eq Expectation wrt Partition} 
			.
			\end{align}
		\end{subequations}
		In the expression \eqref{Eq Conditional Expectation wrt Event}, $ \X(\Omega) $ stands for the range of the random variable $ \X $.
	\end{enumerate}
\end{theorem}
\begin{proof}
	\begin{enumerate}[(i)]
		\item For \eqref{Eq Conditional Probability} see Definition 1.3.7 in \cite{BremaudDiscrete}. For \eqref{Eq Probability Using Conditional Probability and Partition} see Theorem 1.3.9 in  \cite{BremaudDiscrete}.
		
		\item For \eqref{Eq Conditional Expectation wrt Event} see Section 2.3.9, page 49 in \cite{BremaudDiscrete}. For \eqref{Eq Expectation Inside an Event} see Theorem 2.3.1 in in \cite{BremaudDiscrete}. Finally, noticing that $ \Exp(\X ) = \sum_{ n \, = \, 1 }^{ N} \Exp\big( \X\, \ind_{ \Omega_{ n }} \big) $ and the identity  \eqref{Eq Expectation Inside an Event}, the equation \eqref{Eq Expectation wrt Partition} follows.
	\end{enumerate}
\end{proof}
%
%
%
%
%
%
\section{Probabilistic Analysis of 0-1KP}\label{Sec The Expected Performance}
%
%
In this section, we present the probabilistic analysis of the Divide-and-Conquer method. We begin introducing the probabilistic model. 
%
%
\subsection{The Probabilistic Model and the Random 0-1KP}\label{Sec Probabilistic Model}
%
%
%
%
\begin{hypothesis}[The Random Model]\label{Hyp Random Problems}
	The random instances \newline 
	$ \big\langle \delta, \big( \upw(i) \big)_{i = 1}^{\mu}, \big(\upp(i) \big)_{i = 1}^{\mu} \big\rangle $ of the knapsack problem to be analyzed satisfy 
	\begin{enumerate}[a.]
		\item The capacity $ \delta $ and the number of items $ \mu $, with $ \mu = \delta + 1 $, are fixed.
		
		\item The weights $ \big( \upw(i) \big)_{i = 1 }^{ \mu } $ are i.i.d. random variables, uniformly distributed on the discrete set $ [\delta] \defining \{1,\ldots,  \delta\} $ for all $ i \in [\mu] $.
		
		
		\item The profits $ \big( \upp(i) \big)_{i = 1 }^{ \mu } $ are defined by means of the weights and the efficiencies $ \big( \upg(i) \big)_{i = 1 }^{ \mu } $. To define the efficiencies we introduce a set of random variables named \textbf{the increments} $ \big( \upt(i) \big)_{i = 1}^{\mu} $, which are i.i.d., continuous, uniformly distributed on the interval $ (0, 1) $ for all $ i \in [\mu] $. Hence, the efficiencies $\upg(i) $ and profits $ \upp(i) $ are defined by
		\begin{align}\label{Eq Random Efficiencies}
		& \upg(i) = \sum\limits_{t \, = \, i}^{\mu} \upt(t) , &
		& \upp(i) = \upg(i) \, \upw(i), &
		& \text{for all } i \in [\mu].
		\end{align}
	\end{enumerate}
\end{hypothesis}
\begin{definition}[The Random Model]\label{Def Random Problems}
	With the random model introduced in the hypothesis \ref{Hyp Random Problems} above, we define the following problems 
	\begin{enumerate}[(i)]
		\item The random version of the problem \ref{Pblm Original 0-1 KP} is given by
		%
		\begin{subequations}\label{Eqn Random Integer Problem}
			\begin{equation}\label{Eqn Random Integer Problem Objective Function}
			\max \sum\limits_{i\, = \, 1 }^{ \mu } \upp(i) \, x(i) ,
			\end{equation}
			subject to
			\begin{equation}\label{Eqn Random Integer Problem Capacity Constraint}
			\sum\limits_{i\, = \, 1}^{ \mu } \upw(i) \, x(i)  \leq \delta,
			\end{equation}
			\begin{align}\label{Eqn Random Integer Problem Choice Constraint}
			& x(i) \in \{0,1\}, &
			& \text{for all } i \,\in \,[\mu] .
			\end{align}	
		\end{subequations}
		From now on we refer to it as \textbf{0-1RKP}.
		
		\item The random version of problem \ref{Pblm Natural LOP Problem} is analogous to how 0-1RKP is generated. In the sequel, we refer to it as \textbf{0-1RLPK}.
	\end{enumerate}
	%
\end{definition}
\begin{remark}\label{Rem Probabilistic Model}
	\begin{enumerate}[(i)]
		\item It is direct to see that the random instances of \textsc{Problem} \eqref{Eqn Random Integer Problem} satisfy the conditions of \textsc{Hypothesis} \ref{Hyp Non Triviality and Sorting}. In particular the efficiencies $ \big( \upg(i) \big)_{i = 1}^{\mu} $ verify the monotonicity condition 
		\begin{equation}\label{Eq Monotone Especific Weigths}
		\upg(1) \geq \upg(2) \geq \ldots \geq \upg(\mu)  .
		\end{equation}

		\item Since $ \upw(i) \geq 1 $ for all $ i = 1, \ldots, \mu $, it follows that the number of packed items is at most $ \delta $ (i.e., $ \sum_{i\, = \, 1}^{\mu} x(i) \leq \delta $), hence we adopt $ \mu = \delta + 1 $ for mathematical convenience.
		
		\item In the figure \ref{Fig Random Realizations Model} we depict two random realizations for the weights, profits and efficiencies, according to the proposed probabilistic model, \textsc{Table} \ref{Tbl Random Realization Values} summarizes the values of the random variables for both realizations.

		\begin{figure}[h!]
			\centering
			\begin{minipage}{0.49\textwidth}
				\begin{subfigure}[Weights $ \big( \upw(i) \big)_{i = 1}^{13} $ for two random realizations.]
					{\includegraphics[scale = 0.480]{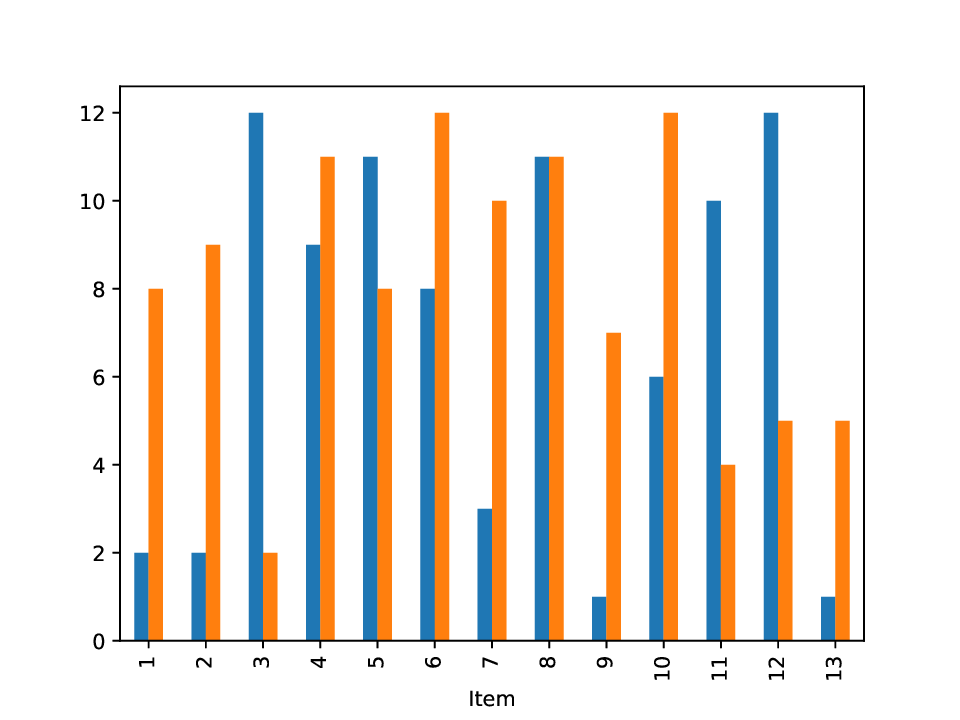} } 
				\end{subfigure}
			\end{minipage}
			%
			%
			~ 
			\begin{minipage}{0.49\textwidth}
				\begin{subfigure}[Profits $ \big( \upp(i) \big)_{i = 1}^{13} $ for two random realizations.]
					{\includegraphics[scale = 0.480]{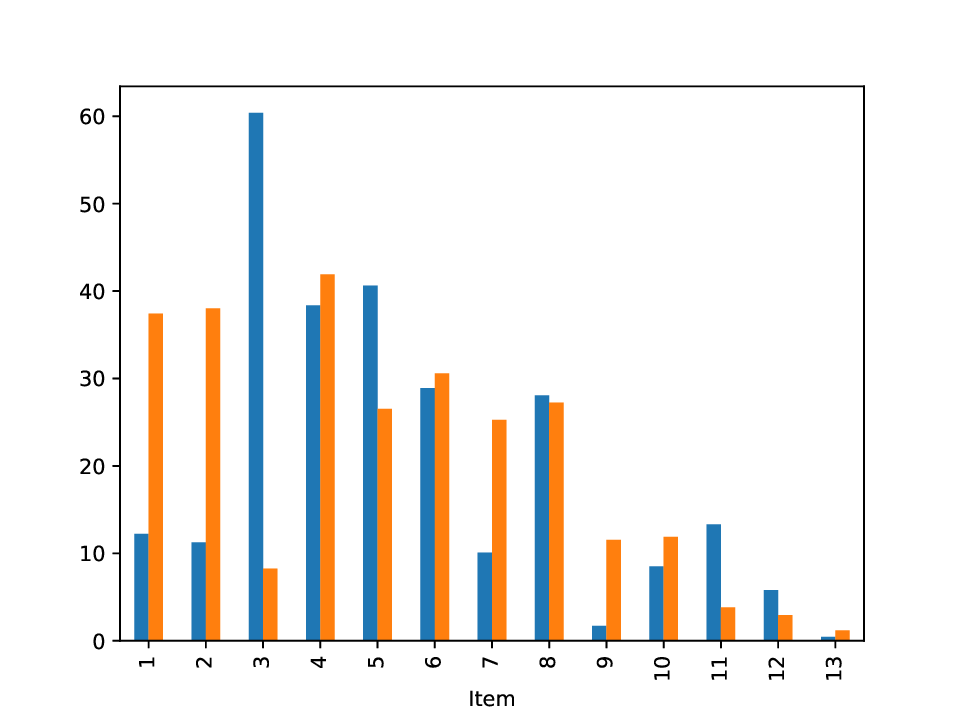} } 
				\end{subfigure}
			\end{minipage}
			\begin{minipage}{0.49\textwidth}  
				\begin{subfigure}[Efficiencies $ \big( \upg(i) \big)_{i = 1}^{13} $ for two random realizations.]
					{\includegraphics[scale = 0.48]{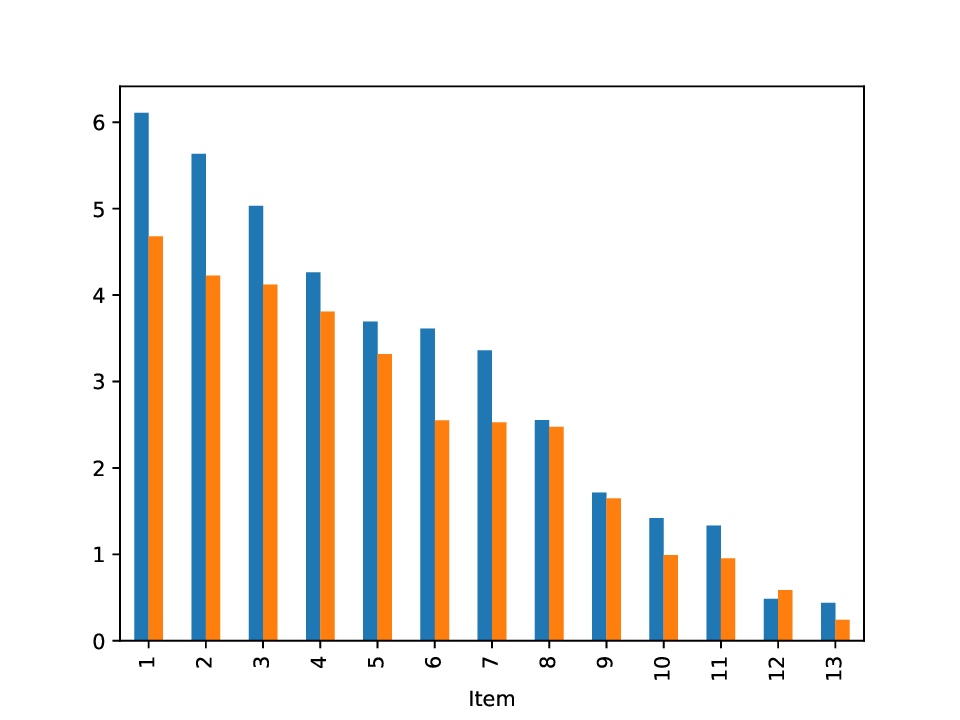} } 
				\end{subfigure}  
			\end{minipage} 
			%
			%
			\begin{minipage}{0.45\textwidth}
				\begin{center}
					\footnotesize{
												\rowcolors{2}{gray!25}{white}
						\begin{tabular}{c | c  c | c  c  | c  c}
														\rowcolor{gray!80}
							\hline
							Item & \multicolumn{2}{ c | }{$ \upw(i) $} & \multicolumn{2}{ c | }{$ \upp(i) $} & \multicolumn{2}{ c }{$ \upg(i) $}\\[2pt] 
														\rowcolor{gray!80}
							$ i $ & 1 & 2 & 1 & 2 & 1 & 2  \\[2pt]
							\hline 
							1 & 2	& 8 & 12.22 & 37.43 & 6.11 & 4.68 \\
							2 & 2	& 9 & 11.27 & 38.03	& 5.63 & 4.23 \\
							3 & 12 & 2 & 60.39 &	8.24	& 5.03 & 4.12 \\
							4 & 9	& 11 & 38.35 &	41.9	& 4.26 & 3.81 \\
							5 & 11 & 8	& 40.63 &	26.53 & 3.69 &	3.32 \\
							6 & 8	& 12	& 28.89 & 30.6	& 3.61 & 2.55 \\
							7 & 3	&10	& 10.08 &	25.26 & 3.36 &	2.53 \\
							8 & 11 & 11 & 28.07 & 27.23 & 2.55 & 2.48 \\
							9 & 1	& 7	& 1.72 & 11.54	& 1.72 & 1.65 \\
							10 & 6 & 12 & 8.51 &	11.88 & 1.42 &	0.99 \\
							11 & 10 & 4 & 13.33 & 3.81 & 1.33 & 0.95 \\
							12 & 12 & 5 & 5.81 & 2.93 & 0.48 & 0.59 \\
							13 & 1 & 5 & 0.44 & 1.2 & 0.44	& 0.24 \\
							\hline		
						\end{tabular}
						\captionof{table}{Numerical values for the two random realizations depicted in the graphs above.}
						\label{Tbl Random Realization Values}
					}
				\end{center}
				\normalsize
			\end{minipage}
			\caption{Two random realizations according to the probabilistic model introduced in \textsc{Definition} \ref{Def Random Problems}, capacity $ \delta = 12 $, number of items $ \mu = 13 $. Figure (a) displays the weights  $ ( \upw(i) )_{i = 1}^{13} $, while figure (b) depicts the profits  $ ( \upp(i) )_{i = 1}^{13} $ and (c) portrays the values of the efficiencies  $ ( \upg(i) )_{i = 1}^{13} $. The blue color indicates the first realization while the orange stands for the second realization. All the corresponding numerical values are summarized in the table \ref{Tbl Random Realization Values}.}
			\label{Fig Random Realizations Model} 
		\end{figure}
	\end{enumerate}
\end{remark}
In order to compute expected values for the Greedy Algorithm, two important random variables have to be introduced
\begin{definition}\label{Def Split Item, Slack and Packed Items Random Variables}
	Let $ \big\langle \delta, \big( \upw(i) \big)_{i = 1}^{\mu} , \big( \upp(i) \big)_{i = 1}^{\mu} \big\rangle $ be a random instance satisfying the hypothesis \ref{Hyp Random Problems}, define
	\begin{enumerate}[(i)]
		\item The \textbf{split item} random variable $ \ups $ is the value of the index $ s $ (introduced in \textsc{Definition} \ref{Def The split item} (i)) for the random instance.
		
		\item The \textbf{slack} random variable is defined by 
		\begin{equation}\label{Eq Slack Random Variable}
		\upk \defining \delta - \sum\limits_{j \, = \, 1}^{\ups- 1} \upw_{j} ,
		\end{equation}
		where $ \ups $ is the split item random variable.
		
		
	\end{enumerate}
	%
\end{definition}
%
%
%
%
\subsection{Expectations of the 0-1RKP and 0-1RLPK related variables
}\label{Sec Expectations Calculations}
%
%
In this section we compute the expectations of the most important random variables related to the probabilistic model introduced in \textsc{Section} \ref{Sec Probabilistic Model}. We begin presenting a result which turns out to be the cornerstone of our whole construction.
\begin{lemma}[Cornerstone Lemma]\label{Thm Cornerstone Lemma}
	Let $ \ups $ and $ \upk $ be the split item and the slack random variables defined above, then
	\begin{align}
	\prob\big( \upk  = k , \ups = s \big) & = 
	\frac{ \delta - k }{\delta^{ s }} {\delta - k - 1\choose s - 2} , 
	\label{Eq Knapsack Slack Joint Conditioning} 
	\end{align}
	for $ s = 2, \ldots, \mu $ and $ k = 0, \ldots, \delta - s + 1 $.
\end{lemma} 
\begin{proof}
	Observe the following equivalence of events
	\begin{equation*}
	\begin{split}
	\prob\big( \upk  = k , \ups = s \big) & = 
	\prob\Big( \delta - \sum\limits_{ j\, = \, 1} ^{ s - 1 } \upw(j)  = k  , \upw( s ) >  k \Big)
	= 
	\prob\Big( \sum\limits_{ j\, = \, 1 }^{ s - 1 } \upw(j)  = \delta - k \Big) \prob\Big( \upw(s) >  k  \Big) 
	.
	\end{split}
	\end{equation*}
	The last equality uses the independence of the weight random variables. For the first factor we observe that the event occurs if and only if $ (\upw(i))_{i = 1}^{s - 1} $ is a composition of $ \delta - k $. According to \textsc{Theorem} \ref{Thm Theorems from Bona} (ii) there are $ {\delta - k - 1\choose s - 2} $ of these compositions and since $ \upw(i) $ are uniformly distributed over $ [\delta] $, the probability for each of these compositions to occur is $ \frac{ 1 }{ \delta^{ s - 1} } $. Next, the event $ [ \upw(s) > k ] $ has probability $ \frac{ \delta - k }{ \delta } $ again due to the uniform distribution of the variable.  Combining the previous observations, the equality \eqref{Eq Knapsack Slack Joint Conditioning} follows. 
	
	Finally, for the range of the variables, it is direct to see that $ \ups, \upk $ are intertwined, then regarding $ \ups $ as the independent and $ \upk $ as the dependent, the first can range freely inside $ \{ 2, \ldots, \delta + 1 \} $ while the second only takes values within $ \{ 2, \ldots, \delta - s + 1 \} $ because 
	$ \sum_{ j\, = \, 1 }^{s - 1} \upw(j) \geq s -1 $.
\end{proof}
Next we compute the distribution, expectation and variance of $ \ups $.
\begin{theorem}\label{Thm Distribution Split Item}
	Let $ \ups $ be the splitting item random variable defined above, then its distribution and expectation are given by
	\begin{subequations}
		\begin{align}
		\prob(\ups = s) & = 
		\frac{ s - 1 }{ \delta^{ s } } 
		{ \delta + 1 \choose s }
		, 
		\quad\text{for all } s = 2, \ldots,  \mu ,
		\label{Eq Split Item Expectation Properties}
		\\
		%
		\Exp( \ups )
		& = 
		\big( 1 + \frac{ 1 }{ \delta } \big)^{ \delta } ,
		\label{Eq Split Expectation}
		\\
		%
		%
		%
		\Var( \ups )
		& = 
		\big( 3 + \frac{ 1 }{ \delta } \big)
		\big( 1 + \frac{ 1 }{ \delta }\big)^{ \delta - 1 } 
		- \big( 1 + \frac{ 1 }{ \delta }\big)^{ 2\delta } ,
		\label{Eq Split Variance}
		%
		\end{align}
		where $ \mu = \delta + 1 $. 
	\end{subequations}
\end{theorem}
\begin{proof}
	Due to the cornerstone lemma \ref{Thm Cornerstone Lemma} if $ \ups = s $, the slack $ k = \delta - \sum_{j \, =\, 1}^{ s - 1 } \upw(j) $ runs from $ 0 $ to $ \delta - s + 1 $. Hence, we split the event $ \{ \ups = s \} $, according to the range of the slack, i.e. 
	\begin{equation*}
	\begin{split}
	\prob(\ups = s) = \sum\limits_{ k\, = \, 0 }^{ \delta - s + 1 } \prob( \upk = k, \ups = s ) 
	& =  \sum\limits_{ k\, = \, 0 }^{ \delta - s + 1 } \frac{ \delta - k }{ \delta^{ s } } { \delta - k - 1 \choose s - 2 } 
	=  \sum\limits_{ m \, = \, s - 1 }^{ \delta } \frac{ m }{ \delta^{ s } } { m - 1 \choose s - 2 }
	.
	\end{split}
	\end{equation*}
	%
	%
	The first equality holds due to the cornerstone identity \eqref{Eq Knapsack Slack Joint Conditioning} while the second is a mere reindexing of the sum. Recalling that $ (m + 1){ m \choose s - 2 } = ( s - 1 ) { m + 1 \choose s - 1} $ due to the identity \eqref{Eq Binomial Coefficients Indexes Shift}, we have
	\begin{equation*}
	\prob(\ups = s) =  
	\frac{ s - 1 }{ \delta^{ s } }
	\sum\limits_{ m \, = \, s - 2 }^{ \delta }  { m  \choose s - 1 }
	= \frac{ s - 1 }{ \delta^{ s } } { \delta + 1 \choose s } ,
	\end{equation*}
	where the last equality holds due to the combinatorial identity \eqref{Eq Theorem from Bona 2}. This proves the identity \eqref{Eq Split Item Expectation Properties}. Next, in order to compute $ \Exp(\ups) $, first recall that $ \mu = \delta + 1 $ and get
	\begin{equation*}
	\begin{split}
	\sum\limits_{ s \, = \, 2 }^{ \delta + 1 }
	s \prob(\ups = s) 
	=
	\sum\limits_{s\, = \, 2}^{ \delta + 1 } 
	\frac{ s ( s - 1 ) }{ \delta^{ s } } 
	{ \delta + 1 \choose s } 
	& 
	=
	\sum\limits_{s\, = \, 0}^{ \delta + 1 } 
	\frac{ s ( s - 1 ) }{ \delta^{ s } } 
	{ \delta + 1 \choose s } 
	=
	\frac{ ( \delta + 1 ) \delta }{ \delta^{ 2 } } 
	\big( 1 + \frac{ 1 }{ \delta } \big)^{ \delta - 1 } .
	\end{split}
	\end{equation*}
	Applying some basic algebraic manipulations, the identity \eqref{Eq Split Expectation} follows. Finally, for the variance, first we compute
	\begin{equation*}
	\begin{split}
	\Exp( \ups^{ 2 } ) 
	& = \sum\limits_{ s \, = \, 2}^{ \delta + 1 }
	s^{ 2}  \prob( \ups = s ) 
	\\
	& 
	= \sum\limits_{ s \, = \, 2}^{ \delta + 1 }
	s^{ 2}  \frac{ s - 1 }{ \delta^{ s } } { \delta + 1 \choose s } 
	\\
	& 
	= \sum\limits_{ s \, = \, 2}^{ \delta + 1 }
	\frac{ s ( s - 1) ( s - 2 ) }{ \delta^{ s } } { \delta + 1 \choose s } 
	+ 2 \sum\limits_{ s \, = \, 2}^{ \delta + 1 }
	\frac{ s ( s - 1 ) }{ \delta^{ s } } { \delta + 1 \choose s } 
	\\
	& 
	= 
	\frac{( \delta + 1 ) \delta ( \delta - 1 ) }{ \delta^{ 3 } } \big( 1 + \frac{ 1 }{ \delta }\big)^{ \delta - 2 }
	+ 2 \big( 1 + \frac{ 1 }{ \delta }\big)^{ \delta} .
	\end{split}
	\end{equation*}
	Here, the third equality is a convenient association of summands and the fourth equality simply uses the derivatives of the Newton's binomial identity.
	Therefore, 
	\begin{equation*}
	\begin{split}
	\Var( \ups ) = \Exp( \ups^{ 2 } ) - \Exp^{2}(\ups)
	& 
	= 
	\big( 1 - \frac{ 1 }{ \delta }\big) \big( 1 + \frac{ 1 }{ \delta }\big)^{ \delta - 1 }
	+ 2 \big( 1 + \frac{ 1 }{ \delta }\big)^{ \delta } 
	- \big( 1 + \frac{ 1 }{ \delta }\big)^{ 2\delta } 
	.
	\end{split}
	\end{equation*}
	From here, the identity \eqref{Eq Split Variance} follows directly.
\end{proof}
Before computing the expectation of $ \Z^{\gr} $ the next technical lemma is needed.
\begin{lemma}\label{Thm Weight Conditional Expectation}
	With the definitions above, the following identities hold
	\begin{subequations}
		\begin{align}
		\prob\big(\upw(j) = w, \ups = s \big) & = 
		\frac{ 1 }{ \delta^{s} }
		\frac{( s - 2 )( \delta + 1) + w}{ s - 1 } { \delta - w \choose s - 2 } , &
		j = 1, \ldots , s - 1 &,
		\label{Eq Weight Conditional Probability S}
		\\
		\Exp\big( \upw(j) \big\vert \ups = s \big)  & =
		\frac{ \delta s + s - 1 }{ s^{ 2 } - 1 } , &
		j = 1, \ldots, s - 1 . &
		\label{Eq Weight Conditional Expectation}
		\end{align}
	\end{subequations}
\end{lemma}
\begin{proof} 
	For the first equality, observe that if $ \upw(j) = w $ then $ \sum_{ i \, = \, 1}^{s - 1} \upw(i) \geq ( s - 2 ) + w $ consequently, the slack $ \upk $ can take values only in the set $ \{ 0, \ldots, \delta - (s - 2) - w \} $. Hence,
	\begin{equation*}
	\begin{split}
	\prob\big(\upw(j) = w  , \ups = s \big) = &
	\sum\limits_{k \, = \, 0}^{\delta - (s - 2) - w} \prob\big(\upw(j) = w , \upk = k, \ups = s \big)
	\\
	= & 
	\sum\limits_{k \, = \, 0}^{\delta - s + 2 - w} \prob\Big(\upw(j) = w , \sum_{ m \, \in \, [s - 1] - j } \upw(m) = \delta - k - w, \upw(s) > k  \Big)
	\\
	= & 
	\sum\limits_{k \, = \, 0}^{\delta - s + 2 - w} \frac{ 1 }{ \delta } \frac{1}{\delta^{ s - 2 }} { \delta - k - w - 1 \choose s - 3 } \frac{ \delta - k }{ \delta }
	\\
	= & 
	\frac{ 1 }{ \delta^{ s } } \sum\limits_{\ell \, = \, s - 3}^{\delta - w - 1} 
	( \ell + 1 + w){ \ell \choose s - 3 } 
	.
	\end{split}
	\end{equation*}
	Here, the second equality is a direct interpretation of the event $ \prob\big(\upw(j) = w , \upk = k, \ups = s \big) $. The third equality is the application of the basic identity \eqref{Eq Knapsack Slack Joint Conditioning}, while the fourth equality is the reindexing of the sum by $ \ell \defining \delta - k - w - 1 $. We compute the latter sum as follows
	\begin{equation*}
	\begin{split}
	\sum\limits_{\ell \, = \, s - 3}^{\delta - w - 1} 
	( \ell + 1 + w){ \ell \choose s - 3 }  = 
	&
	\sum\limits_{\ell \, = \, s - 3}^{\delta - w - 1} 
	( \ell + 1){ \ell \choose s - 3 } 
	+
	w \sum\limits_{\ell \, = \, s - 3}^{\delta - w - 1} { \ell \choose s - 3 } 
	\\
	= & 
	(s - 2)\sum\limits_{\ell \, = \, s - 3}^{\delta - w - 1} 
	{ \ell + 1 \choose s - 2 } 
	+
	w \sum\limits_{\ell \, = \, s - 3}^{\delta - w - 1} { \ell \choose s - 3 } 
	\\
	= & 
	(s - 2)\sum\limits_{m \, = \, s - 2}^{\delta - w } 
	{ m \choose s - 2 } 
	+
	w \sum\limits_{\ell \, = \, s - 3}^{\delta - w - 1} { \ell \choose s - 3 } 
	\\
	= & 
	(s - 2)
	{ \delta - w + 1 \choose s - 1 } 
	+
	w  { \delta - w \choose s - 2 } .
	\end{split}
	\end{equation*}
	In the expression above, the second equality uses the identity \eqref{Eq Binomial Coefficients Indexes Shift} for shifting indexes, the third equality is a mere reindexing of the first sum and the fourth equality applies the identity \eqref{Eq Theorem from Bona 2}. From here, using again the identity $ { \delta - w + 1 \choose s - 1 }  = \frac{ \delta - w + 1}{ s - 1} { \delta - w \choose s - 2 } $ and performing further algebraic simplifications, the equation \eqref{Eq Weight Conditional Probability S} follows. 
	
	Next, we  prove the identity \eqref{Eq Weight Conditional Expectation}. Recalling the identity \eqref{Eq Conditional Expectation wrt Event} for conditional expectation, we get
	\begin{equation*}
	\begin{split}
	\Exp\big( \upw(j) \big\vert \ups = s \big)
	= &\sum\limits_{w \, = \, 1}^{ \delta - (s - 2) } w \prob\big(\upw(j) = w \big\vert \ups = s \big) 
	\\
	= & \frac{1}{ \frac{s - 1}{\delta^{s}} { \delta + 1 \choose s} }\sum\limits_{w \, = \, 1}^{ \delta - (s - 2) } w \prob\big(\upw(j) = w , \ups = s \big) 
	\\
	= & \frac{1}{ \frac{s - 1}{\delta^{s}} { \delta + 1 \choose s} }
	\frac{ 1 }{ \delta^{s} }\sum\limits_{w \, = \, 1}^{ \delta - (s - 2) } \frac{( s - 2 )( \delta + 1) w + w^{2}}{ s - 1 } { \delta - w \choose s - 2 } .
	\end{split}
	\end{equation*}
	Here, the second equality used the identity \eqref{Eq Conditional Probability} combined with \eqref{Eq Split Item Expectation Properties}, while the third used the identity \eqref{Eq Weight Conditional Probability S}. We focus on getting a closed form for the sum; by reindexing $ u \defining \delta - w $ we get
	\begin{equation*}
	\sum\limits_{w \, = \, 1}^{ \delta - (s - 2) } \frac{( s - 2 )( \delta + 1) w + w^{2}}{ s - 1 }  { \delta - w \choose s - 2 }
	=  
	\frac{1}{ s - 1 }\sum\limits_{u \, = \, s - 2}^{ \delta - 1 } \big\{ ( s - 2 )( \delta + 1) (\delta - u) + (\delta - u)^{2} \big\} { u \choose s - 2 }. 
	\end{equation*}
	Appealing to the polynomial identity
	\begin{equation*}
	( s - 2 )( \delta + 1) (\delta - u) + (\delta - u)^{2} 
	= 
	(\delta + 1)^{2}(s - 1)
	- (s \delta + s + 1) (u + 1) +
	(u + 1 )(u + 2) ,
	\end{equation*}
	we have,
	\begin{multline*}
	\sum\limits_{u \, = \, s - 2}^{ \delta - 1 } 
	\big\{ 
	(\delta + 1)^{2}(s - 1)
	-  (s \delta + s + 1)  (u + 1) 
	+
	(u + 1 )(u + 2)
	\big\}
	{ u \choose s - 2 }
	\\
	= 
	(\delta + 1)^{2}(s - 1)
	\sum\limits_{u \, = \, s - 2}^{ \delta - 1 } 
	{ u \choose s - 2 } 
	- (s \delta + s + 1)
	\sum\limits_{u \, = \, s - 2}^{ \delta - 1 } 
	(u + 1)
	{ u \choose s - 2 }
	+ \sum\limits_{u \, = \, s - 2}^{ \delta - 1 } 
	(u + 1 )(u + 2)
	{ u \choose s - 2 }
	. 
	\end{multline*}
	Now, listing the three sums of the left hand side we have
	\begin{align*}
	\sum\limits_{u \, = \, s - 2}^{ \delta - 1 } 
	{ u \choose s - 2 } & = { \delta \choose s - 1} , \\
	\sum\limits_{u \, = \, s - 2}^{ \delta - 1 } 
	(u + 1)
	{ u \choose s - 2 }
	& =
	( s - 1 )
	\sum\limits_{u \, = \, s - 2}^{ \delta - 1 } 
	{ u + 1 \choose s - 1 }
	=
	( s - 1 )
	\sum\limits_{r \, = \, s - 1}^{ \delta  } 
	{ r \choose s - 1 }
	= ( s - 1 ) { \delta  + 1\choose s } ,
	\\
	\sum\limits_{u \, = \, s - 2}^{ \delta - 1 } 
	(u + 2 )(u + 1)
	{ u \choose s - 2 }
	& = 
	s(s - 1)
	\sum\limits_{u \, = \, s - 2}^{ \delta - 1 } 
	{ u + 2 \choose s }
	= 
	s(s - 1)
	\sum\limits_{r \, = \, s }^{ \delta + 1 } 
	{ r \choose s } 
	= s( s- 1) { \delta + 2 \choose s + 1 } .
	\end{align*}
	Combining the above with the previous gives
	\begin{equation*}
		\begin{split}
	\frac{s - 1}{\delta^{s}} { \delta + 1 \choose s} \delta^{s} \Exp\big( \upw(j) \big\vert \ups = s \big)
	= &
	(\delta + 1)^{2} { \delta \choose s - 1}
	- ( s \delta + s + 1) { \delta + 1 \choose s}
	+ s { \delta + 2 \choose s + 1 } 
	\\
	= &
	s (\delta + 1) { \delta + 1 \choose s }
	- ( s \delta + s + 1) { \delta + 1 \choose s}
	+ s \frac{\delta + 2 }{ s + 1 }{ \delta + 1 \choose s  } 
	\\
	= & 
	\frac{ s \delta + s - 1}{s + 1}{ \delta + 1 \choose s } .
		\end{split}
	\end{equation*}
	Here, the second equality uses the identity \eqref{Eq Binomial Coefficients Indexes Shift} in the first and third summand, while the second equality is the mere algebraic sum of the previous line. Finally, a direct simplification of terms yields the identity \eqref{Eq Weight Conditional Expectation} and the result is complete.
\end{proof}
\begin{theorem}\label{Thm Greedy Expected Profit}
	Let  $ \ups $ and $ \Z^{\gr} \defining \sum\limits_{i \, = \, 1}^{\ups - 1} \upp(i) $ be the split item and the greedy algorithm profit random variables for the 0-1RKP \eqref{Eqn Random Integer Problem}. Then, 
	\begin{subequations}
		\begin{align}
		\Exp\big( \Z^{\gr}  \big\vert \ups = s \big) 
		= &  
		\frac{  2 \delta - s + 4 }{ 4 }
		\frac{ \delta s + s - 1 }{ s + 1 } , 
		\text{ for all } s = 2, \ldots, \mu 
		,
		\label{Eq Conditional Greedy Expected Profit} \\
		%
		\begin{split}
		\Exp(\Z^{\gr}) 
		= &  
		- \frac{  ( \delta + 1 )^{ 2 }  }{ 4\delta } \big( 1 + \frac{ 1 }{ \delta } \big)^{ \delta - 1 }
		+
		\frac{ ( 2\delta + 3 ) ( \delta + 2 ) ( \delta + 1 )  }{ 4 \delta }  
		\Big\{ 
		\big( 1 + \frac{ 1 }{\delta } \big)^{ \delta } - 1
		\Big\}
		\\
		& -
		( \delta + 2 )^{ 2 }  
		\Big\{
		\big( 1 + \frac{ 1 }{ \delta } \big)^{ \delta + 1 }
		- 
		\frac{ 2 \delta + 1 }{ \delta }
		\Big\}
		+
		\frac{ 2 \delta + 5 }{ 2 } \, \delta
		\Big\{
		\big( 1 + \frac{ 1 }{ \delta } \big)^{ \delta + 2 }
		- \frac{ 5 \delta^{ 2 } + 7 \delta + 2 }{ 2 \delta^{2} }
		\Big\}
		,
		\end{split}
		\label{Eq Greedy Expected Profit}
		\end{align}
	\end{subequations}
	with $ \mu = \delta + 1 $.
\end{theorem}
\begin{proof}
	We compute the identity \eqref{Eq Conditional Greedy Expected Profit} directly, using the definition of $ \upp(i) $ introduced in \textsc{Equation} \eqref{Eq Random Efficiencies}
	\begin{equation*}
	\Exp\big( \Z^{\gr}  \big\vert \ups = s \big) = \sum\limits_{ j \, = \, 1 }^{ s - 1 } \Exp\big( \upp(j) \big\vert \ups = s \big) 
	= \sum\limits_{ j \, = \, 1 }^{ s - 1 } \Exp\big(\upw(j) \, \upg(j) \big\vert \ups = s \big) 
	= \sum\limits_{ j \, = \, 1 }^{ s - 1 } \Exp\Big(\upw(j)  \sum\limits_{t \, = \, j}^{\mu}   \upt(t)  \big\vert \ups = s \Big) .
	\end{equation*}
	Recalling that the variables $ \big( \upw(i) \big)_{i = 1}^{\mu} $ and $ \big( \upt(i) \big)_{i = 1}^{\mu} $ are independent, we have
	\begin{equation*}
	\begin{split}
	\Exp\big( \Z^{\gr}  \big\vert \ups = s \big) 
	& = \sum\limits_{ j \, = \, 1 }^{ s - 1 } \Exp\big( \upw(j) \big\vert \ups = s \big) 
	\sum\limits_{t \, = \, j}^{\mu}  \Exp\big( \upt(t) \big\vert \ups = s \big)
	\\
	& = \sum\limits_{ j \, = \, 1 }^{ s - 1 } 
	\frac{ \delta s + s - 1 }{ s^{ 2 } - 1 } 
	\frac{\mu - j + 1}{2} 
	\\
	& = 
	\frac{ (s - 1 ) ( 2 \mu - s + 2 )}{ 4 }
	\frac{ \delta s + s - 1 }{ s^{ 2 } - 1 } .
	\end{split}
	\end{equation*}
	Here, the second equality holds due to the identity \eqref{Eq Weight Conditional Expectation} and the distribution of the increments $ ( \upt(i) )_{ i  = 1 }^{ \mu } $ introduced in \textsc{Hypothesis} \ref{Hyp Random Problems}. Simplifying the expression above, the \textsc{Equation} \eqref{Eq Conditional Greedy Expected Profit} follows.
	
	Next, we compute the expectation of $ \Z^{\gr} $ conditioning on the possible values of $ \ups $
	and combining with the identities \eqref{Eq Conditional Greedy Expected Profit}, \eqref{Eq Split Item Expectation Properties}; this gives
	%
	%
	\begin{equation*}
	\begin{split}
	\Exp(\Z^{\gr}) 
	& =
	\sum\limits_{ s \, = \, 2 }^{ \delta + 1 }  
	\Exp\big( \Z^{\gr}  \big\vert \ups = s \big)  \prob\big( \ups = s \big) 
	\\
	& = 
	\sum\limits_{ s \, = \, 2 }^{ \delta + 1 } 
	\frac{ 2\delta - s + 4}{ 4 }
	\frac{ \delta s + s - 1 }{ s + 1 } 
	\frac{  s - 1 }{ \delta^{ s } }
	{ \delta + 1 \choose s }
	\\
	& = \frac{ 1 }{ 4 ( \delta + 2 ) }
	\sum\limits_{ m \, = \, 3 }^{ \delta + 2 } 
	\frac{ ( m - 2 ) ( 2 \delta -  m + 5 ) \big( ( \delta + 1 ) m - \delta - 2 \big) }{ \delta^{ m - 1 } }
	{ \delta + 2 \choose m }
	\\
	& = \frac{ 1 }{ 4 ( \mu + 1 ) }
	\sum\limits_{ m \, = \, 3 }^{ \mu + 1 } 
	\frac{ ( m - 2 ) ( 2 \mu -  m + 3 ) \big( \mu m - \mu - 1 \big) }{ \delta^{ m - 1 } }
	{ \mu + 1 \choose m }
	.
	\end{split}
	\end{equation*}
	The third equality in the expression above is a convenient reindexing of the sum, while the last equality follows from the substitution
	$ \mu = \delta + 1 $. Next, consider the polynomial identity
	\begin{multline*}
	( m - 2 ) ( 2\mu - m + 3)(\mu m - \mu - 1) =
	\\
	- \mu m ( m - 1 ) ( m - 2 ) + ( 2 \mu + 1 )( \mu + 1 )m ( m - 1 ) 
	- 4(\mu + 1)^{2} m +
	( 4 \mu + 6 ) ( \mu + 1 ) ,
	\end{multline*}
	and combine it with the expression above. We get
	\begin{equation*}
	\begin{split}
	\Exp(\Z^{\gr}) 
	= &
	- \frac{ 1 }{ 4\delta ^{ 2 } } \frac{ \mu }{ \mu + 1 }
	\sum\limits_{ m \, = \, 3 }^{ \mu + 1 } 
	\frac{ m ( m - 1 ) ( m - 2 ) }{ \delta^{ m - 3 } } { \mu + 1 \choose m } 
	+
	\frac{ 2\mu + 1 }{ 4 \delta } \sum\limits_{ m \, = \, 3 }^{ \mu + 1 } 
	\frac{ m ( m - 1 ) }{ \delta^{ m - 2 } } { \mu + 1 \choose m } 
	\\
	& -
	( \mu + 1 )  \sum\limits_{ m \, = \, 3 }^{ \mu + 1 } 
	\frac{ m }{ \delta^{ m - 1 } } { \mu + 1 \choose m }
	+
	\delta \frac{ 4 \mu + 6 }{ 4 }
	\sum\limits_{ m \, = \, 3 }^{ \mu + 1 } 
	\frac{ 1 }{ \delta^{ m } } { \mu + 1 \choose m }
	\\
	= &
	- \frac{ 1 }{ 4\delta ^{ 2 } } \frac{ \mu }{ \mu + 1 }
	( \mu + 1 ) \mu ( \mu - 1 ) \big( 1 + \frac{ 1 }{ \delta } \big)^{ \mu - 2 }
	+
	\frac{ 2\mu + 1 }{ 4 \delta }  
	\Big\{ 
	( \mu + 1 ) \mu \big( 1 + \frac{ 1 }{\delta } \big)^{ \mu - 1 } - ( \mu + 1 ) \mu
	\Big\}
	\\
	& -
	( \mu + 1 )  
	\Big\{
	( \mu + 1 )\big( 1 + \frac{ 1 }{ \delta } \big)^{ \mu }
	- ( \mu + 1 ) 
	- \frac{ \mu ( \mu + 1 ) }{ \delta }
	\Big\}
	\\
	& 
	+
	\delta \frac{ 4 \mu + 6 }{ 4 }
	\Big\{
	\big( 1 + \frac{ 1 }{ \delta } \big)^{ \mu + 1 }
	- 1 - \frac{ \mu + 1 }{ \delta}   - \frac{ ( \mu+ 1 ) \mu }{ 2 \delta^{ 2 } }  
	\Big\}
	. 
	\end{split}
	\end{equation*}
	Simplifying the latter and replacing back $ \delta = \mu - 1 $, the equality \eqref{Eq Greedy Expected Profit} follows.
	%
	%
\end{proof}
Next we find the distribution, conditional expectation with respect to $ \ups $ and expectation for the slack $ \upk $.
\begin{theorem}\label{Thm Knapsack Slack Expectation Properties}
	The slack random variable $ \upk $, introduced in \textsc{Definition} \ref{Def Split Item, Slack and Packed Items Random Variables} (ii), satisfies
	%
	\begin{subequations}\label{Eq Knapsack Slack Expectation Properties}
		\begin{align}
		\prob\big( \upk  = k \big) = &
		\frac{ \delta - k }{ \delta^{ 2 } }
		\big( 1 + \frac{ 1 }{ \delta } \big)^{ \delta - k - 1 } , 
		\text{ for all } k = 0, \ldots , \delta .
		\label{Eq Knapsack Slack Probability Distribution} \\
		%
		\Exp\big( \upk \big\vert \ups = s \big)
		= & \frac{ \delta + 1 - s }{ s + 1 }  \, ,
		\text{ for all } s = 2, \ldots, \mu.
		\label{Eq Knapsack Slack Conditional Expectation} \\
		%
		\begin{split}
		\Exp\big( \upk \big) 
		= &  
		- \frac{ ( \delta + 1 ) }{ \delta } 
		\Big\{ 
		\big( 1 + \frac{ 1 }{ \delta } \big)^{ \delta } - 1
		\Big\} 
		+ ( \delta + 3 ) 
		\Big\{ 
		\big( 1 + \frac{ 1 }{ \delta } \big)^{ \delta + 1 } 
		- \frac{ 2 \delta + 1 }{ \delta } 
		\Big\} \\
		& - 2 \delta 
		\Big\{
		\big( 1 + \frac{ 1 }{ \delta } \big)^{ \delta + 2 }
		- \frac{ 5 \delta^{ 2 } + 7 \delta + 2 }{ 2 \delta^{2} }
		\Big\}
		.
		\label{Eq Knapsack Slack Expectation}
		\end{split}
		\end{align}
	\end{subequations}
\end{theorem}
\begin{proof}
	Revisiting the cornerstone lemma \ref{Thm Cornerstone Lemma}, observe that fixing $ \upk = k $, the range of the split index $ s $ is $ \{ 2, \ldots, \delta - k + 1 \} $. Hence, 
	%
	\begin{equation*}
	\begin{split}
	\prob( \upk = k ) = 
	\sum\limits_{ s \, = \, 2 }^{ \delta - k + 1 } \prob( \upk = k , \ups = s) 
	= 
	&
	\sum\limits_{ s \, = \, 2 }^{ \delta - k + 1 }
	\frac{ s - 1 }{\delta^{ s }} { \delta - k \choose s - 1 } 
	= 
	\frac{ 1 }{ \delta^{ 2 } }
	\sum\limits_{ j \, = \, 1 }^{ \delta - k }
	\frac{ j }{\delta^{ j - 1 }} { \delta - k \choose j } 
	.
	\end{split}
	\end{equation*}
	Here, the second equality holds due to the cornerstone identity \eqref{Eq Knapsack Slack Joint Conditioning}, the third equality is a convenient reindexing and association of terms. 
	Then, applying the first derivative of the Newton's binomial expansion, the identity \eqref{Eq Knapsack Slack Probability Distribution} follows.  
	
	Next, we show the equality \eqref{Eq Knapsack Slack Conditional Expectation}. From the cornerstone lemma \ref{Thm Cornerstone Lemma} observe that if $ \ups = s $ the range of the slack $ \upk $ is $ \{ 0, \ldots, \delta - s  + 1 \} $. Hence, recalling the conditional expectation identity \eqref{Eq Conditional Expectation wrt Event} we get
	\begin{equation*}
	\Exp\big( \upk \big\vert \ups = s \big) = 
	\sum\limits_{ k \, = \, 0}^{\delta - s + 1} k \prob\big( \upk = k\big\vert \ups = s ) =
	\frac{ 1 }{\prob(\ups = s)} \sum\limits_{ k \, = \, 0}^{\delta - s + 1} k \prob\big( \upk = k , \ups = s ) .
	\end{equation*}
	Now, appealing to the basic identiy \eqref{Eq Knapsack Slack Joint Conditioning}, we have
	\begin{equation*}
	\begin{split}
	\Exp\big( \upk  \big\vert \ups = s \big)
	& = \frac{ 1 }{\prob(\ups = s)} \frac{ 1 }{ \delta^{ s } } \sum\limits_{ k \, = \, 0}^{\delta - s + 1} ( \delta - k ) k {\delta - k - 1 \choose s - 2}  
	\\
	& = \frac{ 1 }{\prob(\ups = s)} \frac{ 1 }{ \delta^{ s } } \sum\limits_{ j \, = \, s - 1 }^{ \delta } j ( \delta - j ) { j - 1 \choose s - 2}
	\\
	& = \frac{ 1 }{\prob(\ups = s)} \frac{ \delta + 1 }{ \delta^{ s } } \sum\limits_{ j \, = \, s - 1 }^{ \delta } j { j - 1 \choose s - 2 }
	-  \frac{ 1 }{\prob(\ups = s)} \frac{ 1 }{ \delta^{ s } } \sum\limits_{ j \, = \, s - 1 }^{ \delta }  ( j + 1 ) j { j - 1 \choose s - 2 }
	.
	\end{split}
	\end{equation*}
	Here, the second equality follows from reindexing $ j = \delta - k $, while the third equality is a mere convenient association of summands. Next from the identity \eqref{Eq Binomial Coefficients Indexes Shift}, we get the equalities $ \frac{ j }{ s - 1 } { j -1 \choose s - 2 } = { j \choose s - 1} $, $ \frac{ j ( j + 1) }{ s ( s - 1 ) } { j -1 \choose s - 2 } = { j + 1 \choose s}  $ for the first and second summands respectively.  From here, proceeding as in the proofs of \textsc{Lemma} \ref{Thm Weight Conditional Expectation} and \textsc{Theorem} \ref{Thm Greedy Expected Profit},
	the identity \eqref{Eq Knapsack Slack Conditional Expectation} follows.
	
	Finally, we pursue a closed form for $ \Exp( \upk ) $; to that end we apply the identity \eqref{Eq Expectation wrt Partition} and get
	\begin{equation*}
	\begin{split}
	\Exp\big( \upk \big)
	= 
	& 
	\sum\limits_{ s \, = \, 2 }^{ \delta + 1 } \Exp\big( \upk \big\vert \ups = s \big) \prob\big( \ups = s \big) 
	= 
	\sum\limits_{ s \, = \, 2 }^{ \delta + 1 } 
	\frac{ \delta + 1 - s }{ s + 1 } \frac{ s - 1 }{ \delta^{ s } } { \delta + 1 \choose s  }
	\\
	= 
	& 
	\frac{ 1 }{ \delta + 2 }
	\sum\limits_{ s \, = \, 2 }^{ \delta + 1 } 
	\frac{ ( \delta + 1 - s ) ( s - 1 ) }{ \delta^{ s } }
	{ \delta + 2 \choose s + 1 }
	= 
	\frac{ 1 }{ \delta + 2 }
	\sum\limits_{ j \, = \, 3 }^{ \delta + 2 } 
	\frac{ ( \delta + 2 - j ) ( j - 2 ) }{ \delta^{ j - 1 } }
	{ \delta + 2 \choose j }
	.
	\end{split}
	\end{equation*}
	Here, the second equality follows by replacing the equations \eqref{Eq Knapsack Slack Conditional Expectation} and \eqref{Eq Split Item Expectation Properties}. The third equality uses the identity \eqref{Eq Binomial Coefficients Indexes Shift} and the fourth equality follows from the convenient reindexing $ j = s + 1 $. Next, we replace the polynomial identity 
	\begin{equation*}
	(\delta + 2 - j) ( j - 2 ) = - j ( j + 1 ) + ( \delta + 3 ) j - 2 ( \delta + 2 ) ,
	\end{equation*}
	in the expression above and get
	\begin{equation*}
	\begin{split}
	( \delta + 2 ) \Exp\big( \upk \big)
	= & 
	- \frac{ 1 }{ \delta } 
	\sum\limits_{ j \, = \, 3 }^{ \delta + 2 } 
	\frac{ j ( j - 1 ) }{ \delta^{ j - 2 } }
	{ \delta + 2 \choose j }
	+ ( \delta + 3 ) \sum\limits_{ j \, = \, 3 }^{ \delta + 2 } 
	\frac{ j }{ \delta^{ j - 1 } }
	{ \delta + 2 \choose j }
	- 2 \delta ( \delta + 2 ) 
	\sum\limits_{ j \, = \, 3 }^{ \delta + 2 } 
	\frac{ 1 }{ \delta^{ j } }
	{ \delta + 2 \choose j } \\
	= &
	- \frac{ 1 }{ \delta } 
	\Big\{ ( \delta + 2 )( \delta + 1 )
	\big( 1 + \frac{ 1 }{ \delta } \big)^{ \delta } -  ( \delta + 2 ) ( \delta + 1 ) 
	\Big\} \\
	& + ( \delta + 3 ) 
	\Big\{ ( \delta + 2 )
	\big( 1 + \frac{ 1 }{ \delta } \big)^{ \delta + 1 } - ( \delta + 2 ) -  \frac{ ( \delta + 2 ) ( \delta + 1 ) }{ \delta }
	\Big\} \\
	& - 2 \delta ( \delta + 2 ) 
	\Big\{
	\big( 1 + \frac{ 1 }{ \delta } \big)^{ \delta + 2 }
	- 1 
	- \frac{ \delta + 2 }{ \delta }
	- \frac{ ( \delta + 2 ) ( \delta + 1 ) }{ 2 \delta^{ 2 } }
	\Big\}
	.
	\end{split}
	\end{equation*}
	Here, the second equality uses the Newton's binomial expansion, together with its first and second derivatives. 
	Finally, simplifying the latter expression the equality \eqref{Eq Knapsack Slack Expectation} follows.
\end{proof}
\begin{theorem}\label{Thm Linear Relaxation Expected Profit}
	Let  $ \Z^{\lp}
	$ be the optimal profit value given by the linear relaxation of the 0-1RKP \eqref{Eqn Random Integer Problem}. Then, its expected value is given by 
	\begin{equation}\label{Eq Linear Relaxation Expected Profit}
	\begin{split}
	\Exp(\Z^{\lp}) = & \Exp(\Z^{\gr}) 
	+ 
	\frac{ \delta + 1 }{ 2 \delta }
	\big( 1 + \frac{ 1 }{ \delta } \big)^{ \delta - 1 }
	- 
	\frac{ ( \delta + 2 ) ( \delta + 1 )  }{ \delta }
	\Big\{
	( 1 + \frac{ 1 }{ \delta }\big)^{ \delta }
	- 1
	\Big\}
	\\
	& 
	+ \frac{ ( \delta + 5 ) ( \delta + 2 ) }{ 2 }
	\Big\{
	\big( 1 + \frac{ 1 }{ \delta } \big)^{ \delta - 1 }
	-
	\frac{ 2 \delta + 1 }{ \delta }
	\Big\}
	-  ( \delta + 3 ) \delta
	\Big\{
	\big( 1 + \frac{ 1 }{ \delta } \big)^{ \delta + 2 }
	- \frac{ 5 \delta^{ 2 } + 7 \delta + 2 }{ 2 \delta^{2} }
	\Big\} 
	\\
	= & 
	- \frac{  ( \delta + 1 ) ( \delta - 1 )   }{ 4\delta } \big( 1 + \frac{ 1 }{ \delta } \big)^{ \delta - 1 }
	+
	\frac{ ( 2\delta - 1  ) ( \delta + 2 ) ( \delta + 1 )  }{ 4 \delta }  
	\Big\{ 
	\big( 1 + \frac{ 1 }{\delta } \big)^{ \delta } - 1
	\Big\}
	\\
	& 
	-
	\frac{ ( \delta + 2 ) ( \delta - 1 ) }{ 2 }
	\Big\{
	\big( 1 + \frac{ 1 }{ \delta } \big)^{ \delta + 1 }
	- 
	\frac{ 2 \delta + 1 }{ \delta }
	\Big\}
	-
	\frac{ 1 }{ 2 } \, \delta
	\Big\{
	\big( 1 + \frac{ 1 }{ \delta } \big)^{ \delta + 2 }
	- \frac{ 5 \delta^{ 2 } + 7 \delta + 2 }{ 2 \delta^{2} }
	\Big\}
	.
	\end{split}
	\end{equation}
	Here, $ \Z^{\gr} $ is the profit of the solution furnished by the greedy algorithm, whose expectation $ \Exp(\Z^{\gr}) $ is given by the identity \eqref{Eq Greedy Expected Profit}.
\end{theorem}
\begin{proof}
	Due to \textsc{Theorem} \ref{Thm Greed and LP Solutions} (ii), equation \eqref{Eq LR Objective}, we know that
	\begin{equation}\label{Eqn Random Linear Relaxation Solution}
	\Z^{\lp} =  \sum\limits_{j \, = \, 1}^{\ups - 1} \upp(j)
	+ \dfrac{1}{\upw(\ups)} \Big( \delta - \sum\limits_{j \, = \, 1}^{\ups - 1} \upw(j) \Big) \upp(\ups) 
	=  \Z^{\gr} +  
	\upk \cdot \upg(\ups) 
	=  \Z^{\gr} +  
	\upk \sum\limits_{\ell \, = \, \ups}^{\mu} \upt(\ell) .
	\end{equation}
	Hence, conditioning on $ \ups $ through its range and recalling the equalities \eqref{Eq Knapsack Slack Conditional Expectation}, \eqref{Eq Split Item Expectation Properties}, we have

	\begin{equation*}
	\begin{split}
	\Exp(\Z^{\lp} ) 
	= & \Exp(\Z^{\gr}) +  
	\sum_{ s\, = \, 2 }^{ \mu } 
	\Exp\Big(
	\upk
	\sum\limits_{ \ell \, = \, \ups}^{ \mu } \upt( \ell ) \Big\vert \ups = s \Big)
	\prob( \ups = s ) \\
	= & \Exp(\Z^{\gr}) +  
	\sum_{ s \, = \, 2 }^{ \mu } 
	\Exp\big( 
	\upk
	\big\vert \ups = s \big)
	\Exp\Big( \sum\limits_{\ell \, = \, \ups}^{ \mu } \upt(\ell) \Big\vert \ups = s \Big)
	\prob( \ups = s ) 
	\\
	= & \Exp(\Z^{\gr}) +  
	\sum_{ s \, = \, 2 }^{ \mu } 
	\frac{ \delta + 1 - s }{ s + 1 }
	\frac{ \mu - s + 1 }{ 2 }
	\frac{ s - 1 }{ \delta^{ s } } { \delta + 1 \choose s }
	\\
	= & \Exp(\Z^{\gr}) +  
	\frac{ 1 }{ 2 }
	\frac{ 1 }{ \delta + 2 }
	\sum_{ s \, = \, 2 }^{ \mu } 
	\frac{ ( \mu - s ) ( \mu + 1 - s ) ( s - 1 ) }{ \delta^{ s } } 
	{ \delta + 2 \choose s + 1 }
	\\
	= & \Exp(\Z^{\gr}) + 
	\underbrace{ \frac{ 1 }{ 2 }
		\frac{ 1 }{ \mu + 1 }
		\sum_{ m \, = \, 3 }^{ \mu + 1 } 
		\frac{ ( \mu + 1 - m  ) ( \mu + 2 - m ) ( m - 2 ) }{ \delta^{ m - 1 } } 
		{ \mu + 1 \choose m } }_{ \defining \Sigma }
	\end{split}
	\end{equation*}
	In the expression above, the fourth equality uses the identities $ \mu = \delta + 1 $ and \eqref{Eq Binomial Coefficients Indexes Shift}. The fifth equality follows from reindexing $ m = s + 1 $; here we also denote the second summand term by $ \Sigma $. Next, we focus on deriving a closed form for $ \Sigma $, to that end, we appeal to the polynomial identity
	\begin{multline*}
	( \mu + 1 - m  ) ( \mu + 2 - m ) ( m - 2 ) 
	\\
	=
	m ( m - 1 ) ( m - 2 ) - 2 ( \mu + 1 ) m ( m - 1 ) +
	( \mu + 4 ) ( \mu + 1) m 
	- 2 ( \mu + 2 ) ( \mu + 1 ) .
	\end{multline*}
	Replacing the latter in the second summand $ \Sigma $, it transforms in
	\begin{equation*}
	\begin{split}
	\Sigma \equiv &
	\frac{ 1 }{ 2 \delta^{ 2 } }
	\frac{ 1 }{ \mu + 1 }
	\sum_{ m \, = \, 3 }^{ \mu + 1 } 
	\frac{m ( m - 1 ) ( m - 2 ) }{ \delta^{ m - 3 } } 
	{ \mu + 1 \choose m } 
	- 
	\frac{ 1 }{ \delta }
	\sum_{ m \, = \, 3 }^{ \mu + 1 } 
	\frac{m ( m - 1 )}{ \delta^{ m - 2 } } 
	{ \mu + 1 \choose m } 
	\\
	& + \frac{ \mu + 4 }{ 2 }
	\sum_{ m \, = \, 3 }^{ \mu + 1 } 
	\frac{ m }{ \delta^{ m - 1 } } 
	{ \mu + 1 \choose m }
	- \delta ( \mu + 2 )
	\sum_{ m \, = \, 3 }^{ \mu + 1 } 
	\frac{ 1 }{ \delta^{ m } } 
	{ \mu + 1 \choose m } 
	\\
	= &  \frac{ \mu }{ 2 \delta }
	\big( 1 + \frac{ 1 }{ \delta } \big)^{ \mu - 2 }
	- 
	\frac{ ( \mu + 1 ) \mu  }{ \delta }
	\Big\{
	( 1 + \frac{ 1 }{ \delta }\big)^{ \mu - 1 }
	- 1
	\Big\}
	+ \frac{ ( \mu + 4 ) ( \mu + 1 ) }{ 2 }
	\Big\{
	\big( 1 + \frac{ 1 }{ \delta } \big)^{ \mu }
	- 1 - 
	\frac{ \mu }{ \delta }
	\Big\}
	\\
	& 
	- \delta ( \mu + 2 )
	\Big\{
	\big( 1 + \frac{ 1 }{ \delta } \big)^{ \mu + 1 }
	- 1 - \frac{ \mu + 1 }{ \delta}   - \frac{ ( \mu+ 1 ) \mu }{ 2 \delta^{ 2 } }  
	\Big\}
	\\
	= &  \frac{ \delta + 1 }{ 2 \delta }
	\big( 1 + \frac{ 1 }{ \delta } \big)^{ \delta - 1 }
	- 
	\frac{ ( \delta + 2 ) ( \delta + 1 )  }{ \delta }
	\Big\{
	( 1 + \frac{ 1 }{ \delta }\big)^{ \delta }
	- 1
	\Big\}
	+ \frac{ ( \delta + 5 ) ( \delta + 2 ) }{ 2 }
	\Big\{
	\big( 1 + \frac{ 1 }{ \delta } \big)^{ \delta - 1 }
	- 1 - 
	\frac{ \delta + 1 }{ \delta }
	\Big\}
	\\
	& 
	- \delta (\delta + 3 )
	\Big\{
	\big( 1 + \frac{ 1 }{ \delta } \big)^{ \delta + 2 }
	- 1 - \frac{ \delta + 2 }{ \delta}   - \frac{ ( \delta + 2 ) ( \delta + 1 ) }{ 2 \delta^{ 2 } }  
	\Big\}
	.
	\end{split}
	\end{equation*}
	Here, the second equality was attained using Newton's binomial identity, together with its first three derivatives. The last equality was attained by replacing $ \delta = \mu - 1 $. Performing further simplifications we get the first equality in the identity \eqref{Eq Linear Relaxation Expected Profit} and replacing \eqref{Eq Greedy Expected Profit} in it, we obtain the second equality.
\end{proof}
\begin{definition}\label{Def Linear Programming Post-G}
	Define the \textbf{post-greedy profit} random variable, associated with the 0-1RLPK, as
	\begin{equation}\label{Eq Linear Programming Post-G}
	\Y^{\lp} \defining \upk \cdot \upg(\ups) .
	\end{equation}
\end{definition}
\begin{theorem}[Asymptotic Relations]\label{Thm Split Item Asymptotic}
	Let $ \ups, \upk, \Z^{\gr} $ and $ \Z^{\lp} $ be the random variables defined so far, then the following limits hold
	\begin{subequations}
		\begin{align}
		\lim\limits_{ \delta \, \rightarrow \, \infty } 
		\Exp\big( \ups\big) & = e , 
		\label{Stmt Asymptotic Split Item}\\
		\lim\limits_{ \delta \, \rightarrow \, \infty } 
		\Var\big( \ups\big)  
		& = e \big(3e - e \big) 
		\label{Stmt Asymptotic Variance}, \\
		\lim\limits_{ \delta \, \rightarrow \, \infty }  \frac{ \Exp\big( \upk \big) }{ \delta } &= 3 - e ,
		\label{Eqn Relative Slack}
		\\
		%
		\lim\limits_{ \delta \, \rightarrow \, \infty }\frac{ \Exp(\Z^{\gr}) + \dfrac{ \delta - \Exp( \ups ) + 1 }{ 2 } \Exp(\upk) }
		{ \Exp\big( \Z^{\lp} \big) } & = 1 ,
		\label{Eqn Relative LP Approximation}
		\\
		%
		\lim\limits_{ \delta \, \rightarrow \, \infty }\frac{ \Exp\big( \Z^{\gr} \big) }
		{ \Exp\big( \Z^{\lp} \big) } & = e - 2 ,
		\label{Eqn Greedy LP ratio}
		\\
		%
		\lim\limits_{ \delta \, \rightarrow \, \infty }\frac{ \dfrac{ \delta - \Exp( \ups ) + 1 }{ 2 } \Exp(\upk) }
		{ \Exp\big( \Z^{\lp} \big) } & = 3 - e .
		\label{Eqn Post-G Fraction}
		\end{align}
	\end{subequations}
\end{theorem}
\begin{proof}[Sketch of the proof.]
	An elementary calculation of limits on the corresponding closed formulas developed above gives all the desired results.
\end{proof}
\begin{remark}\label{Rem Split Item Asymptotic}
	Observe that if we approximate $ \Exp( \upk \sum_{ t \, = \, \ups }^{\mu} \upt( t ) ) $ with $ \frac{ \delta - \Exp( \ups ) + 1 }{ 2 } \Exp(\upk) $ then, the expression $ \Exp(\Z^{\gr}) + \frac{ \delta - \Exp( \ups ) + 1 }{ 2 } \Exp(\upk) $ is an approximation of $ \Exp(\Z^{\lp}) $ as the equation \eqref{Eqn Random Linear Relaxation Solution} shows. Hence, the statement \eqref{Eqn Relative LP Approximation} proves that this is a good approximation.
\end{remark}
%
%
\subsection{The expected performance of the eligible-first algorithm}
\label{Sec Full-First and Eligible-First algorithms' expected performance}
%
%
We close this section presenting the computation of the eligible-first algorithm expectation $ \Exp(\Z^{\ef}) $. Given that the proofs are remarkably similar to those presented in the previous section, we only present sketches of them with some important highlights. 
\begin{definition}\label{Def Post-G Packable Item}
	Let $ \upk $ and $ \ups $ be the random variables introduced in \textsc{Definition} \ref{Def Split Item, Slack and Packed Items Random Variables}  
	\begin{enumerate}[(i)]

		\item Let $ E $ be the set defined in \eqref{Eq Eligible First Set}. We say an item $ i \in [\mu] $ is eligible-first \textbf{eF} if it is the least element of the set $ E $ i.e., if it is the first eligible item, once the greedy algorithm has stopped packing items. 
		
		\item For the eligible-first algorithm, we define its corresponding \textbf{post-greedy profit} random variable as follows
		\begin{align}\label{Eq Full-First, Eligible-First Increment}
		& \Y^{\ef} = 
		\begin{cases}
		\upp(i) & i \text{ is } \ef ,\\
		0 , & E = \emptyset .
		\end{cases}
		\end{align}

	\end{enumerate}
\end{definition}
\begin{lemma}\label{Thm Full-First, Eligible First, Conditional Expectations}
	With the definitions above we have
	%
	%
	\begin{subequations}
		\begin{align}
		\prob\big( i \text{ is } \ef, \upk = k, \ups = s \big) & = \frac{ \delta - k }{\delta^{ s + 1 }} k \big( 1 - \frac{ k }{ \delta } \big)^{ i - s - 1 }
		{\delta - k - 1\choose s - 2} , 
		\label{Eq Probabilities Eligible-First i, k, s}
		\end{align}
		for $ i = s + 1 , \ldots, \mu $, $ k = 0, \ldots, \delta - s - 1 $, $ s = 2, \ldots, \mu $.
		\begin{align}
		%
		\begin{split}
		\Exp\big( \Y^{\ef} \big\vert  \upk = k, \ups = s\big) 
		= &
		\frac{ k }{ 4 }( \delta - s + 1 ) 
		\Big\{1 - \big(1 - \frac{ k }{ \delta } \big)^{ \delta - s + 1 } \Big\}
		- \frac{ \delta }{ 4 k }
		\big( 1 - \frac{ k }{ \delta }\big)
		\Big\{ 
		1 - 
		\big(1 + \frac{ \delta - s }{\delta} k \big)
		\big( 1 - \frac{ k }{ \delta }\big)^{ \delta - s }
		\Big\}
		.
		\label{Eq Expectation Eligible-First k, s}
		\end{split}
		\end{align}
	\end{subequations}
	%
	%
\end{lemma}
\begin{proof}[Sketch of the proof.]
	In order to prove \eqref{Eq Probabilities Eligible-First i, k, s} first notice that 
	\begin{equation*}
	\prob\big( i \text{ is } \ef \big\vert \upk = k, \ups = s  \big) = \frac{ k }{ \delta } \big( 1 - \frac{ k }{ \delta } \big)^{ i - s - 1 } , 
	\end{equation*}
	because $ \upw(j) $ must be bigger than $ k $ for $ j = s + 1, \ldots , i -1 $ and $ \upw(i) $ must be less or equal than $ k $. Each of the former events has probability $ 1 - \frac{ k }{ \delta } $, which must take place $ i - s - 1 = ( i - 1 ) - ( s + 1 ) + 1 $ times, while the latter event has probability $ \frac{ k }{ \delta } $. Recalling that $ \prob\big( i \text{ is } \ef , \upk = k , \ups = s \big) = \prob\big( i \text{ is } \ef \big\vert \upk = k , \ups = s \big) \prob\big( \upk = k , \ups = s \big) $ together with the cornerstone identity \eqref{Eq Knapsack Slack Joint Conditioning}, the equation \eqref{Eq Probabilities Eligible-First i, k, s} follows. It is also important to stress that $ i $ is $ \ef $ only if $ i > s $.
	
	For the proof of identity \eqref{Eq Expectation Eligible-First k, s} observe that \newline 
	$ \Exp\big( \upw(i) \big\vert i \text{ is } \ef, \upk = k , \ups = s \big) = \frac{ k }{ 2 } $, because the event $ [ \, i \text{ is } \ef \, ] $ implies the event $ [ \, \upw(i) \leq k \, ] $. Hence,
	\begin{equation*}
	\begin{split}
	\Exp\big(\upp(i) \big\vert i \text{ is } \ef,  \upk = k, \ups = s  \big) = & 
	\Exp\big(\upg(i) \upw(i) \big\vert i \text{ is } \ef, \upk = k, \ups = s  \big) 
	\\
	= & 
	\Exp\big(\upg(i) \big\vert i \text{ is } \ef, \upk = k, \ups = s  \big) 
	\Exp\big(\upw(i) \big\vert i \text{ is } \ef, \upk = k, \ups = s  \big) 
	\\
	= & \frac{ \mu - i + 1 }{ 2 } \frac{ k }{ 2 } .
	\end{split}
	\end{equation*}
	From here, we get the identity \eqref{Eq Expectation Eligible-First k, s} using the same preivous reasoning.
\end{proof}
\begin{theorem}[Expected values of $ \Z^{\ef} $]\label{Thm Expectation Full-First, Eligible First}
	With the definitions above, the following expectation holds
	%
	%
	%
	\begin{equation}
		\begin{split}
	\Exp( \Z^{\ef} ) = 
	\Exp( \Z^{\gr} ) 
	& + 
	\sum\limits_{ s \, = \, 2 }^{ \mu }
	\sum\limits_{ k \, = \, 0 }^{ \mu - s + 1 }
	\frac{ k }{ 4 }( \delta - s + 1 ) 
	\Big\{1 - \big(1 - \frac{ k }{ \delta } \big)^{ \delta - s + 1 } \Big\}
	\frac{ \delta - k }{ \delta^{ s } } { \delta - k - 1 \choose s - 2 }
	\\
	& - 
	\sum\limits_{ s \, = \, 2 }^{ \mu }
	\sum\limits_{ k \, = \, 0 }^{ \mu - s + 1 }
	\frac{ \delta }{ 4 k }
	\big( 1 - \frac{ k }{ \delta }\big)
	\Big\{ 
	1 - 
	\big(1 + \frac{ \delta - s }{\delta} k \big)
	\big( 1 - \frac{ k }{ \delta }\big)^{ \delta - s }
	\Big\}
	\frac{ \delta - k }{ \delta^{ s } } { \delta - k - 1 \choose s - 2 }
	.
	\label{Eq Expectation Eligible-First}
		\end{split}
	\end{equation}
	%
	%
	Here $ \Z^{ \ef} $ is the value of the objective function furnished by the eligible-first algorithm, introduced in \textsc{Definition} \ref{Def The split item} part (iv).
\end{theorem}
\begin{proof}[Sketch of the proof.]
	Recalling 
	\begin{equation*}
	\begin{split}
	\Exp(\Y^{\ef}) = &
	\sum\limits_{ s \, = \, 2 }^{ \mu }
	\sum\limits_{ k \, = \, 1 }^{ \mu - ( s - 1 ) }
	\Exp\big(\Y^{\ef} \big\vert \upk = k, \ups = s \big) 
	\prob\big( \upk = k, \ups = s \big) ,
	\end{split}
	\end{equation*}
	together with the fundamental identity \eqref{Eq Knapsack Slack Joint Conditioning}, the equation \eqref{Eq Expectation Eligible-First} follows.
\end{proof}
%
%
%
%
\begin{corollary}[Approximation of $ \Exp(\Z^{\ef}) $]\label{Thm Approximation Expectation Full-First, Eligible First}
	With the definitions above, the following estimate holds
	%
	%
	%
	%
	\begin{equation}\label{Eq Approximation Expectation Eligible-First}
		\begin{split}
	\Exp( \Z^{\ef} ) \sim 
	\ef(\delta) \defining 
	\Exp( \Z^{\gr} ) 
	+ & \frac{ \Exp( \upk ) }{ 4 }( \delta - \Exp( \ups ) + 1 ) 
	\Big\{1 - \big(1 - \frac{ \Exp( \upk ) }{ \delta } \big)^{ \delta - \Exp( \ups ) + 1 } \Big\}
	\\
	- & \frac{ \delta }{ 4 \Exp( \upk ) }
	\big( 1 - \frac{ \Exp( \upk)  }{ \delta }\big)
	\Big\{ 
	1 - 
	\big(1 + \frac{ \delta - \Exp( \ups ) }{\delta} \expk \big)
	\big( 1 - \frac{ \Exp( \upk)  }{ \delta }\big)^{ \delta - \Exp( \ups ) }
	\Big\}
	.
		\end{split}
	\end{equation}
	%
	%
	%
\end{corollary}
\begin{proof}
	Let $ k_{0} \defining \big\lfloor \Exp(\upk)\big\rfloor $ and $ s_{0} \defining \big\lceil \Exp(\ups)\big\rceil $ and notice the approximation
	\begin{equation*}
	\begin{split}
	\Exp( \Y^{\ef} ) \sim & \Exp\big( \Y^{\ef}\big\vert \upk = k_{0}, \ups = s_{0}  \big)\\
	\sim &  
	\frac{ k_{0} }{ 4 }( \delta - s_{0} + 1 ) 
	\Big\{1 - \big(1 - \frac{ k_{0} }{ \delta } \big)^{ \delta - s_{0} + 1 } \Big\} 
	- \frac{ \delta }{ 4 k_{0} }
	\big( 1 - \frac{ k_{0} }{ \delta }\big)
	\Big\{ 
	1 - 
	\big(1 + \frac{ \delta - s_{0} }{\delta} k_{0} \big)
	\big( 1 - \frac{ k_{0} }{ \delta }\big)^{ \delta - s_{0} }
	\Big\}
	\\
	\sim & 
	\frac{ \expk }{ 4 }( \delta - \exps + 1 ) 
	\Big\{1 - \big(1 - \frac{ \expk }{ \delta } \big)^{ \delta - \exps + 1 } \Big\}
	\\
	&
	- \frac{ \delta }{ 4 \expk }
	\big( 1 - \frac{ \expk }{ \delta }\big)
	\Big\{ 
	1 - 
	\big(1 + \frac{ \delta - \exps }{\delta} \expk \big)
	\big( 1 - \frac{ \expk }{ \delta }\big)^{ \delta - \exps }
	\Big\}
	.
	\end{split}
	\end{equation*}
	Here, the first approximation follows by assuming that $ \upk, \ups $ are constant and equal to $ k_{ 0 } $, $ s_{ 0 } $ respectively. The second line follows from the equality \eqref{Eq Expectation Eligible-First k, s} and the third line follows by merely replacing back $ k_{ 0 } , s_{ 0 } $ by the corresponding expected values $ \Exp( \upk ) $ and $ \Exp( \ups ) $ respectively. Next, recalling that $ \Exp( \Z^{\ef}) = \Exp( \Z^{\gr} ) + \Exp( \Y^{\ef} ) $ and using the approximation above, the estimate \eqref{Eq Approximation Expectation Eligible-First} follows. 
\end{proof}
\begin{remark}\label{Rem Approximation Expectation Full-First, Eligible-First}
	Observe that we denote the approximation $ \ef(\delta) $, as a function depending only on the capacity $ \delta $. This is a correct statement because $ \Exp(\ups) $ and $ \Exp(\upk) $ are both functions, exclusively depending on $ \delta $ as the equations \eqref{Eq Split Expectation} and \eqref{Eq Knapsack Slack Expectation} show.
\end{remark}
%
%
%
%
%
%
\section{Probabilistic Analysis of a D\&C Pair}\label{Sec Probabilistic Estimates of the Divide-and-Conquer Algorithm}
%
%
%
%
With the current probabilistic setting it is not possible to get exact expressions for the expected value of $ \Z^{ * } $ (not to mention closed formulas), because it is not possible to give explicit expressions for the optimal solution $ z^{*} $ as we were able to attain for $ z^{\gr} $ in  \eqref{Eq Greedy Decision Variables} and $ z^{\lp} $ in \eqref{Eq LR Objective}. Furthermore, it is not possible to give such explicit descriptions for 
$ z^{\fg} $ or even $ z^{\eg} $, therefore we use the greedy algorithm and the eligible-first algorithm introduced in \textsc{Definition} \ref{Def The split item}, to estimate the expected performance of the Divide-and-Conquer method. 
%
%
%
\subsection{Setting the  $ \Pi_{\lt} $ and $ \Pi_{\rt} $ random subproblems}\label{Sec Left and Right Subproblems Capacities}
%
%
For the analysis of the Divide-and-Conquer method, the induced problems  $ \Pi_{\lt} $ and $ \Pi_{\rt} $ must be analyzed independently. To that end we introduce the random setting for each of these problems  
\begin{definition}\label{Def left-right capacities}
	Define the following elements introduced by one iteration of the Divide-and-Conquer method
	\begin{enumerate}[(i)]
		\item The left and right capacity random variables are given by
		\begin{align}\label{Eq left-right capacities}
		\cleft \defining & \sum\limits_{ i \, \odd }^{\ups - 1 } \upw(i) + \Big\lceil \frac{ \upk }{ 2 } \Big\rceil, &
		\cright \defining & \sum\limits_{  i \, \even }^{\ups - 1 } \upw(i) + \Big\lfloor \frac{ \upk }{ 2 } \Big\rfloor . 
		\end{align}
		\item The left and right subproblems are defined by
		\begin{align}
		\Pi_{\lt} & \defining \big\langle \cleft, (\upp(i) )_{ i \in V_{\lt} } , (\upw(i))_{i \in V_{\lt}} \big\rangle, &
		\Pi_{\rt} & \defining \big\langle \cright, (\upp(i) )_{ i \in V_{\rt} } , (\upw(i))_{i \in V_{\rt}} \big\rangle ,
		\end{align}
		with $ V_{ \lt } \defining \{ i \in [\mu]: i \text{ is odd} \} $ and $ V_{ \rt } \defining \{ i \in [\mu]: i \text{ is even} \} $. 
		
		\item We denote by $ \Z^{ \alg }_{\lt}, \Z^{ \alg }_{\rt} $, the corresponding objective function values to $ \Pi_{\lt}, \Pi_{\rt} $ furnished by the algorithms $ \alg = \ast, \gr, \ff, \ef, \lp $. (Recall that the case $ \alg = * $, stands for the optimal solution, i.e., the optimal value generated by an exact algorithm, e.g., dynamic programming.)
		
		\item We denote by $ \Y^{ \ef }_{\lt}, \Y^{ \ef }_{\rt} $, the corresponding post-greedy profit random variables of $ \Pi_{\lt}, \Pi_{\rt} $ respectively, furnished by the algorithms $ \alg = \ef, \lp $ and according to the definitions \ref{Def Post-G Packable Item} (ii) and \ref{Def Linear Programming Post-G} respectively.
		
		\item Denote by $ \sleft $ ($ \sright $), $ \kleft $ ($ \kright $), the split item and the slack of the $ \Pi_{\lt} $ ($ \Pi_{\rt} $) problem. 
		
	\end{enumerate}
\end{definition}
Before proceeding to the next results, we reduce the cases of analysis adopting the next hypothesis. 
\begin{hypothesis}\label{Hyp Analysis of left and right subproblems}
	From now on it will be assumed that $ \mu = 2\lambda $, i.e., the quantity of eligible items is even. In particular, each subproblem $ \Pi_{\lt} $ and $ \Pi_{\rt} $ has $ \lambda $ eligible items.
\end{hypothesis}
\begin{theorem}\label{Thm Balanced capacities}
	Let $ \cleft, \cright $ be the random variables introduced in \textsc{Definition} \ref{Def left-right capacities} above
	\begin{enumerate}[(i)]
		\item If $ s $ is an odd number, then
		\begin{subequations}\label{Eq Balanced Capacities s odd}
			\begin{align}
			\Exp\big(\cleft \big\vert \upk = k , \ups = s\big) 
			= & \frac{ \delta - k }{ 2 } + \big\lceil \frac{ k }{ 2 }\big\rceil ,
			\label{Eq Left Capacity s odd}
			\\
			\Exp\big(\cright \big\vert \upk = k , \ups = s\big)
			= & \frac{ \delta - k }{ 2 } + \big\lfloor \frac{ k }{ 2 } \big\rfloor .
			\label{Eq Right Capacity s odd}
			\end{align}
		\end{subequations}
		\item If $ s $ is an even number, then
		\begin{subequations}\label{Eq Balanced Capacities s even}
			\begin{align}
			\Exp\big(\cleft \big\vert \upk = k , \ups = s\big) 
			= & \frac{ \delta - k }{ 2 } + \big\lceil \frac{ k }{ 2 }\big\rceil + \frac{ 1 }{ 4 } \frac{ \delta - k }{ s - 1 } ,
			\label{Eq Left Capacity s even}
			\\
			\Exp\big(\cright \big\vert \upk = k , \ups = s\big)
			= & \frac{ \delta - k }{ 2 }  + \big\lfloor \frac{ k }{ 2 } \big\rfloor - \frac{ 1 }{ 4 } \frac{ \delta - k }{ s - 1 }.
			\label{Eq Right Capacity s even}
			\end{align}
		\end{subequations}
	\end{enumerate}
\end{theorem}
\begin{proof}
	Recall that it $ \upk = k $ and $ \ups = s $ then $ \sum_{i \, = \, 1}^{s - 1} \upw(i) = \delta - k $ and $ \upw(s) > \delta - k $; hence $ \big(\upw(i)\big)_{i = 1}^{s - 1} $
	is a composition of $ \delta - k $ in $ s - 1 $ parts. 
	\begin{enumerate}[(i)]
		\item 
		If $ s $ is odd, then $ s - 1 $ is even and due to \textsc{Theorem} \ref{Thm Balance Even-Odd sums in compositions} (i) about compositions, it follows that 
		\begin{equation*}
		\begin{split}
		\Exp\big(\sum\limits_{ i \, \odd  }^{\ups - 1 } \upw(i)  \big\vert \upk = k , \ups = s\big) 
		= 
		\Exp\big(\sum\limits_{ i \, \even }^{\ups - 1 } \upw(i)  \big\vert \upk = k , \ups = s\big) .
		\end{split}
		\end{equation*}
		Recalling that $ \Exp\big(\sum_{ i \, = \, 1 }^{\ups - 1 } \upw(i)  \big\vert \upk = k , \ups = s\big)  = \delta - k $, the result follows.
		
		\item If $ s $ is even, then $ s - 1 = 2 \ell + 1 $ is odd and due to \textsc{Theorem} \ref{Thm Balance Even-Odd sums in compositions} (ii) about compositions, it follows that 
		\begin{equation*}
		\begin{split}
		\Exp\big(\sum\limits_{ i \, \odd  }^{\ups - 1 } \upw(i)  \big\vert \upk = k , \ups = s\big) 
		= &
		\Exp\big(\sum\limits_{ i \, \even }^{\ups - 1 } \upw(i)  \big\vert \upk = k , \ups = s\big) 
		+ \frac{ 1 }{ 2 }\frac{ \delta - k }{ 2 \ell + 1 } \\
		= &
		\Exp\big(\sum\limits_{ i \, \even }^{\ups - 1 } \upw(i)  \big\vert \upk = k , \ups = s\big) 
		+ \frac{ 1 }{ 2 }\frac{ \delta - k }{ s - 1 }.
		\end{split}
		\end{equation*}
		The second equality is a mere replacement of $ s = 2 \ell + 2 $. Hence, recalling that \newline $ \Exp\big(\sum_{ i \, = \, 1 }^{\ups - 1 } \upw(i)  \big\vert \upk = k , \ups = s\big)  = \delta - k $ and solving the $ 2 \times 2 $ linear system, the result follows.
	\end{enumerate}
\end{proof}
\begin{lemma}\label{Thm Expected Ceil and Floor}
	Let $ \upk $ be the slack variable introduced in \textsc{Definition} \ref{Def Split Item, Slack and Packed Items Random Variables} then
	\begin{subequations}\label{Eq Expected Ceil and Floor}
		\begin{align}
		\Exp\Big( \Big\lceil \frac{ \upk }{ 2 }\Big\rceil\Big) 
		= &
		\frac{ \Exp( \upk) }{ 2 }
		+ \frac{ 1 }{ 2\delta^{ 2 } }
		\sum\limits_{k \, \even }^{ \delta }
		k
		\big( 1 + \frac{ 1 }{ \delta }\big)^{ k - 1} ,
		\label{Eq Expected Ceil}  \\
		\Exp\Big( \Big\lfloor \frac{ \upk }{ 2 }\Big\rfloor\Big) 
		= &
		\frac{ \Exp( \upk) }{ 2 } 
		-
		\frac{ 1 }{ 2\delta^{ 2 } }
		\sum\limits_{k \, \even }^{ \delta }
		k
		\big( 1 + \frac{ 1 }{ \delta }\big)^{ k - 1} .
		\label{Eq Expected Floor} 
		\end{align}
	\end{subequations}
\end{lemma}
\begin{proof}
	We prove the statement using the definition $ \Exp( \lceil \frac{ \upk }{ 2 }\rceil) 
	= \sum_{k \, = \, 0}^{ \delta }
	\lceil \frac{ k }{ 2 }\rceil \prob(\upk = k)  $. Hence, separating even and odd indexes we get
	\begin{equation*}
	\begin{split}
	\Exp\Big( \Big\lceil \frac{ \upk }{ 2 } \Big\rceil\Big) 
	= & \sum\limits_{ \ell \, = \, 0}^{ \lambda - 1}
	\Big\lceil \frac{ 2\ell }{ 2 }\Big\rceil \prob(\upk = 2\ell)
	+ 
	\sum\limits_{\ell \, = \, 0}^{ \lambda - 1 }
	\Big\lceil \frac{ 2 \ell + 1 }{ 2 }\Big\rceil \prob(\upk = 2 \ell + 1)
	\\
	= & \sum\limits_{ \ell \, = \, 0}^{ \lambda - 1}
	\ell  \prob(\upk = 2\ell)
	+ 
	\sum\limits_{\ell \, = \, 0}^{ \lambda - 1 }
	( \ell + 1 ) \prob(\upk = 2 \ell + 1)
	\\
	= & \sum\limits_{ \ell \, = \, 0}^{ \lambda - 1}
	\frac{ 2\ell }{ 2 }  \prob(\upk = 2\ell)
	+ 
	\sum\limits_{\ell \, = \, 0}^{ \lambda - 1 }
	\frac{ 2\ell + 1}{ 2 }  \prob(\upk = 2 \ell + 1)
	+ 
	\frac{ 1 }{ 2 }\sum\limits_{\ell \, = \, 0}^{ \lambda - 1 }
	\prob(\upk = 2 \ell + 1)
	\\
	= &
	\frac{\Exp( \upk ) }{ 2 }
	+ 
	\frac{ 1 }{ 2 }\sum\limits_{\ell \, = \, 0}^{ \lambda - 1 }
	\prob(\upk = 2 \ell + 1) .
	\end{split}
	\end{equation*}
	Here, the second equality is the computation of the ceiling function $ \lceil \cdot \rceil $, the third equality is a convenient association of terms and the fourth equality merely recovers the expectation of the slack random variable $ \upk $. Next we focus in the last sum,
	\begin{equation*}
	\begin{split}
	\frac{ 1 }{ 2 }\sum\limits_{\ell \, = \, 0}^{ \lambda - 1 }
	\prob(\upk = 2 \ell + 1)
	= & 
	\frac{ 1 }{ 2 }\sum\limits_{k \, \odd }^{ \delta }
	\frac{ \delta - k }{ \delta^{ 2 } }
	\big( 1 + \frac{ 1 }{ \delta }\big)^{ \delta - k - 1}
	= 
	\frac{ 1 }{ 2\delta^{ 2 } }
	\sum\limits_{m \, \even }^{ \delta }
	m
	\big( 1 + \frac{ 1 }{ \delta }\big)^{ m - 1} .
	\end{split}
	\end{equation*}
	%
	Here, the first equality comes from  the identity \eqref{Eq Knapsack Slack Probability Distribution}. The second equality is the reindexing $ m = \delta - k $ and recalling that $ \delta $ and $ k $ are odd, it follows that $ m $ is even. Combining with the previous expression, the identity \eqref{Eq Expected Ceil} follows. 
	
	In order to prove the identity \eqref{Eq Expected Floor}, it suffices to note that $ \lfloor \frac{ \upk }{2} \rfloor = \upk -  \lceil \frac{ \upk }{2} \rceil$ and use \eqref{Eq Expected Ceil} to conclude the result. 
\end{proof}
\begin{theorem}\label{Thm Expected Right and Left Capacities}
	The random variable capacities of the left and right problems have the following expectations
	\begin{subequations}\label{Eq Expected Right and Left Capacities}
		\begin{align}
		\begin{split}
		\Exp(\cleft) = & 
		\frac{ \delta }{ 2 } 
		+ \frac{ 1 }{ 2\delta^{ 2 } }
		\sum\limits_{k \, \even }^{ \delta }
		k
		\big( 1 + \frac{ 1 }{ \delta }\big)^{ k - 1}
		\\
		&
		+ \frac{ \mu }{ 8 }
		\Big\{
		\big( 1 + \frac{ 1 }{ \delta } \big)^{ \mu }
		+
		\big( 1 - \frac{ 1 }{ \delta } \big)^{ \mu }
		\Big\} 
		-
		\frac{ \delta }{ 8 }
		\Big\{
		\big( 1 + \frac{ 1 }{ \delta } \big)^{ \mu + 1 }
		-
		\big( 1 - \frac{ 1 }{ \delta } \big)^{ \mu + 1}
		\Big\} . 
		\label{Eq Capacity Left Expectation}
		\end{split}
		\\
		\begin{split}
		\Exp(\cright) = &
		\frac{ \delta }{ 2 } -
		\frac{ 1 }{ 2\delta^{ 2 } }
		\sum\limits_{k \, \even }^{ \delta }
		k
		\big( 1 + \frac{ 1 }{ \delta }\big)^{ k - 1}
		\\
		&
		- \frac{ \mu }{ 8 }
		\Big\{
		\big( 1 + \frac{ 1 }{ \delta } \big)^{ \mu }
		+
		\big( 1 - \frac{ 1 }{ \delta } \big)^{ \mu }
		\Big\} 
		+
		\frac{ \delta }{ 8 }
		\Big\{
		\big( 1 + \frac{ 1 }{ \delta } \big)^{ \mu + 1 }
		-
		\big( 1 - \frac{ 1 }{ \delta } \big)^{ \mu + 1}
		\Big\} . 
		\label{Eq Capacity Right Expectation}
		\end{split}
		\end{align}
	\end{subequations}
\end{theorem}
\begin{proof}
	We focus on the calculation of $ \Exp( \cleft ) $ using the definition, i.e., 
	\begin{align*}
	\Exp(\cleft) = & \sum\limits_{ s }\sum\limits_{ k } \Exp\big(\cleft \big\vert \upk = k, \ups = s\big)\prob\big( \upk = k, \ups = s\big) .
	\end{align*}
	According to the expressions \eqref{Eq Left Capacity s odd} and \eqref{Eq Left Capacity s even}, there are two paramount parts: the ``head" $ \frac{ \delta - k }{ 2 } + \big\lceil \frac{ k } { 2 } \big\rceil $, present in both cases and the ``tail" $ \frac{ 1 }{ 4 } \frac{ \delta - k }{ s - 1 } $, present only in the case where $ s $ is even. We compute these separately, for the ``head" we recall the cornerstone identity \eqref{Eq Knapsack Slack Joint Conditioning}  and get 
	\begin{equation}\label{Eq Head Left Capacity Expectation}
	\begin{split}
	\sum\limits_{k \, = \, 0}^{ \delta } \sum\limits_{s \, = \, 2}^{ \delta - k + 1 }  
	\Big\{\frac{ \delta - k }{ 2 } + \big\lceil \frac{ k } { 2 } \big\rceil \Big\}
	\frac{ \delta - k }{ \delta^{s} }  { \delta - k - 1 \choose s - 2} 
	= &
	\sum\limits_{k \, = \, 0}^{ \delta } 
	\Big\{\frac{ \delta - k }{ 2 } + \big\lceil \frac{ k } { 2 } \big\rceil \Big\}
	\sum\limits_{s \, = \, 2}^{ \delta - k + 1 }  
	\frac{ \delta - k }{ \delta^{s} } { \delta - k - 1 \choose s - 2}
	\\
	= & \sum\limits_{k \, = \, 0}^{ \delta } 
	\Big\{\frac{ \delta - k }{ 2 } + \big\lceil \frac{ k } { 2 } \big\rceil \Big\} 
	\prob(\upk = k ) 
	\\
	= &
	\frac{ \delta - \Exp(\upk) }{ 2 } + 
	\Exp\Big( 
	\Big\lceil \frac{ \upk } { 2 } \Big\rceil 
	\Big) 
	\\
	= &
	\frac{ \delta }{ 2 } 
	+ \frac{ 1 }{ 2\delta^{ 2 } }
	\sum\limits_{k \, \even }^{ \delta }
	k
	\big( 1 + \frac{ 1 }{ \delta }\big)^{ k - 1}.
	\end{split}
	\end{equation}
	Here, the first equality is direct, the second holds by definition of $ \prob (\upk = k ) $ (see the proof of \eqref{Eq Knapsack Slack Probability Distribution} in \textsc{Lemma} \ref{Thm Knapsack Slack Expectation Properties} ), the third equality holds by definition of expectation and the fourth equality is obtained combining the latter with \eqref{Eq Expected Ceil}. Next we compute the ``tail", recalling the identities \eqref{Eq Knapsack Slack Joint Conditioning} and \eqref{Eq Binomial Coefficients Indexes Shift}, we have
	\begin{equation*}
	\begin{split}
	\sum\limits_{s \, \even }  \sum\limits_{k \, = \, 0}^{ \delta - s + 1}
	\frac{ 1 }{ 4 } \frac{ \delta - k }{ s - 1 } \frac{ \delta - k }{ \delta^{s} } 
	{\delta - k - 1 \choose s - 2}
	= &
	\sum\limits_{s \, \even } 
	\frac{ 1 }{ 4 \delta^{ s } }
	\sum\limits_{k \, = \, 0}^{ \delta - s + 1}
	( \delta - k )
	{\delta - k  \choose s - 1}
	.
	\end{split}
	\end{equation*}
	We focus on the internal sum
	\begin{equation*}
	\sum\limits_{k \, = \, 0}^{ \delta - s + 1}
	( \delta - k )
	{\delta - k  \choose s - 1}
	= 
	s \sum\limits_{k \, = \, 0}^{ \delta - s + 1}
	{\delta - k + 1 \choose s }
	-
	\sum\limits_{k \, = \, 0}^{ \delta - s + 1}
	{\delta - k  \choose s - 1} 
	= 
	s { \delta + 2 \choose s + 1 } - { \delta + 1 \choose s }
	= 
	\frac{ s \mu - 1}{ s + 1 } {\mu \choose s }
	.
	\end{equation*}
	Then, back to the ``tail" term we have
	\begin{equation*}
	\begin{split}
	\sum\limits_{s \, \even }^{ \mu }
	\frac{ 1 }{ 4 \delta^{ s } } \frac{ s \mu - 1}{ s + 1 } {\mu \choose s }
	= & 
	\frac{ 1 }{ 4 ( \mu + 1) }
	\sum\limits_{s \, \even }^{ \mu } 
	\frac{ s \mu - 1}{ \delta^{ s } } {\mu + 1 \choose s + 1 }
	\\
	= & 
	\frac{ 1 }{ 4 ( \mu + 1) }
	\sum\limits_{s \, \even }^{ \mu + 1} 
	\frac{ s \mu - 1}{ \delta^{ s } } {\mu + 1 \choose s + 1 }
	\\
	= &
	\frac{ \mu }{ 4 ( \mu + 1)  }
	\sum\limits_{s \, \even }^{ \mu + 1 } 
	\frac{ s + 1 }{ \delta^{ s } } {\mu + 1 \choose s + 1 }
	-
	\frac{ \delta }{ 4 } 
	\sum\limits_{s \, \even }^{ \mu + 1 }  
	\frac{ 1 }{ \delta^{ s + 1 } } {\mu + 1 \choose s + 1 } .
	\\
	= &
	\frac{ \mu }{ 4 ( \mu + 1)  }
	\sum\limits_{\ell \, \odd }^{ \mu + 1 } 
	\frac{ \ell }{ \delta^{ \ell - 1 } } {\mu + 1 \choose \ell }
	-
	\frac{ \delta }{ 4 } 
	\sum\limits_{\ell \, \odd }^{ \mu + 1 }  
	\frac{ 1 }{ \delta^{ \ell } } {\mu + 1 \choose \ell } .
	\end{split}
	\end{equation*}
	In the expression above, the first equality is the adjustment of the binomial coefficient using the identity \eqref{Eq Binomial Coefficients Indexes Shift}. The second equality extends the upper limit sum from $ \mu $ to $ \mu + 1 $, which can be done without picking up new summands, because we have assumed that $ \mu $ is even and we are adding over $ s $ even. The third equality is a convenient association of terms. Next, recall that 
	\begin{equation*}
	\begin{split}
	F( x ) \defining \sum\limits_{\ell \text{ odd } }^{ n } { n \choose \ell } x^{ \ell } = &
	\frac{(1 + x)^{ n } - (1 - x )^{n}}{2} .
	\end{split}
	\end{equation*}
	Hence, using the function $ F( \cdot ) $ and its first derivative, the tail term gives
	\begin{equation}\label{Eq Tail Left Capacity Expectation}
		\begin{split}
	\sum\limits_{s \, \even } 
	\frac{ 1 }{ 4 \delta^{ s } } \frac{ s \mu - 1}{ s + 1 } {\mu \choose s }
	= &
	\frac{ \mu }{ 8 ( \mu + 1 ) }
	\Big\{
	(\mu + 1) \big( 1 + \frac{ 1 }{ \delta } \big)^{ \mu }
	+
	(\mu + 1) \big( 1 - \frac{ 1 }{ \delta } \big)^{ \mu }
	\Big\} 
	\\
		&
	-
	\frac{ \delta }{ 8 }
	\Big\{
	\big( 1 + \frac{ 1 }{ \delta } \big)^{ \mu + 1 }
	-
	\big( 1 - \frac{ 1 }{ \delta } \big)^{ \mu + 1}
	\Big\}
	\\
	= &
	\frac{ \mu }{ 8 }
	\Big\{
	\big( 1 + \frac{ 1 }{ \delta } \big)^{ \mu }
	+
	\big( 1 - \frac{ 1 }{ \delta } \big)^{ \mu }
	\Big\} 
	-
	\frac{ \delta }{ 8 }
	\Big\{
	\big( 1 + \frac{ 1 }{ \delta } \big)^{ \mu + 1 }
	-
	\big( 1 - \frac{ 1 }{ \delta } \big)^{ \mu + 1}
	\Big\}
	.
		\end{split}
	\end{equation}
	Putting together the ``head" of the sum \eqref{Eq Head Left Capacity Expectation} and the ``tail" \eqref{Eq Tail Left Capacity Expectation}, the identity 
	\eqref{Eq Capacity Left Expectation} follows. 
	
	We compute the expectation of $ \cright $ given by the expression \eqref{Eq Capacity Right Expectation} using the previous procedure but, keeping in mind that the ``tail" \eqref{Eq Tail Left Capacity Expectation}, has to be subtracted rather than added.
\end{proof}
In oder to ease future calculations, we will use the following estimates
\begin{lemma}\label{Thm Distribution Split Item Left Right}
	Let $ \sleft, \sright $ be the splitting item random variable defined above for the problems $ \Pi_{\lt}, \Pi_{\rt} $ then, their expectations satisfy
	\begin{subequations}
		\begin{align}
		\Exp\big( \sright \big\vert \cright = c \big) & =
		\Exp\big( \sleft \big\vert \cleft = c \big)
		= 
		\big( 1 + \frac{ 1 }{ c } \big)^{ c } ,
		\label{Eq Split Expectation Left Right}
		\\
		\Exp\big( \sleft  \big)
		& \sim 
		\big( 1 + \frac{ 1 }{ \expcl } \big)^{ \expcl } ,
		\label{Est Approximation Split Expectation Left}
		\\
		\Exp\big( \sright \big)
		& \sim 
		\big( 1 + \frac{ 1 }{ \expcr } \big)^{ \expcr } .
		\label{Est Approximation Split Expectation Right}
		\end{align}
	\end{subequations}
\end{lemma}
\begin{proof}[Sketch of the proof.]
	The proof of equation \eqref{Eq Split Expectation Left Right} is analogous to the proof of \textsc{Lemma} \ref{Thm Distribution Split Item}, because once $ \cleft $ is known/fixed, the conditional expectations depend strictly on the capacity of the particular 0-1KP, as well as the weight random variables $ \big( \upw(2i -1) \big)_{i = 1}^{\lambda} $, $ \big( \upw(2i) \big)_{i = 1}^{\lambda} $ for $ \Pi_{\lt}$ and $ \Pi_{\rt} $ respectively, whose distribution is uniform and independent from each other.
	
	The estimates \eqref{Est Approximation Split Expectation Left} and \eqref{Est Approximation Split Expectation Right} follow directly using the same reasoning of \textsc{Corollary} \ref{Thm Approximation Expectation Full-First, Eligible First}.
\end{proof}
\begin{lemma}\label{Thm Knapsack Slack Expectation Left Right}
	The slack random variables $ \kleft, \kright $, introduced in \textsc{Definition} \ref{Def Split Item, Slack and Packed Items Random Variables} (ii), satisfy
	\begin{subequations}\label{Eq Left Right Knapsack Slack Expectation Properties}
	\begin{align}
	\begin{split}
	\Exp\big( \kright \big\vert \cright = c\big)  = &
	\Exp\big( \kleft \big\vert \cleft = c\big) 
	\\
	= &  
	- \frac{ ( c + 1 ) }{ c } 
	\Big\{ 
	\big( 1 + \frac{ 1 }{ c } \big)^{ c } - 1
	\Big\} 
	+ ( c + 3 ) 
	\Big\{ 
	\big( 1 + \frac{ 1 }{ c } \big)^{ c + 1 } 
	- \frac{ 2 c + 1 }{ c } 
	\Big\} \\
	& - 2 c 
	\Big\{
	\big( 1 + \frac{ 1 }{ c } \big)^{ c + 2 }
	- \frac{ 5 c^{ 2 } + 7 c + 2 }{ 2 c^{2} }
	\Big\}
	.
	\label{Eq Knapsack Slack Expectation Left Right}
	\end{split}
	\\
	\begin{split}
	\Exp\big( \kleft \big) 
	\sim &  
	- \frac{ ( \expcl + 1 ) }{ \expcl } 
	\Big\{ 
	\big( 1 + \frac{ 1 }{ \expcl } \big)^{ \expcl } - 1
	\Big\} 
	\\
	& 
	+ ( \expcl + 3 ) 
	\Big\{ 
	\big( 1 + \frac{ 1 }{ \expcl } \big)^{ \expcl + 1 } 
	- \frac{ 2 \expcl + 1 }{ \expcl } 
	\Big\} \\
	& - 2 \expcl
	\Big\{
	\big( 1 + \frac{ 1 }{ \expcl } \big)^{ \expcl + 2 }
	- \frac{ 5 \expcl^{ 2 } + 7 \expcl + 2 }{ 2 \expcl^{2} }
	\Big\}
	.
	\label{Est Approximation Knapsack Slack Expectation Left}
	\end{split}
	\\
	\begin{split}
	\Exp\big( \kright \big) 
	\sim &  
	- \frac{ ( \expcr + 1 ) }{ \expcr } 
	\Big\{ 
	\big( 1 + \frac{ 1 }{ \expcr } \big)^{ \expcr } - 1
	\Big\} 
	\\
	& 
	+ ( \expcr + 3 ) 
	\Big\{ 
	\big( 1 + \frac{ 1 }{ \expcr } \big)^{ \expcr + 1 } 
	- \frac{ 2 \expcr + 1 }{ \expcr } 
	\Big\} \\
	& - 2 \expcr 
	\Big\{
	\big( 1 + \frac{ 1 }{ \expcr } \big)^{ \expcr + 2 }
	- \frac{ 5 \expcr^{ 2 } + 7 \expcr + 2 }{ 2 \expcr^{2} }
	\Big\}
	.
	\label{Est Approximation Knapsack Slack Expectation Right}
	\end{split}
	\end{align}
\end{subequations}
\end{lemma}
\begin{proof}[Sketch of the proof.]
	The proof of the equation \eqref{Eq Knapsack Slack Expectation Left Right} is analogous to that of \textsc{Lemma} \ref{Thm Weight Conditional Expectation}, adjusting the arguments presented in the proof of \textsc{Lemma} \ref{Thm Distribution Split Item Left Right}. 
	
	The estimates \eqref{Est Approximation Knapsack Slack Expectation Left} and \eqref{Est Approximation Knapsack Slack Expectation Right} follow directly applying the same reasoning of \textsc{Corollary} \ref{Thm Approximation Expectation Full-First, Eligible First}.
\end{proof}
%
%
%
%
\subsection{Expectations of eligible-first algorithm for the $ \Pi_{\lt} $ and $ \Pi_{\rt} $ subproblems}\label{Sec Left and Right Subproblems Eligible-First}
%
%
%
%
In the present section we compute the conditional expectation of $ \Z^{\gr}_{\lt}, \Y^{\ef}_{\lt} $ and $ \Z^{\gr}_{\rt}, \Y^{\ef}_{\rt} $ with respect to $ \cleft $ and $ \cright $ respectively.
\begin{theorem}\label{Thm Greedy Left Right Expected Profit}
	Let  $ \Pi_{\lt} $, $ \Pi_{\rt} $ be the left and right subproblems introduced in \textsc{Definition} \ref{Def left-right capacities} and let $ \Z^{\gr}_{\lt} $, $ \Z^{\gr}_{\rt} $ be their corresponding solutions furnished by the greedy algorithm. Then, 

\begin{subequations}
	\begin{align}
	\Exp\big( \Z^{\gr}_{\lt}  \big\vert \cleft = c , \sleft = s \big) 
	= &  
	\frac{  \mu - s + 2 }{ 2 }
	\frac{ s c  + s - 1 }{ s + 1 } , 
	\text{ for all } s = 2, \ldots, \mu 
	,
	\label{Eq Conditional Greedy Expected Profit Left} \\
	%
	\begin{split}
	\Exp\big(\Z^{\gr}_{\lt} \big\vert \cleft = c \big) 
	= &  
	- \frac { c }{ 2 } \big( 1 + \frac{ 1 }{ c }\big)^{ c + 1 } 
	-
	\frac{ ( \mu + 3 ) ( c + 2 ) }{ 2 }   
	\Big\{ 
	\big( 1 + \frac{ 1 }{ c }\big)^{ c + 1 }
	- \frac{ c + 1 }{ c }
	\Big\}
	\\
	& 
	+ \frac{ 2\mu c + 6 c + 3\mu + 10 }{ 2 } 
	+
	c ( \mu + 3 )
	\Big\{
	\big( 1 + \frac{ 1 }{ c }\big)^{ c + 2 }
	- 1 - \frac{ c + 2 }{ c } - \frac{ ( c + 1 ) ( c + 2 )}{ 2 c^{ 2 } }
	\Big\}
	,
	\end{split}
	\label{Eq Greedy Expected Profit Left}
	\\
	\Exp\big( \Z^{\gr}_{\rt}  \big\vert \cleft = c , \sleft = s \big) 
	= &  
	\frac{  \mu - s + 1 }{ 2 }
	\frac{ s c  + s - 1 }{ s + 1 } , 
	\text{ for all } s = 2, \ldots, \mu 
	,
	\label{Eq Conditional Greedy Expected Profit Right} \\
	%
	\begin{split}
	\Exp\big(\Z^{\gr}_{\rt} \big\vert \cleft = c \big) 
	= &  
	- \frac { c }{ 2 } \big( 1 + \frac{ 1 }{ c }\big)^{ c + 1 } 
	-
	\frac{ ( \mu + 2 ) ( c + 2 ) }{ 2 }   
	\Big\{ 
	\big( 1 + \frac{ 1 }{ c }\big)^{ c + 1 }
	- \frac{ c + 1 }{ c }
	\Big\}
	\\
	& 
	+ \frac{ 2\mu c + 4 c + 3\mu + 7 }{ 2 } 
	+
	c ( \mu + 2 )
	\Big\{
	\big( 1 + \frac{ 1 }{ c }\big)^{ c + 2 }
	- 1 - \frac{ c + 2 }{ c } - \frac{ ( c + 1 ) ( c + 2 )}{ 2 c^{ 2 } }
	\Big\}
	,
	\end{split}
	\label{Eq Greedy Expected Profit Right}
	\end{align}
\end{subequations}
	with $ \mu = \delta + 1 $.
\end{theorem}
\begin{proof}
	Recall that $ V_{\lt} $ has only odd indexes, then
	\begin{equation*}
	\begin{split}
	\Exp\big( \Z^{\gr}_{\lt}  \big\vert \cleft = c,  \sleft = s \big) = & \sum\limits_{ j \, = \, 1 }^{ s - 1 } \Exp\big( \upp(2j - 1) \big\vert \cleft = c, \sleft = s \big) \\
	= &\sum\limits_{ j \, = \, 1 }^{ s - 1 } \Exp\big(\upw(2j - 1) \, \upg(2j - 1) \big\vert \cleft = c, \sleft = s \big) \\
	= & \sum\limits_{ j \, = \, 1 }^{ s - 1 } \Exp\Big(\upw(2j - 1)  \sum\limits_{t \, = \, 2j - 1}^{\mu}   \upt(t)  \big\vert \cleft = c, \sleft = s \Big) .
	\end{split}
	\end{equation*}
	Recalling that the variables $ \big( \upw(i) \big)_{i = 1}^{\mu} $ and $ \big( \upt(i) \big)_{i = 1}^{\mu} $ are independent, we have
	\begin{equation*}
	\begin{split}
	\Exp\big( \Z^{\gr}_{\lt}  \big\vert  \cleft = c,  \sleft = s \big) 
	& 
	= \sum\limits_{ j \, = \, 1 }^{ s - 1 } \Exp\big( \upw(2j - 1) \big\vert \cleft = c, \sleft = s \big) 
	\sum\limits_{t \, = \, 2j - 1}^{\mu}  \Exp\big( \upt(t) \big\vert \cleft = c, \sleft = s \big)
	\\
	& = \sum\limits_{ j \, = \, 1 }^{ s - 1 } 
	\frac{ c s + s - 1 }{ s^{ 2 } - 1 } 
	\frac{\mu -( 2 j - 1 ) + 1}{2} 
	\\
	& = 
	\frac{ (s - 1 ) ( \mu - s + 2 )}{ 2 }
	\frac{ c s + s - 1 }{ s^{ 2 } - 1 } .
	\end{split}
	\end{equation*}
	Here, the second equality holds due to the identity \eqref{Eq Weight Conditional Expectation}, while the third is its sum. Simplifying the expression above, the \textsc{Equation} \eqref{Eq Conditional Greedy Expected Profit Left} follows. 
	
	Next, in order to prove \eqref{Eq Greedy Expected Profit Left}, observe that due to the expression \eqref{Eq Expectation wrt Partition} we get
	\begin{equation*}
	\Exp\big(\Z^{\gr}_{\lt} \big\vert \cleft = c \big) =
	\sum\limits_{ s \, = \, 2 }^{ c + 1 }  
	\Exp\big( \Z^{\gr}_{\lt}  \big\vert \cleft = c, \sleft = s \big)  \prob\big( \sleft = s \big\vert \cleft = d \big) .
	\end{equation*}
	Combining the equation \eqref{Eq Conditional Greedy Expected Profit Left} with the cornerstone identity \eqref{Eq Knapsack Slack Joint Conditioning} in the expression above, gives
	%
	\begin{equation*}
	\begin{split}
	\Exp\big(\Z^{\gr}_{\lt} 
	\big\vert \cleft = c \big) 
	& = 
	\sum\limits_{ s \, = \, 2 }^{ c + 1 } 
	\frac{  \mu - s + 2 }{ 2 }
	\frac{ c s + s - 1 }{ s + 1 } 
	\frac{  s - 1 }{ c^{ s } }
	{ c + 1 \choose s }
	.
	\end{split}
	\end{equation*}
	From here the closed formula \eqref{Eq Greedy Expected Profit Left} is derived using the same techniques presented in the proof of \textsc{Theorem} \ref{Thm Greedy Expected Profit}.

	Finally, repeating the procedure above, used for the analysis of $ \Z^{ \gr }_{ \lt } $, the equations\eqref{Eq Conditional Greedy Expected Profit Right} and \eqref{Eq Greedy Expected Profit Right}, involving $ \Z^{ \gr }_{ \rt } $ are attained and the result is complete.
\end{proof}
%
%
Observe that \textsc{Theorem} \ref{Thm Greedy Left Right Expected Profit} computes only the conditional expectations. 
In order to find the expectation we should compute,
\begin{subequations}
	\begin{align}\label{Eq Greedy Left Right Expectation}
	\Exp\big( \Z^{\gr}_{\lt} \big)  = & \sum\limits_{ c }\Exp\big( \Z^{\gr}_{\lt} \big\vert \cleft = c \big)  \prob( \cleft = c ) ,\\
	\Exp\big( \Z^{\gr}_{\rt} \big)  = & \sum\limits_{ c }\Exp\big( \Z^{\gr}_{\rt} \big\vert \cright = c \big)  \prob( \cright = c ) .
	\end{align}
\end{subequations}
However, as it has been shown above, that the random variables $ \cleft $ and $ \cright $ are really wild to be used in this calculation (see the proof of \textsc{Theorem} \ref{Thm Expected Right and Left Capacities}). Hence, we adopt the following estimate 
%
%
%
\begin{corollary}\label{Thm Approximation Greedy Left Right Expected Profit}
	Let  $ \Pi_{\lt} $, $ \Pi_{\rt} $ be the left and right subproblems introduced in \textsc{Definition} \ref{Def left-right capacities} and let $ \Z^{\gr}_{\lt} $, $ \Z^{\gr}_{\rt} $ be their corresponding solutions furnished by the greedy algorithm. Then, the following estimates hold
	\begin{subequations}
		\begin{align}
		\begin{split}
		\Exp\big(\Z^{\gr}_{\lt} \big) 
		& \sim   
		- \frac { \expcl }{ 2 } \Big( 1 + \frac{ 1 }{ \expcl }\Big)^{ \expcl + 1 } 
		\\
		& -
		\frac{ ( \mu + 3 ) ( \expcl + 2 ) }{ 2 }   
		\Big\{ 
		\Big( 1 + \frac{ 1 }{ \expcl }\Big)^{ \expcl + 1 }
		- \frac{ \expcl + 1 }{ \expcl }
		\Big\}
		\\
		& 
		+ \frac{ 2\mu \expcl + 6 \expcl + 3\mu + 10 }{ 2 } 
		\\
		& 
		+
		\expcl ( \mu + 3 )
		\Big\{
		\big( 1 + \frac{ 1 }{ \expcl }\big)^{ \expcl + 2 }
		- 1 - \frac{ \expcl + 2 }{ d } - \frac{ ( \expcl + 1 ) ( \expcl + 2 )}{ 2 \expcl^{ 2 } }
		\Big\}
		,
		\end{split}
		\label{Eq Approximation Greedy Expected Profit Left}
		\end{align}
		\begin{align}
		\begin{split}
		\Exp\big(\Z^{\gr}_{\rt} \big) 
		& \sim   
		- \frac { \expcr }{ 2 } \Big( 1 + \frac{ 1 }{ \expcr }\Big)^{ \expcr + 1 } 
		\\
		& 
		-
		\frac{ ( \mu + 2 ) ( \expcr + 2 ) }{ 2 }   
		\Big\{ 
		\Big( 1 + \frac{ 1 }{ \expcr }\Big)^{ \expcr + 1 }
		- \frac{ \expcr + 1 }{ \expcr }
		\Big\}
		\\
		& 
		+ \frac{ 2\mu \expcr + 4 \expcr + 3\mu + 7 }{ 2 } 
		\\
		+
		\expcr  ( \mu + 2 ) &
		\Big\{
		\Big( 1 + \frac{ 1 }{ \expcr }\Big)^{ \expcr + 2 }
		- 1 - \frac{ \expcr + 2 }{ \expcr } - \frac{ ( \expcr + 1 ) ( \expcr + 2 )}{ 2 \expcr^{ 2 } }
		\Big\}
		,
		\end{split}
		\label{Eq Approximation Greedy Expected Profit Right}
		\end{align}
	\end{subequations}
	with $ \mu = \delta + 1 $.
\end{corollary}
\begin{proof}
	The proof follows by approximating $ \cleft \sim \expcl $ and $ \cright \sim \expcr $ in \textsc{Theorem} \ref{Thm Greedy Left Right Expected Profit}.
\end{proof}
Next we compute some convenient conditional expectations of the post-greedy profit random variables $ \Y_{\lt} $ and $ \Y_{\rt} $. 
\begin{theorem}\label{Thm Full-First, Eligible First, Conditional Probabilities and Conditional Expectations Right Left}
	With the definitions above, we have
	%
	%
	\begin{subequations}
		\begin{equation}\label{Eq Probabilities Full-First i, k, s Left}
		\prob\big( 2 i - 1 \text{ is } \lt \ef , \kleft = k, \sleft = s \big\vert \cleft = c \big) 
		= \frac{ c - k }{ c^{ s + 1 }} k \big( 1 - \frac{ k }{  c } \big)^{ i - s - 1 } 
		{ c - k - 1\choose s - 2} , 
		\end{equation}
		\begin{equation}\label{Eq Probabilities Eligible-First i, k, s Right}
		\prob\big( 2i \text{ is } \rt \ef , \kright = k, \sright = s \big\vert \cright = c \big) 
		= \frac{ c - k }{ c^{ s + 1 }} k \big( 1 - \frac{ k }{ c } \big)^{ i - s - 1 }
		{ c - k - 1\choose s - 2} , 
		\end{equation}
		for $ i = s + 1 , \ldots, \lambda $, $ k = 0, \ldots, \delta - s - 1 $, $ s = 2, \ldots, \lambda $.
		\begin{align}
		\begin{split}
		\Exp \big( \Y^{\ef}_{\lt}  \big\vert  \kleft = k, \sleft = s, \cleft = c \big) 
		= &  
		\frac{ k }{ 4 }( \mu - 2s ) 
		\Big\{1 - \big(1 - \frac{ k }{ c } \big)^{ \lambda - s } \Big\}
		\\
		& 
		- \frac{ c }{ 2 k }
		\big( 1 - \frac{ k }{ c }\big)
		\Big\{ 
		1 - 
		\big(1 + \frac{ \lambda - s - 1 }{ c } k \big)
		\big( 1 - \frac{ k }{ c }\big)^{ \lambda - s - 1} 
		\Big\}
		\end{split}
		\label{Eq Conditional Expectation Eligible-First k, s Left}
		\\
		%
		\begin{split}
		\Exp \big( \Y^{\ef}_{\rt}  \big\vert  \kright = k, \sright = s, \cright = c \big) 
		= &  
		\frac{ k }{ 4 }( \mu - 2s - 1 ) 
		\Big\{1 - \big(1 - \frac{ k }{ c } \big)^{ \lambda - s } \Big\}
		\\
		& 
		- \frac{ c }{ 2 k }
		\big( 1 - \frac{ k }{ c }\big)
		\Big\{ 
		1 - 
		\big(1 + \frac{ \lambda - s - 1 }{ c } k \big)
		\big( 1 - \frac{ k }{ c }\big)^{ \lambda - s - 1} 
		\Big\}
		.
		\label{Eq Conditional Expectation Eligible-First k, s Right}
		\end{split}
		\end{align}
	\end{subequations}
	%
	%
	Here $ \Y_{\lt}^{\ef} $ and $ \Y_{\rt}^{\ef} $ are the post-greedy profit random variables introduced in \textsc{Definition} \ref{Def left-right capacities} (iv).
\end{theorem}
\begin{proof}[Sketch of the proof.]
	The result is attained adjusting the procedure used in the proof of \textsc{Lemma} \ref{Thm Full-First, Eligible First, Conditional Expectations}. The identities \eqref{Eq Probabilities Full-First i, k, s Left} and \eqref{Eq Probabilities Eligible-First i, k, s Right} follow directly.
	For the proof of \eqref{Eq Conditional Expectation Eligible-First k, s Left}, we only provide details of the following conditional expectation. Recall that $ \# V_{\lt} = \# V_{\rt}  = \lambda = \frac{ 1 }{ 2 } \mu $, due to the hypothesis \ref{Hyp Analysis of left and right subproblems}; therefore
	%
	\begin{equation*}
	\begin{split}
	\Exp\big( \Y^{\ef}_{\lt} \big\vert \kleft  = k,  \sleft = s,  \cleft = c \big)
	=  &
	\sum\limits_{ i \, = \, s + 1 }^{ \lambda } 
	\Exp\big( \upp(2i - 1) \big\vert 2i - 1 \text{ is } \lt\ef ,\kleft = k, \sleft = s, \cleft = c \big)  \\
	& \times \prob\big( 2i - 1 \text{ is } \lt\ef \big\vert \kleft = k , \sleft = s , \cleft = c \big)
	\\
	= &
	\sum\limits_{ i \, = \, s + 1 }^{ \lambda }
	\frac{ \mu - 2i + 2 }{ 2 } \frac{ k }{ 2 }
	\frac{ k }{ c } \big( 1 - \frac{ k }{ c } \big)^{ i - s - 1 } 
	.
	\end{split}
	\end{equation*}
	Here, the first equality is the mere definition of conditional expectation, while the second equality computes directly the conditional probability of the event inside the sum. From here, solving the sum with the techniques presented in the proof of \textsc{Lemma} \ref{Thm Full-First, Eligible First, Conditional Expectations}, the identity \eqref{Eq Conditional Expectation Eligible-First k, s Left} follows. The proof of the identity \eqref{Eq Conditional Expectation Eligible-First k, s Right} is similar.
\end{proof}
\begin{corollary}[Expected values of $ \Z^{\ef}_{\lt} $ and $ \Z^{\ef}_{\rt} $]\label{Thm Conditional Expectation Full-First, Eligible First Left Right}
	With the definitions above, the following conditional expectations hold
	%
	\begin{subequations}\label{Eq Expectation Full-First, Eligible-First Left Right}
		\begin{equation}\label{Eq Conditional Expectation Capacity Eligible-First Left}
				\begin{split}
		\Exp\big( \Z^{\ef}_{\lt} \big\vert   \cleft = c \big)  = &
		\Exp\big( \Z^{\gr}_{\lt} \big\vert \cleft = c \big) 
		+
		\sum\limits_{ s \, = \, 2 }^{ \lambda }
		\sum\limits_{ k \, = \, 1 }^{ \lambda - s + 1 }
		\frac{ k }{ 4 }( \mu - 2s ) 
		\Big\{1 - \big(1 - \frac{ k }{ c } \big)^{ \lambda - s } \Big\}
		\frac{ c - k }{ c^{ s + 1 }} 
		{ c - k - 1\choose s - 2}
		\\
				& 
		- 
		\sum\limits_{ s \, = \, 2 }^{ \lambda }
		\sum\limits_{ k \, = \, 1 }^{ \lambda  - s + 1 }
		\frac{ 1 }{ 2 k }
		\big( 1 - \frac{ k }{ c }\big)
		\Big\{ 
		1 - 
		\big(1 + \frac{ \lambda - s - 1 }{ c } k \big)
		\big( 1 - \frac{ k }{ c }\big)^{ \lambda - s - 1} 
		\Big\}
		\frac{ c - k }{ c^{ s } } { c - k - 1 \choose s - 2 } ,
				\end{split}
		\end{equation}
		\begin{equation}
				\begin{split}
		\Exp( \Z^{\ef}_{\rt} \big\vert  \cright = c ) 
		= & 
		\Exp( \Z^{\gr}_{\rt} \big\vert \cright = c ) 
		+
		\sum\limits_{ s \, = \, 2 }^{ \lambda }
		\sum\limits_{ k \, = \, 1 }^{ \lambda - s + 1 }
		\frac{ k }{ 4 }( \mu - 2s - 1) 
		\Big\{1 - \big(1 - \frac{ k }{ c } \big)^{ \lambda - s } \Big\}
		\frac{ c - k }{ c^{ s + 1 }} 
		{ c - k - 1\choose s - 2}
		\\
				& 
		- 
		\sum\limits_{ s \, = \, 2 }^{ \lambda }
		\sum\limits_{ k \, = \, 1 }^{ \lambda  - s + 1 }
		\frac{ 1 }{ 2 k }
		\big( 1 - \frac{ k }{ c }\big)
		\Big\{ 
		1 - 
		\big(1 + \frac{ \lambda - s - 1 }{ c } k \big)
		\big( 1 - \frac{ k }{ c }\big)^{ \lambda - s - 1} 
		\Big\}
		\frac{ c - k }{ c^{ s } } { c - k - 1 \choose s - 2 } 
		.
		\label{Eq Conditional Expectation Capacity Eligible-First Right}
				\end{split}
		\end{equation}
	\end{subequations}
	Here $ \Z^{ \ef }_{\lt} $, $ \Z^{ \ef}_{\rt} $ are the corresponding values of the objective function, furnished by the eligible-first algorithm for the problems $ \Pi_{\lt}$ and $ \Pi_{\rt} $. 
\end{corollary}
\begin{proof}[Sketch of the proof.]
	The proof is analogous to the one presented in \textsc{Theorem} \ref{Thm Expectation Full-First, Eligible First}. 
\end{proof}
Finally, we close this section presenting an estimate for the expected performance of the eligible-first algorithm on the $ \Pi_{\lt} $ and $ \Pi_{\rt} $ subproblems.
\begin{corollary}[Approximation of $ \Exp(\Z^{\ef}_{\lt}), \Exp(\Z^{\ef}_{\rt}) $]\label{Thm Approximation Expectation Eligible First Left and Right}
	With the definitions above, the following estimates hold
	\begin{subequations}\label{Eq Approximation Expectation Eligible First Left and Right}
		\begin{align}
		\begin{split}
		\Exp \big( \Z^{\ef}_{\lt} \big) 
		& \sim  
		\Exp \big( \Z^{\gr}_{\lt} \big) 
		+
		\frac{ \expkl }{ 4 }\big( \mu - 2\expsl \big) 
		\Big\{1 - \big(1 - \frac{ \expkl }{ \expcl } \big)^{ \lambda - \expsl } \Big\}
		\\
		& 
		- \frac{ \expcl }{ 2 \expkl }
		\big( 1 - \frac{ \expkl }{ \expcl }\big) 
		\Big\{ 
		1 - 
		\big(1 + \frac{ \lambda - \expsl - 1 }{ \expcl } \expkl \big)
		\big( 1 - \frac{ \expkl }{ \expcl }\big)^{ \lambda - \expsl - 1} 
		\Big\}
		\end{split}
		\label{Eq Approximating Expectation Eligible-First Left}
		\\
		%
		\begin{split}
		\Exp \big( \Z^{\ef}_{\rt} \big) 
		& \sim   
		\Exp \big( \Z^{\gr}_{\rt} \big) 
		+
		\frac{ \expkr }{ 4 }( \mu - 2\expsr - 1 ) 
		\Big\{1 - \big(1 - \frac{ \expkr }{ \expcr } \big)^{ \lambda - \expsr } \Big\}
		\\
		& 
		- \frac{ \expcr }{ 2 \expkr }
		\big( 1 - \frac{ \expkr }{ \expcr }\big) 
		\Big\{ 
		1 - 
		\big(1 + \frac{ \lambda - \expsr - 1 }{ \expcr } \expkr \big)
		\big( 1 - \frac{ \expkr }{ \expcr }\big)^{ \lambda - \expsr - 1} 
		\Big\}
		.
		\label{Eq Approximating Expectation Eligible-First Right}
		\end{split}
		\end{align}
	\end{subequations}
\end{corollary}
\begin{proof}[Sketch of the proof.]
	Similar to the proof of \textsc{Corollary} \ref{Thm Approximation Expectation Full-First, Eligible First}
\end{proof}
%
%
%
%
%
%
%
%
%
\section{Performance Estimates for the Divide-and-Conquer Method}
\label{Sec Performance Estimates for the Divide-and-Conquer}
%
%
%
%
In the current section, we use the previous analysis to derive performance parameters, some for efficiency-reference and other as lower bound estimates for the expected (average) performance of the Divide-and-Conquer method. We also compute with higher accuracy, the performance of the method for the $ \Z^{\ef} $ and $ \Z^{\lp} $ bounding algorithm solutions, to estimate the expected performance of Divide-and-Conquer on the optimal solution $ \Z^{*} $. We begin this section by evaluating numerically the aforementioned parameters for one iteration of the method. 
%
%
\subsection{Expected performance for one iteration of the \newline Divide-and-Conquer Method}
\label{Sec Expected performance for one iteration of the Divide-and-Conquer Method}
%
%
In this section, we finally apply the analytical results previously developed to estimate the performance of one iteration of the Divide-and-Conquer method. Observe that the complexity of the analytical expressions, forces us to seek a numerical evaluation of them in order to attain a tangible value (or reference lower bounds) of the method's efficiency. It is important to stress that for most of the cases, the numerical computations will use the approximations introduced in the
lemmas \ref{Thm Distribution Split Item Left Right}, \ref{Thm Knapsack Slack Expectation Left Right} and the corollaries \ref{Thm Approximation Greedy Left Right Expected Profit}, \ref{Thm Approximation Expectation Eligible First Left and Right} above. This approach is adopted because, the conditional expectations of $ \cleft $ and $ \cright $ with respect to $ \upk $ and $ \ups $ have a wild structure, as they heavily depend on whether the split value is even or odd (see the equations \eqref{Eq Balanced Capacities s odd} and \eqref{Eq Balanced Capacities s even} in \textsc{Theorem} \ref{Thm Balanced capacities}). This case-wise structure makes hard to use the identities \eqref{Eq Balanced Capacities s odd} and \eqref{Eq Balanced Capacities s even} for further calculations beyond the expectations $ \expcl $ and $ \expcr $ (e.g., the equations \eqref{Eq Conditional Expectation Capacity Eligible-First Left} and \eqref{Eq Conditional Expectation Capacity Eligible-First Right}). 

On the other hand, it is important to observe that the approximation $ \ef( \delta )$ for $ \Exp(\Z^{\ef}) $ given in \eqref{Eq Approximation Expectation Eligible-First} (similar to all the estimates adopted) is very accurate with respect to the exact values \eqref{Eq Expectation Eligible-First}, as it can be seen in Table \ref{Tbl Approximation Accuracy} below. Additionally, \textsc{Theorem} \ref{Thm Split Item Asymptotic} shows a convergent asymptotic behavior for the paramount random variables of the 0-1RKP (equation \eqref{Eqn Random Integer Problem}). Furthermore, the statement \eqref{Eqn Relative LP Approximation} in \textsc{Theorem} \ref{Thm Split Item Asymptotic} shows analytically, that the upper bound $ \Exp(\Z^{\lp}) $ can be accurately approximated, as pointed out in \textsc{Remark} \ref{Rem Split Item Asymptotic}, in an analogous way to our approximation $ \Exp(\Z^{\ef}) \sim \ef(\delta) $. 
%
%
\begin{table}[h!]
	\small{
		\begin{centering}
			\begin{tabular}{c | cccccccccccc}
								\rowcolor{gray!80}
				\hline
				\diagbox{Accuracy}{ $ \delta $ }& 
				10 & 20 & 30 & 40 & 50 & 60 & 70 & 80 & 90 & 100 & 110 & 120 \\
				$ 100 \times \dfrac{\ef(\delta)}{ \Exp(\Z^{\ef}) } $ & 2.66 & 0.75 & 0.45 & 0.34 & 0.28 & 0.24 & 0.21 & 0.19 & 0.17 & 0.15 & 0.14 & 0.13 \\
				\hline		
			\end{tabular}
			\caption{Accuracy of the approximation $ \ef(\delta) $. We present the relative accuracy of the approximation in percentage terms for several values of the capacity $ \delta $.}
			\label{Tbl Approximation Accuracy}
		\end{centering}
	}
\end{table}

Hence, the numerical evidence of \textsc{Table} \ref{Tbl Approximation Accuracy}, together with the expected asymptotic behavior, stated in \textsc{Theorem} \ref{Thm Split Item Asymptotic}, are solid grounds to estimate the expected performance of the Divide-and-Conquer method using the approximations \eqref{Eq Approximating Expectation Eligible-First Left} and \eqref{Eq Approximating Expectation Eligible-First Right} for the eligible-first algorithm. Next, we introduce the following set of parameters to estimate the performance of the Divide-and-Conquer. 
\begin{definition}\label{Def DAC Efficiency}
	Let $ \Pi $ be an instance of the 0-1RKP introduced in \textsc{Definition} \ref{Def Random Problems} and let $ \Pi_{\lt} $ and $ \Pi_{ \rt } $ be the problems induced by one iteration of the Divide-and-Conquer method (see \textsc{Definition} \ref{Def Divide and Conquer Setting}). 	Let $ \Exp(\Z^{*}) $ $ \Exp(\Z^{ \ef }) $ and $ \Exp(\Z^{ \lp }) $ be the expected objective function values for the optimal, eligible-first and linear relaxation respectively; moreover the analogous notation holds when the subindex makes reference to the $ \Pi_{\lt} $ or $ \Pi_{ \rt } $ random subproblems. 
	
	\begin{enumerate}[(i)]
		\item Define the following efficiency-reference parameters	
		\begin{subequations}\label{Eq Efficiencies}
			\begin{align}
			\rho \defining & \frac{ \Exp\big(\Z^{ * }_{ \lt }\big) + \Exp\big(\Z^{ * }_{ \rt }\big) }{ \Exp\big(\Z^{ * }\big) } \times 100 , &
			\rho_{ \side } \defining & \frac{ \Exp\big(\Z^{ * }_{ \side }\big) }{ \Exp\big(\Z^{ * }\big) } \times 100 ,
			\label{Eq DAC Efficiency}
			\\
			\rho^{ \ef } \defining & \frac{ \Exp\big(\Z^{ \ef }_{ \lt }\big) + \Exp\big(\Z^{ \ef }_{ \rt }\big) }{ \Exp\big(\Z^{ \ef }\big) } \times 100 , &
			\rho_{ \side }^{ \ef } \defining & \frac{ \Exp\big(\Z^{ \ef }_{ \side }\big) }{ \Exp\big(\Z^{ \ef }\big) } \times 100 ,
			\label{Eq DAC Eligible-First Efficiency}
			\\
			\rho^{ \lp } \defining & \frac{ \Exp\big(\Z^{ \lp }_{ \lt }\big) + \Exp\big(\Z^{ \lp }_{ \rt }\big) }{ \Exp\big(\Z^{ \lp }\big) } \times 100 , &
			\rho_{ \side }^{ \lp } \defining & \frac{ \Exp\big(\Z^{ \lp }_{ \side }\big) }{ \Exp\big(\Z^{ \lp }\big) } \times 100 ,
			\label{Eq DAC Linear Relaxation Efficiency}
			\end{align}
		\end{subequations}
		where $ \side \in \{ \lt, \rt \} $.
		
		\item Define the following lower bound parameters
		\begin{subequations}\label{Eq Lower Bounds}
			\begin{align}\
			\lb^{\gr} \defining & \frac{ \Exp\big(\Z^{ \gr }_{ \lt }\big) + \Exp\big(\Z^{ \gr }_{ \rt }\big) }{ \Exp\big(\Z^{ \lp }\big) } \times 100 , &
			\lb^{\gr}_{ \side } \defining & \frac{ \Exp\big(\Z^{ \gr }_{ \side }\big) }{ \Exp\big(\Z^{ \lp }\big) } \times 100 ,
			\label{Eq LB Greedy DAC Efficiency}
			\\
			\lb^{\ef } \defining & \frac{ \Exp\big(\Z^{ \ef }_{ \lt }\big) + \Exp\big(\Z^{ \ef }_{ \rt }\big) }{ \Exp\big(\Z^{ \lp }\big) } \times 100 , & 
			\lb^{\ef }_{ \side } \defining & \frac{ \Exp\big(\Z^{ \ef }_{ \side }\big) }{ \Exp\big(\Z^{ \lp }\big) } \times 100 ,
			\label{Eq LB Eligible-First DAC Efficiency}
			\end{align}
		\end{subequations}
		where $ \side \in \{ \lt, \rt \} $.
	\end{enumerate}
\end{definition}
It is direct to see that the parameters of equations \eqref{Eq Efficiencies} account for the efficiency of the Divide-and-Conquer method acting on the three solutions $ \Z^{*}, \Z^{\ef} $ and $ \Z^{\lp} $. The analogous holds whenever the subindex $ \side \in \{ \lt, \rt \} $ is present. However, we still need to show that the parameters introduced in the equations \eqref{Eq Lower Bounds} are actually lower bounds.
\begin{proposition}\label{Thm Lower Bounds Hold}
	With the definitions above for the performance parameters, the following estimates hold
	\begin{subequations}
		\begin{align}
		& \lb^{ \gr } \leq \lb^{ \ef } \leq \rho, 
		\label{Ineq Estimate Full Problem}
		\\
		& \lb^{ \gr }_{ \lt } \leq \lb^{ \ef }_{ \lt } \leq \rho_{ \lt }, 
		\label{Ineq Estimate Left Problem}
		\\
		& \lb^{ \gr }_{ \rt } \leq \lb^{ \ef }_{ \rt } \leq \rho_{ \rt } 
		\label{Ineq Estimate Right Problem}.
		\end{align}
	\end{subequations}
\end{proposition}
\begin{proof}
	Recall that due to the algorithms' definition $ \Z^{ * } \leq \Z^{ \lp } $ and  $ \Z^{ \gr }_{ \side } \leq \Z^{ \ef }_{ \side } \leq \Z^{ * }_{ \side } $ for $ \side = \lt, \rt $, for any instance of the problem. Then, $ \Exp(\Z^{ \gr }_{\lt}) + \Exp( \Z^{ \gr }_{ \rt } )
	\leq  \Exp(\Z^{ \ef }_{\lt}) + \Exp(\Z^{ \ef }_{ \rt })
	\leq \Exp(\Z^{ * }_{\lt}) + \Exp(\Z^{ * }_{ \rt }) $, consequently
	\begin{equation*}
	\frac{ \Exp(\Z^{ \gr }_{\lt}) + \Exp( \Z^{ \gr }_{ \rt } ) }{ \Exp(\Z^{ \lp } ) } = \lb^{ \gr }
	\leq  \frac{ \Exp(\Z^{ \ef }_{\lt}) + \Exp(\Z^{ \ef }_{ \rt }) }{ \Exp( \Z^{ \lp } ) } = \lb^{ \ef }
	\leq \frac{ \Exp(\Z^{ * }_{\lt}) + \Exp(\Z^{ * }_{ \rt }) }{ \Exp( \Z^{*} ) } = \rho.
	\end{equation*}
	The above shows the inequality \eqref{Ineq Estimate Full Problem}. The proof of the estimates \eqref{Ineq Estimate Left Problem} and \eqref{Ineq Estimate Right Problem} is analogous.
\end{proof}
Clearly, we want to compute the values of $ \rho, \rho_{\lt} $ and $ \rho_{\rt} $, however, as discussed in \textsc{Remark} \ref{Rem Greedy Algorithms} above, the probabilistic analysis of $ \Z^{*} $ is not tractable (or even $ \Z^{\eg}, \Z^{\fg} $). Hence, we use the values of $ \Z^{\gr}, \Z^{\ef}, \Z^{\lp} $ whose probabilistic analysis has been described accurately enough in the sections \ref{Sec The Expected Performance} and \ref{Sec Probabilistic Estimates of the Divide-and-Conquer Algorithm} above. We analyze the behavior of the Divide-and-Conquer method from two points of view, 
\begin{enumerate}[{view} a.]
	\item We compute the efficiency of the method for $ \Z^{ \ef } $ and $ \Z^{ \lp } $ (equations \eqref{Eq DAC Eligible-First Efficiency} and \eqref{Eq DAC Linear Relaxation Efficiency}) to have an idea of the expected performance of the method for $ \Z^{*} $ (equation \eqref{Eq DAC Efficiency}), see \textsc{Table} \ref{Tbl Eligible-First and Linear Relaxation Performance} and \textsc{Figure} \ref{Fig Performance of LP and EF} below. 
	
	%
	\begin{table}[h!]
		\begin{minipage}{0.5\textwidth}
			\begin{center}
				\scriptsize{
										\rowcolors{2}{white}{gray!25}
					\begin{tabular}{c  c | c  c c | c c c}
												\rowcolor{gray!80}
						\hline
						Capacity & Items & \multicolumn{3}{ c | }{$ \Z^{\ef} $} & \multicolumn{3}{ c }{$ \Z^{\lp} $} \\[2pt] 
												\rowcolor{gray!80}
						$ \delta $ & $ \mu $ & $ \rho^{\ef} $ & $ \rho^{\ef}_{\lt} $ & $ \rho^{\ef}_{\rt} $ & $ \rho^{\lp} $ & $ \rho^{\lp}_{\lt} $ & $ \rho^{\lp}_{\rt} $ \\[2pt]
						\hline 
						49 & 50 & 99.59	& 68.35 & 31.24 & 91.05 & 63.62 & 27.43 \\
						99 & 100 & 99.78 & 68.37 & 31.41 & 91.97 & 64.23 & 27.74 \\
						149 &	150	& 99.85 &	68.38 & 31.48 & 92.28 & 64.44	& 27.84 \\
						199 & 200 & 99.89 &	68.38 & 31.51 & 92.44 & 64.54	& 27.90 \\
						249 & 250 & 99.91 & 68.38	 & 31.53 & 92.53 & 64.6 & 27.93 \\
						299 & 300	& 99.93 &	68.39 & 31.54 & 92.59 & 64.64 & 27.95 \\
						399 & 400 & 99.94 &	68.39 & 31.56 & 92.67 & 64.69	& 27.98 \\
						499 & 500	& 99.96 &	68.39 & 31.57 & 92.72 & 64.72 & 28.00 \\
						599 & 600	& 99.96 &	68.39 & 31.57 & 92.75 & 64.75 & 28.01 \\
						699 & 700	& 99.97 &	68.39 & 31.58 & 92.77 & 64.76	& 28.01 \\
						799 & 800	& 99.97 &	68.39 & 31.58 & 92.79 & 64.77	& 28.02 \\
						899 & 900	& 99.97 &	68.39 & 31.58 & 92.81 & 64.78	& 28.03 \\
						999 &	1000	& 99.98 & 68.39 & 31.59 & 92.82 & 64.79	& 28.03 \\
						\hline		
					\end{tabular}
					\caption{Expected performance of $ \Z^{ \ef } $ and $ \Z^{ \gr } $.}
					\label{Tbl Eligible-First and Linear Relaxation Performance}
				}
			\end{center}
		\end{minipage}
		\begin{minipage}{0.5\textwidth}
			\centering
			{\includegraphics[scale = 0.450]{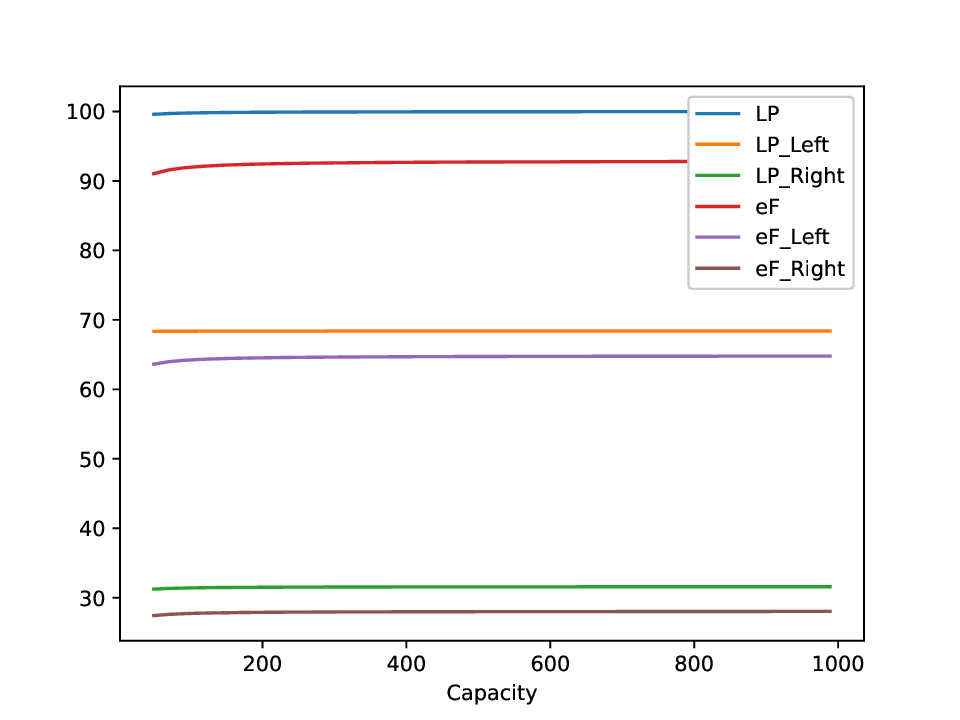} } 
			\captionof{figure}{Expected performance of the Divide-and-Conquer method on $ \Z^{ \lp } $ and $ \Z^{ \ef }$.}
			\label{Fig Performance of LP and EF}
		\end{minipage}
	\end{table}

	\item We compute lower bounds (equations \eqref{Eq Lower Bounds}) for the expected performance of the Divide-and-Conquer method for $ \Z^{ * } $, see \textsc{Table} \ref{Tbl Lower Bounds Greedy and Eligible-First} and \textsc{Figure} \ref{Fig Lower bounds Greedy and Eligible-First} below.  
	
	%
	\begin{table}[h!]
		\begin{minipage}{0.5\textwidth}
			\begin{center}
				\scriptsize{
										\rowcolors{2}{white}{gray!25}
					\begin{tabular}{c  c | c  c c | c c c}
												\rowcolor{gray!80}
						\hline
						Capacity & Items & \multicolumn{3}{ c | }{$ \Z^{\gr} / \Z^{\lp} $} & \multicolumn{3}{ c }{$ \Z^{\ef} / \Z^{\lp} $} \\[2pt] 
												\rowcolor{gray!80}
						$ \delta $ & $ \mu $ & $ \lb^{\gr} $ & $ \lb^{\gr}_{\lt} $ & $ \lb^{\gr}_{\rt} $ & $ \lb^{\ef} $ & $ \lb^{\ef}_{\lt} $ & $ \lb^{\ef}_{\rt} $ \\[2pt]
						\hline 
						49 & 50 & 72.74 & 49.75 & 22.99 & 79.19 & 55.34 & 23.85 \\
						99 & 100 & 72.28 & 49.44 & 22.84 & 79.51 & 55.53 & 23.98 \\
						149 & 150	& 72.13 & 49.33 & 22.79 & 79.62 & 55.59	& 24.02 \\
						199 & 200	& 72.05 & 49.28 & 22.77 & 79.67 & 55.62	& 24.04 \\
						249 & 250	& 72.01 &	49.25 & 22.76 & 79.70 & 55.64 & 24.06 \\
						299 & 300	& 71.98 & 49.23 & 22.75 & 79.72 & 55.66	& 24.07 \\
						399 & 400	& 71.94 &	49.20 & 22.74 & 79.75 & 55.67	& 24.08 \\
						499 & 500	& 71.92 &	49.19 & 22.73 & 79.76 & 55.68	& 24.08 \\
						599 & 600	& 71.90 & 49.18 & 22.72 & 79.77 & 55.69	& 24.09 \\
						699 & 700	& 71.89 &	49.17 & 22.72 & 79.78 & 55.69	& 24.09 \\
						799 & 800	& 71.88 & 49.16 & 22.72 & 79.79 & 55.69	& 24.09 \\
						899 & 900	& 71.88 &	49.16 & 22.72 & 79.79 & 55.70 & 24.09 \\
						999 &1000	& 71.87 &	49.16 & 22.72 & 79.79 & 55.70	& 24.10 \\
						\hline		
					\end{tabular}
					\caption{Lower bounds, ratios $ \Z^{ \gr }/ \Z^{ \lp } $ and $ \Z^{ \ef } / \Z^{ \lp } $.}
					\label{Tbl Lower Bounds Greedy and Eligible-First}
				}
			\end{center}
		\end{minipage}
		\begin{minipage}{0.5\textwidth}
			\centering
			{\includegraphics[scale = 0.450]{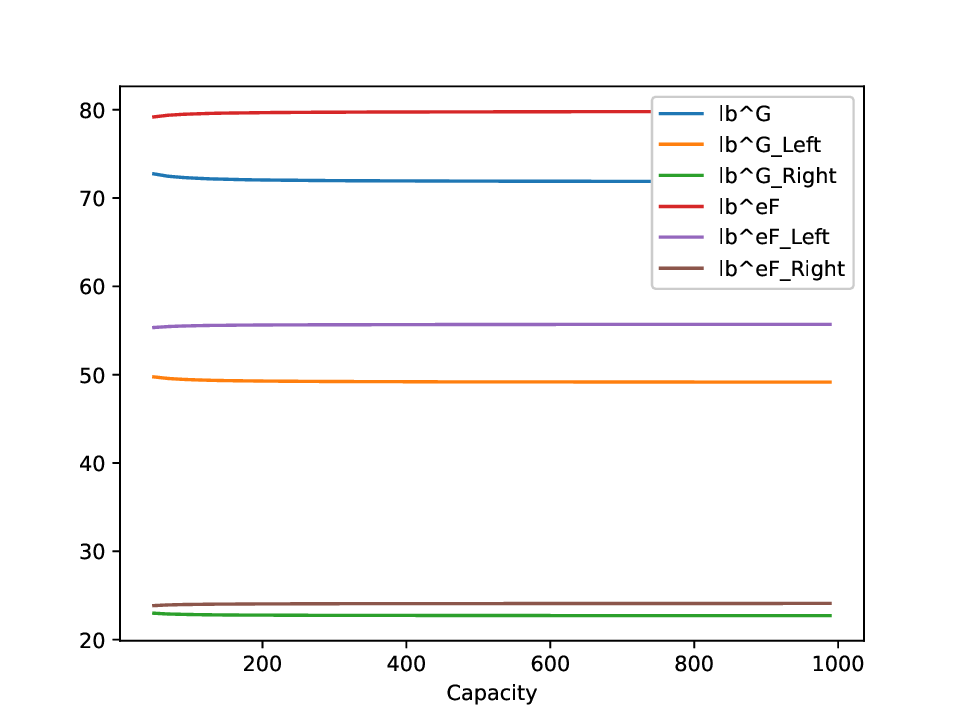} } 
			\captionof{figure}{Lower bounds for the expected performance of the Divide-and-Conquer method on $ \Z^{ \lp } $ and $ \Z^{ \ef }$.}
			\label{Fig Lower bounds Greedy and Eligible-First}
		\end{minipage}
	\end{table}

\end{enumerate}
\begin{remark}\label{Rem Approximation Efficiency and Bounds}
	Strictly speaking we are adopting the following approximations for the performance parameters.
	\begin{equation}\label{Est Approximation Efficiency and Bounds}
	\begin{split}
	\Exp\Big( \frac{ \Z^{ * }_{\lt} + \Z^{ * }_{\rt} }{ \Z^{ * } } \times 100 \Big) \sim & 
	\frac{ \Exp\big(\Z^{ * }_{ \lt }\big) + \Exp\big(\Z^{ * }_{ \rt }\big) }{ \Exp\big(\Z^{ * }\big) } \times 100 = \rho , \\
	\Exp\Big( \frac{ \Z^{ \gr }_{\lt} + \Z^{ \gr }_{\rt} }{ \Z^{ \lp } } \times 100 \Big) \sim &
	\frac{ \Exp\big(\Z^{ \gr }_{ \lt }\big) + \Exp\big(\Z^{ \gr }_{ \rt }\big) }{ \Exp\big(\Z^{ \lp }\big) } \times 100 = \lb^{\gr} ,
	\end{split}
	\end{equation}
	and similarly for all the efficiency (equations \eqref{Eq Efficiencies}) and the lower bound (equations \eqref{Eq Lower Bounds}) parameters that we have introduced in \textsc{Definition} \ref{Def DAC Efficiency}. However, it must be observed that these assumptions are mild as their values are very close to the empirical results. On the other hand, finding the expectation of the left hand side in the estimates \eqref{Est Approximation Efficiency and Bounds} is significantly more complex and provides little extra accuracy. Finally, given that we want to merely estimate the expected efficiency of the Divide-and-Conquer method on the 0-1RKP, it is safe to give up such level of precision.
\end{remark}
%
%
%
%
\subsection{Expected performance for a Divide-and-Conquer Tree}\label{Sec Expected performance for a Divide-and-Conquer Tree}
%
%
In this section we can finally deliver tangible values for the performance of the Divide-and-Conquer method. First for one iteration and then we furnish a method to estimate the expected performance for any D\&C tree (see \textsc{Example} \ref{Exm Labeling DAC Tree} below). 

Observe that for all the parameters introduced in the previous section, the variance is remarkably low. Therefore, we can adopt the averages as the value of the corresponding performance parameters for one iteration of the Divide-and-Conquer method, see \textsc{Table} \ref{Tbl Variance Lower Bounds}. Moreover, due to the low value of the variance, it is safe to assume the same performance of the method through all the iterations of the full binary D\&C tree. 
\begin{table}[h!]
	\begin{center}
		\footnotesize{
			\rowcolors{2}{white}{gray!25}
			\begin{tabular}{c | c  c c | c c c | c c c | c c c}
				\rowcolor{gray!80}
				\hline
				&
				\multicolumn{3}{ c | }{$ \Z^{\ef}  $} &
				\multicolumn{3}{ c | }{$ \Z^{\lp} $} &
				\multicolumn{3}{ c | }{$ \Z^{\gr} / \Z^{\lp} $} & 
				\multicolumn{3}{ c }{$ \Z^{\ef} / \Z^{\lp} $} \\[2pt] 
				\rowcolor{gray!80}
				& $ \rho^{\ef} $ & $ \rho^{\ef}_{\lt} $ & $ \rho^{\ef}_{\rt} $ 
				& $ \rho^{\lp} $ & $ \rho^{\lp}_{\lt} $ & $ \rho^{\lp}_{\rt} $
				& $ \lb^{\gr} $ & $ \lb^{\gr}_{\lt} $ & $ \lb^{\gr}_{\rt} $ 
				& $ \lb^{\ef} $ & $ \lb^{\ef}_{\lt} $ & $ \lb^{\ef}_{\rt} $ \\[2pt]
				\hline 
				mean &
				99.93 & 68.39 & 31.54 & 92.59 & 64.64 & 27.95 & 
				71.98 & 49.23 & 22.75 & 79.72 & 55.65 & 24.07 \\
				variance &
				0.01 & 0.00 & 0.00 & 0.12 & 0.05 & 0.01 & 0.03 &
				0.01 & 0.00 & 0.01 & 0.00 & 0.00 \\
				\hline		
			\end{tabular}
			\caption{Mean and Variance for the performance parameters defined in equations \eqref{Eq DAC Eligible-First Efficiency}, 
				\eqref{Eq DAC Linear Relaxation Efficiency}, \eqref{Eq LB Greedy DAC Efficiency} and \eqref{Eq LB Eligible-First DAC Efficiency}.}
			\label{Tbl Variance Lower Bounds}
		}
	\end{center}
\end{table}

Next, we mark the D\&C tree vertices in a particular way.
%
%
\begin{definition}\label{Def Labelling of Trees}
	Let $ \T $ be a D\&C tree 
	\begin{enumerate}[(i)]
		\item For every vertex $ \Pi $ of $ \T $ we construct a marker $ \mar^{\Pi} $ in the following way. If the vertex is different from the root then the marker is the sequential list of left and/or right turns, that the unique path from the root to it, takes. If the vertex is the root simply assign an empty list as its marker. (See \textsc{Figure} \ref{Fig Divide-and-Conquer labeled tree} in \textsc{Example} \ref{Exm Labeling DAC Tree} below.)
		
		\item Let $ \Pi $ of $ \T $ be a vertex with its corresponding marker $ \mar^{\Pi} $. We define the factor 
		\begin{align}\label{Eq Vertex Performance Factor}
		& \Phi(\Pi) \defining 100 \times \prod\limits_{ i \, = \, 1 }^{ \vert \text{length } \Pi  \vert } \frac{ 1 }{ 100 } \, \Phi_{ \mar^{\Pi} ( i ) } , &
		& \Phi \in \{ \rho^{\ef}, \rho^{\lp}, \lb^{\gr}, \lb^{\ef} \} ,
		\end{align}
		with the convention that $ \varphi( \text{root} ) = 1 $. 
		
		\item The value of the performance parameter of the tree  $ \T $ is given by
		\begin{align}\label{Eq Tree Performance Factor}
		& \Phi(\T) \defining \max \Big\{ 50, \sum\limits_{ L \text{ is a leave of } \T } \Phi( L ) , \Big\}&
		& \Phi \in \{ \rho^{\ef}, \rho^{\lp}, \lb^{\gr}, \lb^{\ef} \} .
		\end{align}
	\end{enumerate}
\end{definition}
\begin{remark}\label{Rem Labelling of Trees}
	We observe the following
	\begin{enumerate}[(i)]
		\item 
		Due to \textsc{Definition} \ref{Def Divide and Conquer Setting} (iii), every internal vertex of a D\&C tree has exactly two children: left and right. Therefore, the marking process is well-defined because, given any arbitrary vertex of the tree, all its ancestors excepting the root, are necessarily the left or right child of its parent. 
		
		\item The marking of vertices is completely analogous to the well-know binary expansion of numbers in the interval $ [0, 1] $. Hence,  a vertex can be very well identified with its marking list. In particular, the length of the marking list is the depth of the vertex.
		
		\item In the expression \eqref{Eq Vertex Performance Factor} each of the percentages is switched to the real number fractions, so that they can be multiplied properly. We set back to the percentage format once the product is executed. In contrast, the expression \eqref{Eq Tree Performance Factor} does not need these precautions because its definition only involves sums. 
		
		\item The computation of $ \Phi(\T) $ involves a maximum between the derived algebraic expression and a 50\% value. This is due to the quality certificate of 50\% in the worst case scenario presented in \textsc{Theorem} \ref{Thm Worst Case Scenario} (ii).
		
	\end{enumerate}
\end{remark}
\begin{example}[Continuation of \textsc{Example}\ref{Exm 0-1KP and DC tree Asymetric}]\label{Exm Labeling DAC Tree} 
	We compute the performance parameters for the D\&C tree presented in \textsc{Example} \ref{Exm 0-1KP and DC tree Asymetric} above. The figure \ref{Fig Divide-and-Conquer labeled tree} depicts the marking of each of the vertices of the tree, while \textsc{Table} \ref{Tbl Example DAC performance} summarizes the values of the four efficiency parameters introduced in \textsc{Section} \ref{Sec Expected performance for one iteration of the Divide-and-Conquer Method} above.
	\begin{table}[h!]
		\begin{minipage}{0.5\textwidth}
			\begin{center}
				\small{
										\rowcolors{2}{white}{gray!25}
					\begin{tabular}{c | c  c  | c  c}
												\rowcolor{gray!80}
						\hline
						Element & \multicolumn{2}{ c | }{Efficiency} & \multicolumn{2}{ c  }{Lower Bound} \\[2pt] 
												\rowcolor{gray!80}
						& $ \rho^{\ef} $ & $ \rho^{\lp} $ & $ \lb^{\gr} $ & $ \lb^{\ef} $ \\[2pt]
						\hline 
						$ \Pi_{0} $ & 1 & 1 & 1 & 1  \\
						$ \Pi_{1} $  & 68.39 & 64.64 & 49.20 & 55.65 \\
						$ \Pi_{2} $  & 46.77 & 41.78 & 24.21 & 30.97 \\
						$ \Pi_{3} $  & 21.57 & 18.06 & 11.20 & 13.95 \\
						$ \Pi_{4} $  & 31.54 & 27.95 & 22.75 & 25.07 \\
						\hline
												\rowcolor{gray!80}
						$ \T  $         & 99.88 & 87.79 & 58.16 & 69.99 \\
						\hline		
					\end{tabular}		
					\caption{Performance parameters \textsc{Example} \ref{Exm Labeling DAC Tree}.}
					\label{Tbl Example DAC performance}
				}
			\end{center}
		\end{minipage}
		\begin{minipage}{0.5\textwidth}
			\centering
			{
				\begin{tikzpicture}
				[scale=.6,auto=left,every node/.style={}]
				\node (n0) at (6,5) {$ \stackrel{ \bullet }{ \Pi_{0} , \, \mar_{0} =  ( \, ) }  $};
				\node (n1) at (2,2.5) {$ \stackrel{ \bullet }{ \Pi_{1} , \, \mar_{1} = (l) } $};
				\node (n2) at (0,0)  {$ \stackrel{\bullet }{ \Pi_{2} , \, \mar_{2} = (l, l) } $};
				\node (n3) at (4, 0)  {$ \stackrel{ \bullet }{ \Pi_{3} , \, \mar_{3}  = (l , r) } $};
				\node (n4) at (10,2.5)  {$ \stackrel { \bullet }{ \Pi_{4} , \, \mar_{4} = (r) } $};
				
				\foreach \from/\to in {n0/n1, n1/n2,n1/n3, n0/n4
				}
				\draw[thick, ->] (\from) -- (\to);
				\end{tikzpicture}
			}
			
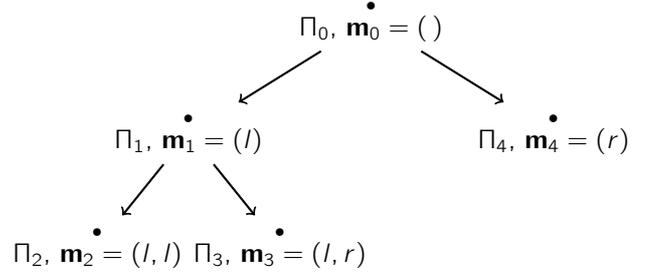
\captionof{figure}{Divide-and-Conquer labeled tree. For each vertex $ \Pi_{i} $ the path from the root to the vertex is indicated as a sequence of left and right turns and denoted as its marking $ \mar_{i} $. (The notation $ \mar^{\Pi_{i}} $ of \textsc{Definition} \ref{Def Labelling of Trees} is omitted for visual purposes.)}
			\label{Fig Divide-and-Conquer labeled tree}
		\end{minipage}
	\end{table}
\end{example}
%
%
%
%
\subsection{Empirical Verification of the Results}\label{Sec Empirical Verification of the Results}
%
%
In the current section we describe the numerical verification of the results presented so far. First, we need to define a number of trials in our experiments, to that end we recall the following result on confidence
\begin{theorem}\label{Them Confidence Interval}
	Let $ x $ be a scalar statistical variable with mean $ \bar{x} $, variance $ \sigma^{2} $.
	\begin{enumerate}[(i)]
		\item The number of trials necessary to get a 95\% confidence interval is given by 
		\begin{equation}\label{Eq Bernoulli Trials for Confidence Intervals}
		n \defining \big( \frac{1.96}{0.05}\big)^{2} \sigma^{2}.
		\end{equation}

		\item The \textbf{95 percent confidence interval} is given by 
		\begin{equation}\label{Eq Confidence Interval Real-Valued}
		I_{x} \defining \Bigg[ \bar{x} - 1.96 \, \sqrt{ \frac{\sigma^{2}}{n} }, \bar{x} + 1.96 \,\sqrt{ \frac{\sigma^{2}}{n} } \, \Bigg] .
		\end{equation}
	\end{enumerate}
\end{theorem}
\begin{proof}
	The proof is based on the Central Limit Theorem, see \cite{Thompson} for details.
\end{proof}
Next, we summarize the guidelines for the experiments design
\begin{enumerate}[a.]
	\item The split index variable $ \ups $ is used to determine the number of trials for our numerical experiments, because we have an analytical expression for its variance given by \textsc{Equation} \ref{Eq Split Variance}. 
	
	\item For simplicity, the sizes of the 0-1RKP's for which the theoretical results are to be verified have the structure $ \delta = 2^{j} - 1 $. These sizes, together with their corresponding number of trials, using the equations \eqref{Eq Split Variance} and \eqref{Eq Bernoulli Trials for Confidence Intervals} are summarized in the table \ref{Tbl Summary of Experiments and Bernoulli Trials} below.
	\begin{table}[h!]
		\centering
		\begin{minipage}{0.45\textwidth}
			\begin{center}
				\normalsize{
										\rowcolors{2}{white}{gray!25}
					\begin{tabular}{c  c  c  c }
												\rowcolor{gray!80}
						\hline
						Capacity & Items & Variance & Trials\\[2pt] 
												\rowcolor{gray!80}
						$ \delta $ & $ \mu $ & $ \Var(\ups) $ & $ n $ \\[2pt]
						\hline 
						63 & 64 &	0.7329 & 1127 \\
						127 & 128 & 0.7493 & 1152 \\
						255 & 256	& 0.7575 & 1165 \\
						511 &	512	& 0.7616 & 1171 \\
						1023 & 1024 & 0.7637 & 1174 \\
						\hline		
					\end{tabular}		
					\caption{Summary of Experiments and Number of Trials}
					\label{Tbl Summary of Experiments and Bernoulli Trials}
				}
			\end{center}
		\end{minipage}
		\quad\quad\quad
		\begin{minipage}{0.45\textwidth}
			\begin{center}
				\normalsize{
										\rowcolors{2}{white}{gray!25}
					\begin{tabular}{c  c c  c }
												\rowcolor{gray!80}
						\hline
						Tree & Height & Number of \\[2pt] 
												\rowcolor{gray!80}
						$ \T $ & $ h $ & Nodes \\[2pt]
						\hline 
						1 & 1 & 3 \\
						2 & 2 & 6 \\
						3 & 3 & 14  \\
						4 & 4 & 30 \\
						\hline		
					\end{tabular}		
					\caption{Summary of Tree Structures}
					\label{Tbl Summary of Tree Structures}
				}
			\end{center}
		\end{minipage}
	\end{table}

	\item For simplicity, the D\&C tree structures to be evaluated are the complete binary trees of the heights detailed in \textsc{Table} \ref{Tbl Summary of Tree Structures}.
	
	\item Each capacity $ \delta $ of \textsc{Table} \ref{Tbl Summary of Experiments and Bernoulli Trials} is tested through all the D\&C trees of \textsc{Table} \ref{Tbl Summary of Tree Structures}.
	
\end{enumerate}

The table \ref{Tbl Experiments 64 Elements} displays the empirical efficiency results for the first case of \textsc{Table} \ref{Tbl Summary of Experiments and Bernoulli Trials}: knapsack capacity $ \delta = 63 $, number of items $ \mu = 64 $. The remaining experiments of the table \ref{Tbl Summary of Experiments and Bernoulli Trials} yield similar efficiency results to the first case, presented in \textsc{Table} \ref{Tbl Experiments 64 Elements}. The table \ref{Tbl Theoretical Tree Efficiencies} summarizes the theoretical results, computed using the approximation method introduced in \textsc{Definition} \ref{Def Labelling of Trees} and explained in \textsc{Example} \ref{Exm Labeling DAC Tree}. As it can be seen, the empirical results are more favorable than the theoretical results, for all the analyzed trees. (The same holds for all the remaining experiments of the table \ref{Tbl Summary of Experiments and Bernoulli Trials}.)  
\begin{table}[h!]
	\begin{minipage}{0.5\textwidth}
		\begin{center}
			\scriptsize{
								\rowcolors{2}{white}{gray!25}
				\begin{tabular}{c | c c  c  | c  c}
										\rowcolor{gray!80}
					\hline
					Tree & \multicolumn{3}{ c | }{Efficiency} & \multicolumn{2}{ c  }{Lower Bound} \\[2pt] 
										\rowcolor{gray!80}
					$ \T $ & $ \rho $ & $ \rho^{\ef} $ & $ \rho^{\lp} $ & $ \lb^{\gr} $ & $ \lb^{\ef} $ \\[2pt]
					\hline 
					1 & 97.66 & 99.83 & 98.82	& 92.78 &	94.99 \\
					2 & 95.45 & 99.46 & 97.63	& 92.91 &	93.87 \\
					3 & 94.75 & 96.40 & 97.00 & 92.96 &	93.29 \\
					4 & 94.55 & 94.30 & 96.81	& 93.00 & 	93.12 \\
					\hline		
				\end{tabular}		
				\caption{Empirical Tree Efficiencies, $ \delta = 63 $, $ \mu = 64 $, Number of Trials $ n = 1152 $.}
				\label{Tbl Experiments 64 Elements}
			}
		\end{center}
	\end{minipage}
	\quad
	\begin{minipage}{0.5\textwidth}
		\begin{center}
			\scriptsize{
								\rowcolors{2}{white}{gray!25}
				\begin{tabular}{c | c  c  | c  c}
										\rowcolor{gray!80}
					\hline
					Tree & \multicolumn{2}{ c | }{Efficiency} & \multicolumn{2}{ c  }{Lower Bound} \\[2pt] 
										\rowcolor{gray!80}
					$ \T $ & $ \rho^{\ef} $ & $ \rho^{\lp} $ & $ \lb^{\gr} $ & $ \lb^{\ef} $ \\[2pt]
					\hline 
					1 & 99.93 & 92.59 & 71.98	& 79.72 \\
					2 & 99.86 & 85.73 & 51.81	& 63.55 \\
					3 & 99.79 & 79.38 & 50.00 & 50.66 \\
					4 & 99.72 & 73.49 & 50.00 & 50.00 \\
					\hline		
				\end{tabular}		
				\caption{Theoretical Tree Efficiency Estimates. These are constructed based on the values of \textsc{Table} \ref{Tbl Variance Lower Bounds}. }
				\label{Tbl Theoretical Tree Efficiencies}
			}
		\end{center}
	\end{minipage}
\end{table}
\begin{remark}\label{Rem Numerical Experiments}
	It is important to stress that the same set of experiments of \textsc{Table} \ref{Tbl Summary of Experiments and Bernoulli Trials} were used to verify the results developed in this work. For all the random variables involved, its empirical expectation falls into their corresponding confidence interval presented in \textsc{Theorem} \ref{Them Confidence Interval}. The correctness of the developed expressions was verified, using the full knapsack problem for the results of \textsc{Section} \ref{Sec The Expected Performance} and using the basic D\&C tree, $ \T = 1 $ of the table \ref{Tbl Summary of Tree Structures} (three nodes and height one), to check those presented in \textsc{Section} \ref{Sec Probabilistic Estimates of the Divide-and-Conquer Algorithm}. 
	
\end{remark}
%
%
%
%
\section{Conclusions and Final Discussion}\label{Sec Conclusions and Final Discussion}
%
%
The present work yields the following conclusions
\begin{enumerate}[(i)]
	
	\item A complete and detailed theoretical analysis for the Divide-and-Conquer method's efficiency has been presented. The analysis has been done from two points of view: the worst case scenario and the expected performance. Before this work, the method's efficiency was analyzed only from the empirical point of view.
	
	\item For the worst case scenario, it suffices to find a control solution (see \textsc{Theorem} \ref{Thm quality certificate DCM}) which is computationally cheap. In our case furnished by the extended-greedy algorithm ($ \x^{\eg} $ and $ z^{\eg} $) and then split the problems: the restriction of this solution belongs to all the search spaces of the D\&C subproblems. This was done by carefully computing the knapsack capacities of the subproblems, given that the mechanism for splitting items (even and odd indexes) was already decided as discussed in \textsc{Remark} \ref{Rem Divide-and-Conquer tree}.
	
	\item It is possible to use another control solution for the worst case scenario, rather than the one presented here. For instance, the algorithm $ G^{\frac{3}{4}} $ presented in \cite{kellerer2005knapsack}, which is computationally more expensive, but it certifies a worst case scenario of 75\%. However, for this or any other control solution, the computation of the knapsack capacities $ \delta_{\lt}, \delta_{\rt} $ detailed in \textsc{Algorithm} \ref{Alg Branch Function}, needs to be adjutsed in order to satisfy the hypothesis of \textsc{Theorem} \ref{Thm quality certificate DCM}.
	
	\item The analysis of the expected performance is considerably harder than the previous one. A discrete probabilistic setting has to be established (see \textsc{Hypothesis} \ref{Hyp Random Problems}) and a randomized version of the problem, 0-1RKP, has to be introduced (see \textsc{Definition} \ref{Def Random Problems}). 
	
	\item The probabilistic analysis was done in two parts: \textsc{Section} \ref{Sec The Expected Performance} analyzes the 0-1RKP in full, while \textsc{Section} \ref{Sec Probabilistic Estimates of the Divide-and-Conquer Algorithm} analyzes the expected behavior of one single iteration of the Divide-and-Conquer method. In the first case, all the expectations were computed with absolute accuracy. In the second case, the same rigor was kept only for the computation of the left and right knapsack capacities but, in order to pursue further results, we approximated the expression assuming independence of the slack $ \upk $ and split $ \ups $ variables. The latter approximation has solid grounds because of the smooth behavior of the expectations of the main variables of the model, as shown in \textsc{Theorem} \ref{Thm Split Item Asymptotic}.
	
	\item In \textsc{Section} \ref{Sec Performance Estimates for the Divide-and-Conquer} several parameters to measure the performance of the method were introduced. Here, the expressions previously attained were numerically evaluated (due to its complexity) in order to obtain concrete, tangible values of the method's performance; first for one single D\&C iteration, then, an approximation is given for a general D\&C tree (see \textsc{Definition} \ref{Def Labelling of Trees}). Once again, hypothesis of independence between random variables were adopted, in order to compute the desired values (see \textsc{Remark} \ref{Rem Approximation Efficiency and Bounds}). Finally, the theoretical results are verified empirically with numerical experiments statistically sound.
	
	\item The empirical verification of our results (displayed in \textsc{Table} \ref{Tbl Experiments 64 Elements}), show that the theoretical approximations (summarized in the table \ref{Tbl Theoretical Tree Efficiencies}) are a lower estimate for the performance of Divide-and-Conquer and can be used to evaluate the method in general terms. To this end, two pairs of parameters were introduced: $ \rho^{\ef}, \rho^{\lp} $ as a \emph{reference of} and $ \lb^{\gr}, \lb^{\ef} $ as \emph{lower bounds of} the \emph{expected performance}. Hence, if the first pair of parameters is used to decide, it is recomendable to use the method with at most three iterations ($ \T = 3 $). However, a more conservative approach using the lower bounds' pair, states that Divide-and-Conquer should be used with at most two iterations ($ \T = 2 $), because beyond that height they are no better than those of the worst case scenario, already furnished by the extended-greedy algorithm ($ \x^{\eg} $ and $ z^{\eg} $).    
	
	\item Finally, a more daring approach would use the empirical evidence to decide the limit extension of Divide-and-Conquer trees, summarized in \textsc{Table} \ref{Tbl Experiments 64 Elements} (similar to all the other experiments of \textsc{Table} \ref{Tbl Summary of Experiments and Bernoulli Trials}). From this point of view, the method is still highly recommendable for four iterations ($ \T = 4 $). This is consistent with the empirical findings of \cite{MoralesMartinez}, where six D\&C iterations produced satisfactory results in \emph{average}.  
	
\end{enumerate}
%
%
%
%
%
%
%
%
\noindent\textbf{Acknowledgements}\\
\textit{\noindent The first Author wishes to thank Universidad Nacional de Colombia, Sede Medell\'in for supporting the production of this work through the project Hermes 54748 as well as granting access to Gauss Server, financed by ``Proyecto Plan 150x150 Fomento de la cultura de evaluaci\'on continua a trav\'es del apoyo a planes de mejoramiento de los programas curriculares". (\url{gauss.medellin.unal.edu.co}), where the numerical experiments were executed. All the polynomial identities used in the proofs were verified in \url{wolframalpha.com}; due to its level of complexity it would have not been possible to develop them without this remarkable free tool.}


\end{document}